\documentclass[a4paper,USenglish,pdfa,cleveref,autoref,numberwithinsect,thm-restate]{lipics-v2021-patched}

\usepackage[T5,T1]{fontenc} 

\title{Comparison-free Polyregular Functions} 

\author{{\fontencoding{T5}\selectfont Lê Thành Dũng (Tito) Nguyễn}}{Laboratoire d'informatique de Paris Nord,
  Villetaneuse, France \and \url{https://nguyentito.eu/}}{nltd@nguyentito.eu}{https://orcid.org/0000-0002-6900-5577}{}

\author{Camille Noûs}{Laboratoire Cogitamus \and
  \url{https://www.cogitamus.fr/camilleen.html}}{}{}{}

\author{Cécilia Pradic}{Department of Computer Science, University of Oxford,
  United Kingdom}{}{https://orcid.org/0000-0002-1600-8846}{}

\authorrunning{L.~T.~D.~{\fontencoding{T5}\selectfont{}Nguyễn}, C.~Noûs and C.~Pradic}

\Copyright{{\fontencoding{T5}\selectfont Lê Thành Dũng Nguyễn} and Cécilia Pradic}



\relatedversiondetails{In ICALP'21 proceedings}{https://doi.org/10.4230/LIPIcs.ICALP.2021.139}

\supplement{}

\acknowledgements{Thanks to Mikołaj Bojańczyk and Sandra Kiefer for inspiring
  discussions, to Gaëtan Douéneau-Tabot and Amina Doumane for explaining some
  features of their work to us, to Charles Paperman for his help with bibliography and to the reviewers for their feedback.}

\nolinenumbers 

\hideLIPIcs  

\EventEditors{Nikhil Bansal, Emanuela Merelli, and James Worrell}
\EventNoEds{3}
\EventLongTitle{48th International Colloquium on Automata, Languages, and Programming (ICALP 2021)}
\EventShortTitle{ICALP 2021}
\EventAcronym{ICALP}
\EventYear{2021}
\EventDate{July 12--16, 2021}
\EventLocation{Glasgow, Scotland (Virtual Conference)}
\EventLogo{}
\SeriesVolume{198}
\ArticleNo{132}

\usepackage{bussproofs}
\usepackage{csquotes}
\usepackage{mdframed}
\usepackage{mathpartir}

\usepackage{stmaryrd}
\usepackage{cmll}

\usepackage{tikz}

\newcommand{\rk}{\mathrm{rk}}
\newcommand{\powerset}{\mathcal{P}}
\newcommand{\naturalN}{\mathbb{N}}
\newcommand{\bbQ}{\mathbb{Q}}
\newcommand{\bbZ}{\mathbb{Z}}

\newcommand{\CbS}{\mathrm{CbS}}
\newcommand{\fmap}{\mathbf{map}}

\newcommand{\cT}{\mathcal{T}}
\newcommand{\cC}{\mathcal{C}}
\newcommand{\card}[1]{|#1|}
\newcommand{\len}[1]{\left|#1\right|}

\newcommand{\Hom}{\mathrm{Hom}}
\newcommand{\rmst}{\mathrm{st}}
\newcommand{\rmreg}{\mathrm{reg}}

\newcommand{\deltast}{\delta_\rmst}
\newcommand{\deltareg}{\delta_\rmreg}

\newcommand{\ttreverse}{\mathtt{reverse}}
\newcommand{\tterase}{\mathtt{erase}}
\newcommand{\cfpow}{\mathtt{cfpow}}
\newcommand{\cfsquaring}{\mathtt{cfsquaring}}
\newcommand{\transmono}{\mathcal{M}}
\newcommand{\transmonocl}{\mathcal{M}^{\mathrm{cl}}}
\newcommand{\transmonozo}{\mathcal{M}^{\mathrm{cl01}}}
\newcommand{\prodmono}{\mathcal{N}}
\newcommand{\prodmonocf}{\mathcal{N}'}

\newcommand\leftmarker{\triangleright}
\newcommand\rightmarker{\triangleleft}
\newcommand\markers[1]{\leftmarker #1 \rightmarker}
\newcommand\interp[1]{\llbracket #1 \rrbracket}
\newcommand\Stack[1]{\mathrm{Stack}_{#1}}
\newcommand\true{\mathtt{true}}
\newcommand\false{\mathtt{false}}
\newcommand\noincl{\mathrel{\supset\hspace{-0.5em}\subset}}
\newcommand\longto\longrightarrow
\newcommand\bnfalt{\; | \;}

\newcommand\Succ{{\sf S}}

\theoremstyle{claimstyle}
\newtheorem{problem}[theorem]{Problem}

\usepackage{todonotes}

\begin{document}

\maketitle

\begin{abstract}
  This paper introduces a new automata-theoretic class of string-to-string
  functions with polynomial growth. Several equivalent definitions are
  provided: a machine model which is a restricted variant of pebble transducers,
  and a few inductive definitions that close the class of regular functions
  under certain operations. Our motivation for studying this class comes from
  another characterization, which we merely mention here but prove elsewhere,
  based on a $\lambda$-calculus with a linear type system.

  As their name suggests, these \emph{comparison-free polyregular functions}
  form a subclass of polyregular functions; we prove that the inclusion is
  strict. We also show that they are incomparable with HDT0L transductions,
  closed under usual function composition -- but not under a certain
  \enquote{map} combinator -- and satisfy a comparison-free version of the
  pebble minimization theorem.

  On the broader topic of polynomial growth transductions, we also consider the
  recently introduced layered streaming string transducers (SSTs), or
  equivalently $k$-marble transducers. We prove that a function can be obtained
  by composing such transducers together if and only if it is polyregular, and
  that $k$-layered SSTs (or $k$-marble transducers) are closed under
  \enquote{map} and equivalent to a corresponding notion of $(k+1)$-layered
  HDT0L systems.
\end{abstract}

\section*{Addendum (2023)}

Our proof of the comparison-free pebble minimization theorem
(\Cref{thm:characterization-rank}) is heavily based on a 2020
paper~\cite{PebbleMinimization} that claimed to show pebble minimization for
general polyregular functions. While we could reuse many sound and useful ideas
from that paper, that central claim turned out to be wrong, as shown
in~\cite{PolyregularGrowth,PolyregularGrowthBis}. However, we are confident that
\Cref{thm:characterization-rank} is still valid; it has even been reproved and
generalized to a larger subclass of pebble transducers using different
techniques~\cite{LastPebble} (subsequent papers such
as~\cite{LastPebble,PolyregularGrowthBis} refer to the class of functions
introduced here by the shorter name \enquote{polyblind}).
Some typos have also been fixed after publication, thanks to the reviewers of
the first author's PhD thesis.

\section{Introduction}

The theory of transducers (as described in the
surveys~\cite{siglog,MuschollPuppis}) has traditionally dealt with devices that
take as input strings of length $n$ and output strings of length $O(n)$.
However, several recent works have investigated function classes going beyond
linear growth. We review three classes in this landscape below.
\begin{itemize}
\item \emph{Polyregular functions} (\S\ref{sec:polynomial-growth}) are thus named because they have (at most)
  polynomial growth and include regular functions (\S\ref{sec:reg}) (the most expressive of the
  traditional string-to-string transduction classes). They were defined in
  2018~\cite{polyregular} by four equivalent computational models, one of which
  -- the \emph{pebble transducers} -- is the specialization to strings of a tree
  transducer model that existed previously in the literature~\cite{Pebble}
  (this specialization had been investigated earlier
  in~\cite{PebbleString,PebbleComposition}). A
  subsequent work~\cite{polyregularMSO} gave a logical
  characterization based on Monadic Second-Order logic (MSO). They enjoy two nice properties:
  \begin{itemize}
  \item \emph{preservation of regular languages (by preimage)}: if $f : \Gamma^*
    \to \Sigma^*$ is polyregular and $L \subseteq \Sigma^*$ is regular, then
    $f^{-1}(L) \subseteq \Gamma^*$ is regular;
  \item \emph{closure under function composition}: if $f : \Gamma^* \to
    \Delta^*$ and $g : \Delta^* \to \Sigma^*$ are both polyregular, then so is
    $g \circ f : \Gamma^* \to \Sigma^*$.
  \end{itemize}
\item \emph{HDT0L transductions} (\S\ref{subsec:hdt0l}) form another superclass of regular functions,
  whose output size may be at most exponential in the input size. They are older
  than polyregular functions, and we shall discuss their history in
  \Cref{subsec:hdt0l}; suffice to say for now, they also admit various
  equivalent characterizations scattered in several
  papers~\cite{FerteMarinSenizergues,FiliotReynier,Marble}. These functions
  preserve regular languages by preimage, but are \emph{not} closed under
  composition (the growth rate of a composition of HDT0L transductions may be a
  tower of exponentials).
\item Very recently, the polynomially bounded HDT0L transductions (\S\ref{sec:polynomial-growth}) have been
  characterized using two transducer models~\cite{Marble}. One of them, the
  \emph{$k$-marble transducers} (where $k \in \naturalN$ depends on the function
  to be computed), is obtained by putting a syntactic constraint on the model of \emph{(unbounded) marble transducers}~\cite{Marble} which computes HDT0L transductions. But it can also be seen as
  a restricted variant of pebble transducers; it follows
  (although this is not explicitly stated in~\cite{Marble}) that a HDT0L
  transduction has polynomial growth if and only if it is polyregular. Moreover,
  as claimed in~\cite[Section~6]{Marble}, the functions computed by $k$-marble
  transducers are not closed under composition either, and thus form a strict
  subclass of polyregular functions.
\end{itemize}

\subparagraph{A new subclass of polyregular functions}

In this paper, we start by proving a few results on the above classes
(\Cref{sec:complements}). For instance, we supply a proof for the aforementioned
claim of~\cite[Section~6]{Marble}, and show that the polyregular functions are
exactly those computable by compositions of $k$-marble transducers. Those
complements are not particularly difficult nor surprising and are included
mostly for the sake of giving a complete picture.

But our main contribution is the introduction of a new class, giving its title
to the paper; as we show, it admits three equivalent definitions:
\begin{itemize}
  \item two ways to inductively generate the class (Sections~\ref{sec:cbs}
        and~\ref{sec:composition} respectively):
  \begin{itemize}
  \item by closing regular functions under a certain \enquote{composition by
      substitution} operation;
  \item by combining regular functions and a certain kind of \emph{squaring}
    functions (less powerful than the squaring plus underlining functions used
    to characterize general poyregular functions) with usual function
    composition;
  \end{itemize}
  \item a restriction on pebble transducers (\Cref{sec:cfpebble}) -- we
        disallow comparing the positions of a transducer's multiple reading
        heads, hence the name \emph{comparison-free polyregular functions}
        (henceforth abbreviated as \emph{cfp}).
\end{itemize}



\subparagraph{Properties}

By the third definition above, comparison-free polyregular functions are indeed
polyregular, while the second one implies that our new class contains the
regular functions and is closed under composition. (In fact, in the proof that
our first definition is equivalent to the second one, most of the work goes into
showing that the former enjoys closure under composition.) We rule out
inclusions involving the other classes that we mentioned by proving some
\emph{separation results} (\Cref{sec:separation}): there exist
\begin{itemize}
\item comparison-free polyregular functions that are not HDT0L (we take one
  example from~\cite{Marble}),
\item and polynomially bounded HDT0L transductions which are not comparison-free:
  \begin{itemize}
    \item one of our examples follows from a precise characterization of cfp
          functions over unary input alphabets (extending a known result for
          regular functions with unary inputs~\cite{Eugenia}), which we give in
          \Cref{sec:sequences};
  \item another example shows that unlike (poly)regular functions, cfp functions
    are \emph{not} closed under a certain counterpart of the \enquote{map}
    operation in functional programming.
  \end{itemize}
\end{itemize}
We summarize the inclusions and separations between classes that we get in
\Cref{fig:inclusions}.
\begin{figure}
\begin{center}
\includegraphics[scale=0.75]{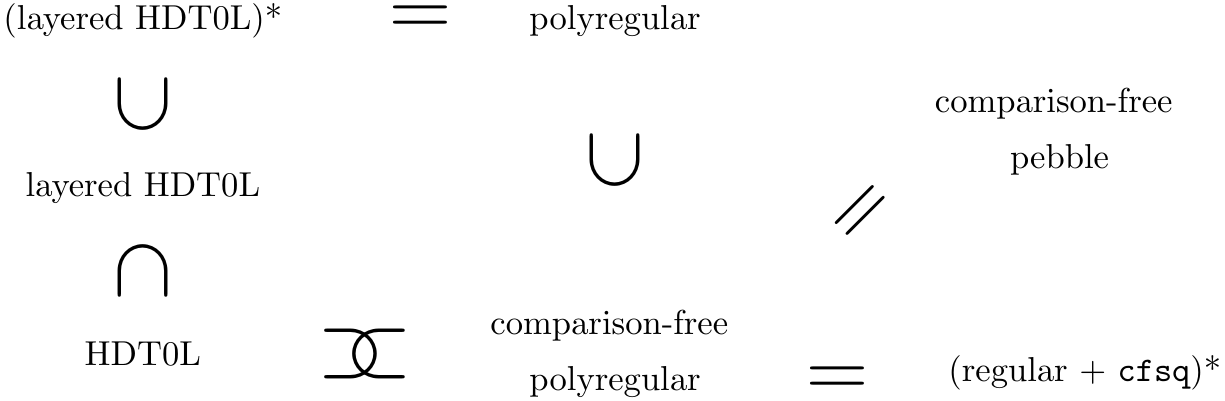}
\end{center}
\caption{Summary of the known relationships between superlinear transduction
  classes, taking our results into account. Inclusions $\subset$ are strict, and
  $\noincl$ means that there is no inclusion either way. Finally $C^*$ denotes
  the composition closure of the class $C$.}
\label{fig:inclusions}
\end{figure}

Finally, we show in \Cref{sec:minimization} that the number of pebbles required
to compute a function using a comparison-free transducer is related to its
growth rate. The analogous result for pebble transducers was proved recently,
with a whole paper dedicated to it~\cite{PebbleMinimization}; we adapt its
arguments to our setting, resulting in our longest and most technical proof.
There is a similar property for $k$-marble transducers~\cite{Marble}, but it is
proved using very different tools.

\subparagraph{Motivations}

Although this is the first proper paper to introduce comparison-free pebble
transducers, we were told that they had already been considered by several
colleagues (Miko{\l}aj Boja{\'n}czyk, personal communication). But in fact, the
starting point in our investigation was a characterization of regular functions
using a linear $\lambda$-calculus (in the sense of linear logic) that we had
previously obtained~\cite{freeadditives}; this was part of a research programme
relating automata and functional programming that we initiated
in~\cite{aperiodic}. As we reported in a previous version of the present paper,
by tweaking a parameter in this characterization, one gets the cfp functions
instead; we initially defined the latter using composition by substitution, and
only later realized the connection with pebble transducers. One interesting
feature of the $\lambda$-calculus characterization is that it is trivially
closed under composition, and this led us to take inspiration from the
category-theoretic machinery that we used in~\cite{freeadditives} for our
standalone composition proof in this paper.

\emph{Added in 2023: an \enquote{official} reference for this $\lambda$-calculus
  characterization can now be found in the first author's PhD
  thesis~\cite[Theorem~1.2.3]{titoPhD}.}




\section{Preliminaries}
\label{sec:preliminaries}

\subparagraph{Notations}

The set of natural numbers is $\naturalN = \{0,1,\dots\}$. We write $|w|$ for
the length of a string $w \in \Sigma^*$; for $\Pi \subseteq \Sigma$, we write
$|s|_\Pi$ for the number of occurrences of letters from $\Pi$ in $w$; and for
$c \in \Sigma$, we abbreviate $|w|_{\{c\}}$ as $|w|_{c}$. The $i$-th letter of
$w \in \Sigma^{*}$ is denoted by either $w_i$ or $w[i]$ (for
$i \in \{1,\ldots,|w|\}$). Given monoids $M$ and $N$, $\Hom(M,N)$ is the set of
monoid morphisms. We write $\varepsilon$ for the empty word and
$\underline\Sigma = \{\underline{a} \mid a \in \Sigma\}$ for a disjoint copy of
the alphabet~$\Sigma$ made of \enquote{underlined} letters.

\subsection{HDT0L transductions and streaming string transducers}
\label{subsec:hdt0l}

\emph{L-systems} were originally introduced by Lindenmayer~\cite{Lindenmayer} in
the 1960s as a way to generate formal languages, with motivations from biology.
While this language-centric view is still predominant, the idea of considering
variants of L-systems as specifications for string-to-string functions -- whose
range are the corresponding languages -- seems to be old. For instance, in a
paper from 1980~\cite{ERS}, one can find (multi-valued) string functions defined
by ET0L systems.

More recently, Ferté, Marin and Sénizergues~\cite{FerteMarinSenizergues}
provided alternative characterizations\footnote{Those characterizations had
  previously been announced in an invited paper by
  Sénizergues~\cite{Senizergues}. Some other results announced
  in~\cite{Senizergues} are proved in~\cite{PolyRecSeq}.} (by catenative
recurrent equations and higher-order pushdown transducers of level 2) of the
string-to-string functions that \emph{HDT0L systems} can express -- what we call
here \emph{HDT0L transductions}.
Later work by Filiot and Reynier~\cite{FiliotReynier} and then by
Douéneau-Tabot, Filiot and Gastin~\cite{Marble} -- that does not build
on~\cite{Senizergues,FerteMarinSenizergues} -- proved the equivalence with,
respectively, copyful SSTs (\Cref{def:sst}) and unbounded marble transducers
(not presented here).

\begin{definition}[following~{\cite{FiliotReynier}}]
  \label{def:hdt0l}
  A \emph{HDT0L system} consists of:
  \begin{itemize}
  \item an input alphabet $\Gamma$, an output alphabet $\Sigma$, and a working
    alphabet $\Delta$ (all \emph{finite});
  \item an initial word $d \in \Delta^*$;
  \item for each $c \in \Gamma$, a \emph{monoid morphism} $h_c \in \Hom(\Delta^*,
    \Delta^*)$;
  \item a final morphism $h' \in \Hom(\Delta^*, \Sigma^*)$.
  \end{itemize}
  It defines the transduction taking $w = w_1 \ldots w_n \in \Gamma^*$ to $h'
  \circ h_{w_1} \circ \ldots \circ h_{w_n}(d) \in \Sigma^*$.
\end{definition}

(The definition of HDT0L systems given
in~\cite{Senizergues,FerteMarinSenizergues} makes slightly different choices of
presentation\footnote{The family $(h_c)_{c \in \Gamma}$ is presented as a
  morphism $H : \Gamma^* \to \Hom(\Delta^*,\Delta^*)$ (whose codomain is indeed
  a monoid for function composition). And an initial \emph{letter} is used
  instead of an initial word; this is of no consequence regarding the functions
  that can be expressed (proof sketch: consider $\Delta' = \Delta \cup \{x\}$
  with a new letter $x \notin \Delta$, take $x$ as the initial letter and let
  $h_c(x) = h_c(w)$, $h'(x) = h'(w)$).}.) To define the equivalent model of
copyful streaming string transducers, we must first introduce the notion of
register assignment.

\begin{definition}
  \label{def:register-assignment}
  Fix a finite alphabet $\Sigma$. Let $R$ and $S$ be two finite sets
  \emph{disjoint from $\Sigma$}; we shall consider their elements to be
  \enquote{register variables}.

  For any word $\omega \in (\Sigma \cup R)^*$, we write
  $\omega^\dagger : (\Sigma^*)^R \to \Sigma^*$ for the map that sends
  $(u_r)_{r \in R}$ to $\omega$ in which every occurrence of a register variable
  $r \in R$ is replaced by $u_r$ -- formally, we apply to $\omega$ the morphism
  $(\Sigma \cup R)^* \to \Sigma^*$ that maps $c \in \Sigma$ to itself and
  $r \in R$ to $u_r$.

  A \emph{register assignment}\footnote{Some papers
    e.g.~\cite{AperiodicSST,Marble} call register assignments
    \emph{substitutions}. We avoid this name since it differs from its meaning
    in the context of our \enquote{composition by substitution} operation.
    } $\alpha$ from $R$ to $S$ (over $\Sigma$) is a map
  $\alpha : S \to (\Sigma \cup R)^*$. It induces the action $\alpha^\dagger :
  \vec{u} \in (\Sigma^*)^R \mapsto (\alpha(s)^\dagger(\vec{u}))_{s \in S} \in
  (\Sigma^*)^S$ (which indeed goes \enquote{from $R$ to $S$}).

\end{definition}

\begin{definition}[\cite{FiliotReynier}]
  \label{def:sst}
  A (deterministic copyful) \emph{streaming string transducer} (SST) with input
  alphabet $\Gamma$ and output alphabet $\Sigma$ is a tuple $\cT = (Q, q_0, R,
  \delta, \vec{u}_I, F)$ where
\begin{itemize}
\item $Q$ is a finite set of \emph{states} and $q_0 \in Q$ is the \emph{initial
    state};
\item $R$ is a finite set of \emph{register variables}, that we require to be
  \emph{disjoint from $\Sigma$};
\item $\delta : Q \times \Gamma \to Q \times (R \to (\Sigma \cup R)^*)$ is the
  \emph{transition function} -- we abbreviate $\deltast = \pi_1 \circ \delta$
  and $\deltareg = \pi_2 \circ \delta$, where $\pi_i$ is the projection from
  $X_1 \times X_2$ to its $i$-th component $X_i$;
\item $\vec{u}_I \in (\Sigma^*)^R$ describes the \emph{initial register values};
\item $F : Q \to (\Sigma \cup R)^*$ describes how to recombine the final values
  of the registers, depending on the final state, to produce the output.
\end{itemize}
The function $\Gamma^* \to \Sigma^*$ computed by $\cT$ is
\[ w_1 \ldots w_n ~~ \mapsto ~~ F(q_n)^\dagger \circ \deltareg(q_{n-1},
  w_n)^\dagger \circ \ldots \circ \deltareg(q_0, w_1)^\dagger(\vec{u}_I) \]
where the sequence of states $(q_i)_{0 \leq i \leq n}$ (sometimes called the
\emph{run} of the transducer over the input word) is inductively defined,
starting from the fixed initial state $q_0$, by $q_i = \deltast(q_{i-1},w_i)$.
\end{definition}

\begin{example}
  \label{ex:prefixes}
  Let $\Sigma = \Gamma \cup \underline{\Gamma}$. We consider a SST $\cT$
  with $Q = \{q\}$, $R = \{X,Y\}$ and
  \[ \vec{u}_I = (\varepsilon)_{r \in R} \qquad  F(q) = Y \qquad
    \forall c \in \Gamma,\; \delta(q,c) = (q,\, (X \mapsto cX,\,
    Y \mapsto \underline{c}XY))\]
  If we write $(v,w)$ for the family $(u_r)_{r \in R}$ with $u_X = v$ and $u_Y =
  w$, then the action of the register assignments may be described as
  $(X \mapsto cX,\, Y \mapsto \underline{c}XY)^\dagger(v,w) = (c\cdot v,\,
    \underline{c} \cdot v \cdot w)$.

    Let $\mathtt{1},\mathtt{2},\mathtt{3},\mathtt{4} \in \Gamma$. After reading
    $\mathtt{1234} \in \Gamma^*$, the values stored in the registers of $\cT$ are
  \[ (X \mapsto \mathtt{4}X,\, Y \mapsto \underline{\mathtt{4}}XY)^\dagger \circ
    \ldots \circ (X \mapsto \mathtt{1}X,\, Y \mapsto
    \underline{\mathtt{1}}XY)^\dagger(\varepsilon,\varepsilon) = (\mathtt{4321},
    \mathtt{\underline{4}321\underline{3}21\underline{2}1\underline{1}}) \]
  Since $F(q) = Y$, the function defined by $\cT$ maps $\mathtt{1234}$ to
  $\mathtt{\underline{4}321\underline{3}21\underline{2}1\underline{1}} \in
  (\Gamma \cup \underline\Gamma)^* = \Sigma^*$.
\end{example}

This gives us an example of HDT0L transduction $\Gamma^* \to (\Gamma \cup
\underline\Gamma)^*$, since:

\begin{theorem}[{\cite{FiliotReynier}}]
  \label{thm:equiv-hdt0l}
  A function $\Gamma^* \to \Sigma^*$ can be computed by a copyful SST if and
  only if it can be specified by a HDT0L system.
\end{theorem}

\begin{remark}
  \label{rem:hdt0l-single-state}
  As observed in~\cite[Lemma~3.3]{FiliotReynier}, there is a natural translation
  from HDT0L systems to SSTs whose range is composed precisely of the
  \emph{single-state} SSTs whose transitions and final output function \emph{do
    not access the letters of their output alphabet} -- those are called
  \emph{simple SSTs} in~\cite[\S5.1]{Marble}. This involves a kind of reversal:
  the initial register values correspond to the final morphisms, while the final
  output function corresponds to the initial word. Thus, \Cref{thm:equiv-hdt0l}
  is essentially a state elimination result; a direct translation from SSTs to
  single-state SSTs has also been given by Benedikt et
  al.~\cite[Proposition~8]{PolynomialAutomata}. However, it does \emph{not}
  preserve the subclass of \emph{copyless} SSTs (this would contradict
  \Cref{prop:single-state-copyless-weak}).

  The lookahead elimination theorem for macro tree
  transducers~\cite[Theorem~4.21]{Macro} arguably generalizes this to trees.
  Indeed, while those transducers are generally presented as a top-down model,
  their formal definition can also be read as bottom-up register tree
  transducers in the style of~\cite[\S4]{FOTree}, and top-down lookahead
  corresponds to bottom-up states.
\end{remark}

\subsection{Regular functions}
\label{sec:reg}

\begin{definition}[Alur and {\v{C}}ern{\'{y}}~{\cite{SST}}]
  A register assignment $\alpha : S \to (\Sigma \cup R)^*$ from $R$ to $S$ is
  said to be \emph{copyless} when each $r \in R$ \emph{occurs at most once among
    all the strings $\alpha(s)$ for $s \in S$}, i.e.\ it does not occur at least
  twice in some $\alpha(s)$, nor at least once in $\alpha(s)$ and at least once
  in $\alpha(s')$ for some $s \neq s'$. (This restriction does not apply to the
  letters in $\Sigma$.)
  
  A streaming string transducer is \emph{copyless} if all the assignments in the
  image of its transition function are copyless. In this paper, we take
  computability by copyless SSTs as the definition of \emph{regular functions}
  (but see \Cref{thm:reg-2dft} for another standard definition).
\end{definition}
\begin{remark}
  Thanks to \Cref{thm:equiv-hdt0l}, every regular function is a HDT0L
  transduction.
\end{remark}
\begin{remark}
  \label{rem:prefixes-quadratic}
The SST of \Cref{ex:prefixes} is \emph{not} copyless: in a transition $\alpha =
\deltareg(q,c)$, the register $X$ appears twice, once in $\alpha(X) = cX$
and once in $\alpha(Y) = \underline{c}XY$; in other words, its value is
\emph{duplicated} by the action $\alpha^\dagger$. In fact, it computes a
function whose output size is quadratic in the input size, while regular
functions have linearly bounded output.
\end{remark}
\begin{example}[Iterated reverse~{\cite[p.~1]{polyregular}}]
  \label{ex:iterated-reverse}
  The following single-state SST is copyless:
  \[ \Gamma = \Sigma\ \text{with}\ \# \in \Sigma \qquad Q = \{q\} \qquad R =
    \{X,Y\} \qquad \vec{u}_I = (\varepsilon)_{r \in R} \qquad F(q) = XY\]
  \[ \delta(q,\#) = (q,\, (X \mapsto XY\#,\, Y \mapsto \varepsilon)) \qquad
    \forall c \in \Sigma \setminus \{\#\},\; \delta(q,c) = (q,\, (X \mapsto X,\,
    Y \mapsto cY)) \]
  For $u_1,\ldots,u_n \in (\Sigma \setminus \{\#\})^*$, it maps
  $u_1\#\ldots\#u_n$ to $\ttreverse(u_1)\#\ldots\#\ttreverse(u_n)$.
\end{example}

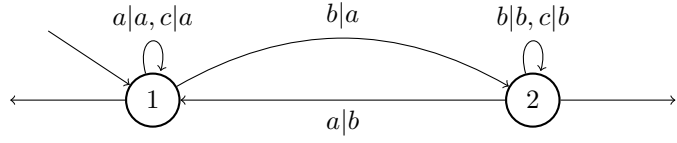
\begin{figure}
  \centering
    \begin{tikzpicture}
    \node (init) at (-0.5,2) {};
    \node (outa) at (-1,1) {};
    \node (outb) at (8,1) {};
    \node[circle,draw,thick,minimum size=20pt] (qa) at (1,1) {$1$};
    \node[circle,draw,thick,minimum size=20pt] (qb) at (6,1) {$2$};
    
    \draw[->] (init) -- (qa);
    \draw[->] (qa) edge [loop above] node {$a|a, c|a$} ();
    \draw[->] (qa) edge  [bend left] node[above] {$b|a$} (qb);
    \draw[->] (qb) edge [loop above] node {$b|b, c|b$} ();
    \draw[->] (qb) edge node[below] {$a|b$} (qa);
    \draw[->] (qa) edge node[below] {} (outa);
    \draw[->] (qb) edge node[below] {} (outb);
  \end{tikzpicture}
  \caption{An example of sequential transducer.}
  \label{fig:mealy}
\end{figure}

The concrete SSTs (copyless or not) that we have seen for now are all
single-state. As a source of stateful copyless SSTs, one can consider the
translations of \emph{sequential transducers}. These are usual finite automata,
whose transitions additionally produce a word catenated to the end of
the would-be output function. For instance, the one in \Cref{fig:mealy}
computes the function $\{a,b,c\}^* \to \{a,b\}^*$ that replaces each $c$ in its
input by the closest non-$c$ letter on its left (or $a$ if no such letter
exists). We do not give a detailed definition
(which can be found e.g. in~\cite[Chapter~V]{Sakarovitch})
here, but for our purpose, it suffices to observe any sequential transducer can
be translated into a copyless SST with the same set of states and a single register.

\subsection{Polynomial growth transductions}
\label{sec:polynomial-growth}

Next, we recall one way to define Bojańczyk's \emph{polyregular
  functions}~\cite{polyregular}.
\begin{definition}[\cite{polyregular}]
  \label{def:polyreg}
  The class of \emph{polyregular functions} is the smallest class of
  string-to-string functions closed under composition containing:
  \begin{itemize}
  \item the functions computed by sequential transducers (for instance, the one
    of \Cref{fig:mealy});
  \item the \emph{iterated reverse} function of \Cref{ex:iterated-reverse}, over any
    finite alphabet containing $\#$;
  \item the \emph{squaring with underlining} functions $\mathtt{squaring}_\Gamma
    : \Gamma^* \to (\Gamma \cup \underline\Gamma)^*$, for any finite~$\Gamma$,
    illustrated by $\mathtt{squaring}_\Gamma(\mathtt{1234}) =
    \mathtt{\underline12341\underline23412\underline34123\underline4}$.
  \end{itemize}
\end{definition}

As mentioned in the introduction, the intersection between the above class and
HDT0L transductions has been recently characterized by Douéneau-Tabot et
al.~\cite{Marble}.
\begin{theorem}[{\cite{Marble}}]
  Let $f : \Gamma^* \to \Sigma^*$. The following conditions are equivalent:
  \begin{itemize}
  \item $f$ is both a polyregular function and a HDT0L transduction;
  \item $f$ is a HDT0L transduction and has at most polynomial growth: $|f(w)| =
    |w|^{O(1)}$;
  \item there exists $k \in \naturalN$ such that $f$ is computed by some
    \emph{$k$-layered SST}, defined below.
  \end{itemize}
\end{theorem}
(Another equivalent model, the \emph{$k$-marble transducers}, was mentioned in
the introduction, but we will not use it in the rest of the paper.) Those
$k$-layered SST propose a compromise between copyful and copyless SSTs:
duplication is controlled, but not outright forbidden.
\begin{definition}[\cite{Marble}]
  A register assignment $\alpha : R \to (\Sigma \cup R)^*$ is \emph{$k$-layered}
  (for $k \in \naturalN$) with respect to a partition $R = R_0 \sqcup \ldots
  \sqcup R_k$ when for $0 \leq i \leq k$,
  \begin{itemize}
  \item for $r \in R_i$, we have $\alpha(r) \in (\Sigma \cup R_0 \cup \ldots
    \cup R_i)^*$;
  \item each register variable in $R_i$ appears at most once among all the
    $\alpha(r)$ for $r \in R_i$ (however, those from $R_0 \sqcup \ldots \sqcup
    R_{i-1}$ may appear an arbitrary number of times).
  \end{itemize}
  A SST is \emph{$k$-layered} if its registers can be partitioned in such a way
  that all assignments in the transitions of the SST are $k$-layered w.r.t.\
  that partition.
\end{definition}
Beware: with this definition, the registers of a $k$-layered SST are actually
divided into $k+1$ layers, not $k$. In particular, a SST is copyless if and only
if it is 0-layered. (We chose this convention for backwards compatibility
with~\cite{Marble}; see also \Cref{rem:numbering-convention}.)

For instance, the transducer of \Cref{ex:prefixes} is 1-layered with $R_0 =
\{X\}$ and $R_1 = \{Y\}$. There also exist register assignments that cannot be
made $k$-layered no matter the choice of partition, such as $X \mapsto XX$.
Using such assignments, one can indeed build SSTs that compute functions $f$
such that e.g.\ $|f(w)| = 2^{|w|}$.


\begin{remark}
  \label{rem:rational-series}
  There is arguably an old precursor to this recent characterization of HDT0L
  transductions with polynomial growth by a syntactic \enquote{layering}
  condition: Schützenberger's theorem on polynomially bounded $\bbZ$-rational
  series, which dates back to the 1960s (see for instance~\cite[Chapter~9,
  Section~2]{BerstelReutenauer} -- the preface of the same book describes this
  theorem as \enquote{one of the most difficult results in the area}). Let us
  give a brief exposition.

  A \emph{$\bbZ$-rational series} $f : \Sigma^{*} \to \bbZ$ is a function of the
  form $f : w \in \Sigma^{*} \mapsto X^{T} \cdot \Phi(w) \cdot Y$ where
  $X,Y \in \bbZ^{R}$ and $\Phi$ is a morphism from $\Sigma^{*}$ to the
  multiplicative monoid of $R$-indexed square matrices over $\bbZ$, where $R$ is
  a finite set. This data $(X,\Phi,Y)$ has a clear interpretation as a
  \enquote{simple SST} (cf.~\Cref{rem:hdt0l-single-state}) with register set
  $R$, whose register values are integers rather than strings. Schützenberger's
  theorem says that any $\bbZ$-rational series $f$ with polynomial growth (i.e.\
  $|f(w)| = |w|^{O(1)}$ where $|\,\cdot\,|$ on the left is the absolute value)
  can be written as $f : w \mapsto X^{T} \cdot \Phi(w) \cdot Y$ where
  \begin{enumerate}[(i)]
    \item the image of $\Phi$ has a block triangular structure;
    \item the projection of this image on each diagonal block is a finite
          monoid.
  \end{enumerate}
  The first item gives us a partition of the register into layers where each
  layer \enquote{depends} only on the ones below them. The finiteness condition
  in the second item is equivalent to having bounded coefficients, which means
  that the register assignments within each layer are \emph{bounded-copy}, while
  in a layered SST, they would be \emph{copyless} instead -- but bounded-copy
  SSTs are known to be equivalent to copyless SSTs (see
  e.g.~\cite{AperiodicSST}). The theorem also states a relationship between the
  number of blocks and the growth rate; compare this to
  \Cref{rem:marble-minimization}.

  Via the canonical isomorphism $\{a\}^{*} \cong \naturalN$, HDT0L transductions
  with unary output alphabet are the same thing as \emph{$\naturalN$-rational
    series}. The counterpart of Schützenberger's theorem over $\naturalN$ is
  thus a corollary of the results of~\cite{Marble} on layered SSTs.
\end{remark}

\subsection{Transition monoids for streaming string transducers}
\label{sec:transition-monoids}

To wrap up the preliminaries, let us recall some algebraic tools for working
with SSTs (this technical section can be safely skipped on a first reading). Let
us start by putting a monoid structure on register assignments
(\Cref{def:register-assignment}).
\begin{definition}
  Let $\transmono_{R,\Sigma} = R \to (\Sigma \cup R)^*$ for $R \cap \Sigma =
  \varnothing$. We endow it with the following composition operation, that makes
  it into a monoid:
  \[ \alpha \bullet \beta = \alpha^\odot \circ \beta\quad\text{where}\
    \alpha^\odot \in \Hom((\Sigma \cup R)^*,(\Sigma \cup R)^*),\;
    \alpha^\odot(x) = \begin{cases}
      \alpha(x) & \text{for}\ x \in R\\
      x & \text{for}\ x \in \Sigma
    \end{cases}
  \]
\end{definition}
The monoid $\transmono_{R,\Sigma}$ thus defined is isomorphic to a submonoid of
$\Hom((\Sigma \cup R)^*,(\Sigma \cup R)^*)$ with function composition. It admits
a submonoid of \emph{copyless} assignments.
\begin{definition}
  \label{def:transmonocl}
  We write $\transmonocl_{R,\Sigma}$ for the set of all $\alpha \in
  \transmono_{R,\Sigma}$ such that each letter $r \in R$ occurs at most once
  among all the $\alpha(r')$ for $r' \in R$.
\end{definition}
\begin{proposition}
  $\transmonocl_{R,\Sigma}$ is a submonoid of $\transmono_{R,\Sigma}$. In other
  words, copylessness is preserved by composition (and the identity assignment
  is copyless).
\end{proposition}
The following proposition ensures that this composition does what we expect.
Recall from \Cref{def:register-assignment} that $(-)^\dagger$ sends
$\transmono_{R,\Sigma}$ to $(\Sigma^*)^R \to (\Sigma^*)^R$.
\begin{proposition}
  For all $\alpha,\beta \in \transmono_{R,\Sigma}$, we have $(\alpha \bullet
  \beta)^\dagger = \beta^\dagger \circ \alpha^\dagger$.
\end{proposition}

To incorporate information concerning the states of an SST, we define below a
special case of the \emph{wreath product} of transformation monoids.
\begin{definition}
  \label{def:wreath}
  Let $M$ be a monoid whose multiplication is denoted by $m,m' \mapsto m \cdot
  m'$. We define $M \wr Q$ as the monoid whose set of elements is $Q \to Q
  \times M$ and whose monoid multiplication is, for $\mu,\mu' : Q \to Q
  \times M$,
  \[ (\mu \bullet \mu') : q \mapsto (\pi_1 \circ \mu' \circ \pi_1 \circ
    \mu(q),\; (\pi_2 \circ \mu(q)) \cdot (\pi_2 \circ \mu' \circ \pi_1 \circ
    \mu(q))) \]
  where $\pi_1 : Q \times M \to Q$ and $\pi_2 : Q \times M \to M$ are the
  projections.
\end{definition}
For instance, if $M$ is the trivial monoid with one element, $Q \wr M$ is
isomorphic to $Q \to Q$ with \emph{reverse} composition as the monoid
multiplication: $f \bullet g = g \circ f$.

\begin{proposition}
  Let $(Q, q_0, R, \delta, \vec{u}_I, F)$ be an SST that computes $f : \Gamma^*
  \to \Sigma^*$ (using the notations of \Cref{def:sst}). For all $c \in \Gamma$,
  we have $\delta(-,c) \in \transmono_{R,\Sigma} \wr Q$, and the SST is copyless
  if and only if $\{\delta(-,c) \mid c \in \Gamma\} \subseteq
  \transmonocl_{R,\Sigma} \wr Q$. Furthermore, for all $w_1 \dots w_n \in \Gamma^*$,
  \[ f(w_1 \dots w_n) \quad=\quad
    F(g(q_0))^\dagger(\alpha^\dagger(\vec{v}))\quad\text{where}\quad (g,\alpha)
    = \delta(-,w_1) \bullet \dots \bullet \delta(-,w_n)\]
\end{proposition}
Finally, it will sometimes be useful to consider monoids of assignments over an
\emph{empty} output alphabet. This allows us to keep track of how the registers
are shuffled around by transitions.
\begin{proposition}
  \label{prop:erase-morphism}
  Let $R$ and $\Sigma$ be disjoint finite sets. There is a monoid morphism
  $\transmono_{R,\Sigma} \to \transmono_{R,\varnothing}$, that sends the
  submonoid $\transmonocl_{R,\Sigma}$ to $\transmonocl_{R,\varnothing}$. For any
  $Q$, this extends to a morphism $\transmono_{R,\Sigma} \wr Q \to
  \transmono_{R,\varnothing} \wr Q $ that sends
  $\transmonocl_{R,\Sigma} \wr Q $ to $\transmonocl_{R,\varnothing} \wr Q$. We
  shall use the name $\tterase_\Sigma$ for both morphisms ($R$ and $Q$
  being inferred from the context).
\end{proposition}
\begin{remark}
  \label{rem:substitution-transition-monoid}
  Consider an SST with a transition function $\delta$. Let $\varphi_\delta \in
  \Hom(\Gamma^*,\,\transmonocl_{R,\varnothing} \wr Q)$ be defined by
  $\varphi_\delta(c) = \tterase_\Sigma(\delta(-,c))$ for $c \in \Gamma$.
  The range $\varphi_\delta(\Gamma^*)$ is precisely the \emph{substitution
    transition monoid (STM)} defined in~\cite[Section~3]{AperiodicSST}.
\end{remark}
\begin{proposition}
  \label{prop:transmonocl-finite}
  For any finite $R$, the monoid $\transmonocl_{R,\varnothing}$ is \emph{finite}.
  As a consequence, the substitution transition monoid of any \emph{copyless}
  SST is finite.
\end{proposition}
\begin{proof}[Proof idea]
  For all $\alpha \in \transmonocl_{R,\varnothing}$ and $r \in R$,
  observe that $|\alpha(r)| \leq |R|$.
\end{proof}

\section{Complements on HDT0L systems, SSTs and polyregular functions}
\label{sec:complements}

Before embarking on the study of our new comparison-free polyregular functions,
we state some minor results that consolidate our understanding of pre-existing
classes.

\subparagraph{Layered HDT0L systems}

Let us transpose the layering condition from SSTs to HDT0L systems. The hierarchy
of models that we get corresponds \emph{with an offset} to layered SSTs.
\begin{definition}
  A HDT0L system $(\Gamma,\Sigma,\Delta,d,(h_c)_{c\in\Gamma},h')$ is
  \emph{$k$-layered} if its working alphabet can be partitioned as $\Delta =
  \Delta_0 \sqcup \dots \sqcup \Delta_k$ such that, for all $c \in \Gamma$ and
  $i \in \{0,\ldots,k\}$:
  \begin{itemize}
  \item for $r \in \Delta_i$, we have $h_c(r) \in (\Delta_0 \sqcup \dots \sqcup
    \Delta_i)^*$;
  \item each letter in $\Delta_i$ appears at most once among all the $h_c(r)$
    for $r \in \Delta_i$ (but those in $\Delta_0 \sqcup \dots \sqcup
    \Delta_{i-1}$ may appear an arbitrary number of times).
  \end{itemize}
\end{definition}

\begin{restatable}{theorem}{layeredhdtOlequiv}
  \label{thm:layered-hdt0l-equiv}
  For $k\in\naturalN$, a function can be computed by a $k$-layered SST if and
  only if it can be specified by a $(k+1)$-layered HDT0L system.

  In particular, regular functions correspond to 1-layered HDT0L systems.
\end{restatable}
The obvious translation from HDT0L systems to SSTs preserves
1-layeredness and produces a single-state machine, so one
may sacrifice copylessness to eliminate states for SSTs.

\begin{restatable}{corollary}{regularsinglestateonelayered}
  \label{cor:regular-single-state-1-layered}
  Every regular function can be computed by a single-state 1-layered SST.
\end{restatable}

The converse to this corollary does not hold: the single-state 1-layered SST of
\Cref{ex:prefixes} computes a function which is not regular (cf.\
\Cref{rem:prefixes-quadratic}).

\subparagraph{Polyregular functions vs layered SSTs}

By applying some results from~\cite{polyregular}, we can state a variant of
\Cref{def:polyreg} which is a bit more convenient for us.
\begin{restatable}{proposition}{polyregregsquaring}
  \label{prop:polyreg-reg-squaring}
  Polyregular functions are the smallest class closed under composition that
  contains the regular functions and the squaring with underlining functions
  $\mathtt{squaring}_\Gamma$.
\end{restatable}
This allows us to show that composing HDT0L transductions with at most
polynomial growth yields the polyregular functions. One direction of
this equivalence is proved by encoding $\mathtt{squaring}_\Gamma$ as a
composition of two SSTs, one of which is \Cref{ex:prefixes}. More precisely:
\begin{restatable}{theorem}{polyreglayered}
  \label{thm:polyreg-layered}
  Let $f : \Gamma^* \to \Sigma^*$. The following are equivalent:
  \begin{enumerate}[(i)]
  \item\label{thm-item:polyreg-layered-polyreg} $f$ is polyregular;
  \item\label{thm-item:polyreg-layered-k-layered} $f$ can be obtained as a
    composition of layered SSTs;
  \item\label{thm-item:polyreg-layered-1-layered} $f$ can be obtained as a
    composition of single-state 1-layered SSTs.
  \end{enumerate}
\end{restatable}

But layered SSTs by themselves are \emph{strictly less expressive} than
polyregular functions, as we shall see later in \Cref{thm:cf-not-hdt0l}.
Therefore, as promised in the introduction:
\begin{corollary}[claimed in {\cite[Section~6]{Marble}}]
  Layered SSTs are \emph{not} closed under composition.
\end{corollary}

\subparagraph{The importance of being stateful}

One interesting aspect of \Cref{thm:layered-hdt0l-equiv} is that 1-layered HDT0L
systems can be seen, through \Cref{rem:hdt0l-single-state}, as a kind of one-way
transducer model for regular functions that does not use an explicit control
state. This is in contrast with copyless SSTs, whose expressivity critically
depends on the states (unlike copyful SSTs).

\begin{restatable}{proposition}{singlestatecopylessweak}
  \label{prop:single-state-copyless-weak}
  The sequential (and therefore regular) function defined by the transducer of
  \Cref{fig:mealy} (\Cref{sec:reg}) \emph{cannot} be computed by a
  \emph{single-state copyless SST}.
\end{restatable}
In fact, the knowledgeable reader can verify that this counterexample belongs to
the \emph{first-order letter-to-letter sequential functions}, one of the weakest
classical transduction classes.

\subparagraph{Closure under map}

The pattern of \Cref{ex:iterated-reverse} (iterated reverse) can be generalized:

\begin{definition}
  \label{def:fmap} 
  Let $f : \Gamma^* \to \Sigma^*$ and suppose that $\#\notin\Gamma\cup\Sigma$.
  We define the function $\fmap(f) : w_1\#\dots\#w_n \in (\Gamma\cup\{\#\})^*
  \mapsto f(w_1)\#\dots\#f(w_n) \in (\Sigma\cup\{\#\})^*$.
\end{definition}

\begin{restatable}{proposition}{fmaphdtOl}
  \label{prop:fmap-hdt0l}
  If $f$ is an HDT0L transduction, then so is $\fmap(f)$. For each $k \geq 1$,
  the functions that can be computed by $k$-layered HDT0L systems are also
  closed under $\fmap$.
\end{restatable}

As an immediate corollary, closure under $\fmap$ holds for both regular and
polyregular functions, but this was already known. In fact,
$\mathtt{map}(f, [x_{1},\dots,x_{n}]) = [f(x_{1}),\dots,f(x_{n})]$ is an
essential primitive in the \emph{regular list functions}~\cite{Daviaud} and
\emph{polynomial list functions}~\cite[\S4]{polyregular}, two list-processing
programming languages that characterize regular and polyregular functions
respectively. We will come back to this point in \Cref{cor:fmap-cfpolyreg} and
the subsequent remark.

\section{Composition by substitution}
\label{sec:cbs}

At last, we now introduce the class of \emph{comparison-free polyregular
  functions}. The simplest way to define them is to start from the regular
functions.

\begin{definition}
  \label{def:cbs}
  Let $f : \Gamma^* \to I^*$, and for each $i \in I$, let $g_i : \Gamma^* \to
  \Sigma^*$. The \emph{composition by substitution} of $f$ with the family
  $(g_i)_{i\in I}$ is the function
  \[\CbS(f,(g_i)_{i\in I}) ~~:~~ w ~\mapsto~
  g_{i_1}(w) \ldots g_{i_k}(w)\ ~~\text{where}\ i_1 \ldots i_k = f(w)\]
  That is, we first apply $f$ to the input, then every letter $i$ in the result
  of $f$ is substituted by the image of the original input by $g_i$.
  Thus, $\CbS(f,(g_i)_{i\in I})$ is a function $\Gamma^* \to \Sigma^*$.
\end{definition}

\begin{definition}
  \label{def:cfpolyreg-cbs}
  The smallest class of string-to-string functions closed under $\CbS$ and
  containing all regular functions is called the class of \emph{comparison-free
    polyregular functions}.
\end{definition}

\begin{example}
  \label{ex:squaring}
  The following variant of \enquote{squaring with underlining} (cf.\
  \Cref{def:polyreg}) is comparison-free polyregular:
  $\cfsquaring_\Gamma : \mathtt{123} \in \Gamma^* \mapsto
  \mathtt{\underline{1}123\underline{2}123\underline{3}123} \in
  (\Gamma\cup\underline\Gamma)^*$.

  Indeed, it can be expressed as $\cfsquaring_\Gamma = \CbS(f,(g_i)_{i \in I})$
  where $I = \Gamma\cup\{\#\}$, the function $f : w_1 \dots w_n \mapsto w_1 \#
  \dots w_n \#$ is regular (more than that, a morphism between free monoids) and
  $g_\# = \mathrm{id}$, $g_c : w \mapsto \underline{c}$ for $c \in \Gamma$ are
  also regular. Its growth rate is quadratic, while regular functions have at
  most linear growth. Other examples that also require a single composition by
  substitution are given in \Cref{thm:cf-not-hdt0l}.
\end{example}

We can already justify the latter half of the name of our new class. Using the
\enquote{polynomial list functions} mentioned at the end of the previous
section, we prove:
\begin{restatable}{theorem}{polyregcbs}
  \label{thm:polyreg-cbs}
  Polyregular functions are closed under composition by substitution.
\end{restatable}
\begin{corollary}
  Every comparison-free polyregular function is, indeed, polyregular.
\end{corollary}


Fundamentally, \Cref{def:cfpolyreg-cbs} is inductive: it considers the functions
generated from the base case of regular functions by applying compositions by
substitution. The variant below with more restricted generators is sometimes
convenient.



\begin{definition}
  A string-to-string function is said to be:
  \begin{itemize}
  \item of \emph{rank at most $0$} if it is regular;
  \item of \emph{rank at most $k+1$ (for $k \in \naturalN$)} if it can be
    written as $\CbS(f, (g_i)_{i \in I})$ where $f : \Gamma^* \to I^*$ is
    \emph{regular} and each $g_i : \Gamma^* \to \Sigma^*$ is of rank at most $k$.
  \end{itemize}
\end{definition}
\begin{restatable}{proposition}{CbSind}
  \label{prop:CbS-ind}
  A function $f$ is comparison-free polyregular if and only if there exists some
  $k \in \naturalN$ such that $f$ has rank at most $k$. In that case, we write
  $\rk(f)$ for the least such $k$ and call it the \emph{rank} of $f$. If
  $(g_i)_{i \in I}$ is a family of comparison-free polyregular functions,
  \[\qquad \rk(\CbS(f,(g_i)_{i \in I})) \leq 1 + \rk(f) + \max_{i \in I} \rk(g_i) \]
\end{restatable}

A straightforward consequence of this definition is that, just like regular
functions, cfp functions are closed under \emph{regular conditionals} and
concatenation.
\begin{restatable}{proposition}{cfpconditionals}
  \label{prop:cfp-conditionals}
  Let $f, g : \Gamma^* \to \Sigma^*$ be comparison-free polyregular functions
  and $L \subseteq \Gamma^*$ be a regular language. The function that coincides
  with $f$ on $L$ and with $g$ on $\Gamma^* \setminus L$ is cfp, and so is $w
  \in \Gamma^* \mapsto f(w) \cdot g(w)$; both have rank at most
  $\max(\rk(f),\rk(g))$.
\end{restatable}

\section{Comparison-free pebble transducers}
\label{sec:cfpebble}

We now characterize our function class by a machine model that will explain our
choice of the adjective \enquote{comparison-free}, as well as the operational
meaning of the notion of rank we just defined. It is based on the \emph{pebble
  transducers} first introduced for trees by Milo, Suciu and Vianu~\cite{Pebble}
and later investigated in the special case of strings by Engelfriet and
Maneth~\cite{PebbleString,PebbleComposition}. However, the definition using
composition by substitution will remain our tool of choice to prove further
properties, so the next sections do not depend on this one.

\begin{definition}
  Let $k \in \naturalN$ with $k \geq 1$. Let $\Gamma,\Sigma$ be finite alphabets
  and $\leftmarker,\rightmarker \notin \Gamma$.

  A \emph{$k$-pebble stack} on an input string $w \in \Gamma^*$ consists of
  a list of $p$ positions in the string $\markers{w}$ (i.e.\ of $p$
  integers between 1 and $|w|+2$) for some $p \in \{1,\ldots,k\}$. We therefore
  write $\Stack{k} = \naturalN^0 \cup \naturalN^1 \cup \dots \cup \naturalN^k$, keeping in mind
  that given an input $w$, we will be interested in \enquote{legal} values
  bounded by $|w|+2$.

  A \emph{comparison-free $k$-pebble transducer ($k$-CFPT)} consists of a finite
  set of states $Q$, an initial state $q_I \in Q$ and
  a family of transition functions
  \[ Q \times (\Gamma \cup \{\leftmarker,\rightmarker\})^p \to Q \times
    (\naturalN^p \to \Stack{k}) \times \Sigma^* \quad \text{for}\ 1 \leq p \leq
    k \]
  where the $\naturalN^p$ on the left is considered as a subset of $\Stack{k}$.
  For a given state and given letters $(c_1,\ldots,c_p) \in (\Gamma \cup
  \{\leftmarker,\rightmarker\})^p$, the allowed values for the stack update
  function $\naturalN^p \to \Stack{k}$ returned by the transition function are:
  
  \begin{tabular}{cccccc}
    (identity) & $(i_1,\ldots,i_p)$ & $\mapsto$ & $(i_1,\ldots,i_p)$
    & $\in$ & $\naturalN^p$\\
 (move left, only allowed when $c_p \neq \leftmarker$)& $(i_1,\ldots,i_p)$
    & $\mapsto$ & $(i_1,\ldots,i_p - 1)$ & $\in$ & $\naturalN^p$\\
 (move right, only allowed when $c_p \neq \rightmarker$) &
    $(i_1,\ldots,i_p)$ & $\mapsto$ & $(i_1,\ldots,i_p + 1)$ & $\in$ & $\naturalN^p$\\
    (push, only allowed when $p \leq k-1$) & $(i_1,\ldots,i_p)$ & $\mapsto$ &
    $(i_1,\ldots,i_p, 1)$ & $\in$ & $\naturalN^{p+1}$\\
    (pop, only allowed when $p \geq 1$) & $(i_1,\ldots,i_p)$ & $\mapsto$ &
     $(i_1,\ldots,i_{p-1})$ & $\in$ & $\naturalN^{p-1}$
  \end{tabular}

  \noindent
  (Note that the codomains of all these functions are indeed subsets of
  $\Stack{k}$.)
\end{definition}


The run of a CFPT over an input string $w \in \Gamma^*$ starts in the initial
configuration comprising the initial state $q_I$, the initial $k$-pebble stack
$(1) \in \naturalN^1$, and the empty string as an initial output log. As long as
the current stack is non-empty
a new configuration is
computed by applying the transition function to $q$ and to
$((\markers{w})[i_1],\ldots,(\markers{w})[i_p])$ where $(i_1,\ldots,i_p)$ is the
current stack; the resulting stack update function is applied to
$(i_1,\ldots,i_p)$ to get the new stack, and the resulting output string in
$\Sigma^*$ is appended to the right of the current output log. If the CFPT ever
terminates by producing an empty stack, the \emph{output associated to
  $w$} is the final value of the output log.

This amounts to restricting in two ways\footnote{There is also an inessential
  difference: the definition given in~\cite{polyregular} does not involve end
  markers and handles the edge case of an empty input string separately. This
  has no influence on the expressiveness of the transducer model. Our use of end
  markers follows~\cite{EngelfrietHoogeboom,PebbleMinimization}.} the
definition of \emph{pebble transducers} from~\cite[\S2]{polyregular}:
\begin{itemize}
\item in a general pebble transducer, one can \emph{compare positions}, i.e.\
  given a stack $(i_1,\ldots,i_p)$, the choice of transition can take into
  account whether\footnote{One would get the same computational power, with the
    same stack size, by only testing whether $i_j = i_p$ for $j \leq p-1$ as
    in~\cite{Pebble} (this is also essentially what happens in the nested
    transducers of~\cite{PebbleMinimization}).} $i_j \leq i_{j'}$ (for any $1
  \leq j,j' \leq p$);
\item in a \enquote{push}, new pebbles are initialized to the leftmost
  position
  ($\leftmarker$) for a CFPT, instead of starting at the same
  position as the previous top of the stack (the latter would ensure the
  equality of two positions at some point; it is therefore an implicit
  comparison that we must relinquish to be truly \enquote{comparison-free}).
\end{itemize}
This limitation is similar to (but goes a bit further than) the
\enquote{invisibility} of pebbles in a transducer model introduced by Engelfriet
et al.~\cite{InvisiblePebbles} (another difference, unrelated to position
comparisons, is that their transducers use an unbounded number of invisible
pebbles).

\begin{remark}
  Our definition guarantees that \enquote{out-of-bounds errors} cannot happen
  during the run of a comparison-free pebble transducer. The sequence of
  successive configurations is therefore always well-defined. But it may be
  infinite, that is, it may happen that the final state is never reached. 
  Thus, a CFPT defines a \emph{partial} function.

  That said, the set of inputs for which a given pebble tree transducer does not
  terminate is always a regular language~\cite[Theorem~4.7]{Pebble}. This
  applies \textit{a fortiori} to CFPTs. Using this, it is
  possible\footnote{Proof idea: do a first left-to-right pass to determine
   whether the input leads to non-termination of the original CFPT; if so,
   terminate immediately with an empty output; otherwise, move the first pebble
   back to the leftmost position and execute the original CFPT's behavior. This
   can be implemented by adding finitely many states, including those for a DFA
   recognizing non-terminating inputs.}
  to extend any partial
  function $f : \Gamma^* \rightharpoonup \Sigma^*$ computed by a $k$-CFPT into a
  total function $f' : \Gamma^* \to \Sigma^*$ computed by another $k$-CFPT for
  the same $k \in \naturalN$, such that $f'(x) = f(x)$ for $x$ in the domain of
  $f$ and $f'(x) = \varepsilon$ otherwise. This allows us to \emph{only consider
  CFPTs computing total functions} in the remainder of the paper.
\end{remark}

A special case of particular interest is $k=1$: the transducer has a single
reading head, push and pop are always disallowed.
\begin{theorem}[{\cite{SST}}]
  \label{thm:reg-2dft}
  Copyless SSTs and 1-CFPTs -- which are more commonly called \emph{two-way
    (deterministic) finite transducers (2DFTs)} -- are equally expressive.
\end{theorem}

Since we took copyless SSTs as our reference definition of regular functions,
this means that 2DFTs characterize regular functions. But putting it this way is
historically backwards: the equivalence between 2DFTs and MSO transductions came
first~\cite{EngelfrietHoogeboom} and made this class deserving of the name
\enquote{regular functions} before the introduction of copyless SSTs.

\begin{remark}
  \label{rem:numbering-convention}
  There are two different numbering conventions for pebble transducers.
  In~\cite{polyregular,PebbleMinimization}, 2DFTs are 1-pebble transducers,
  which is consistent with our choice. However, several other papers
  (e.g.~\cite{Pebble,PebbleString,PebbleComposition,InvisiblePebbles,Gaetan})
  consider that a 2DFT is a \emph{0-pebble} transducer (likewise,
  in~\cite{Marble}, 2DFTs are 0-marble transducers). This is because they think
  of a pebble automaton not as a restricted multi-head automaton, but as an
  enriched 2DFA that can drop stationary markers (called pebbles) on input
  positions, with a single moving head that is not a pebble.
\end{remark}


Let us now show the equivalence with \Cref{def:cfpolyreg-cbs}. The reason for
this is similar to the reason why $k$-pebble transducers are equivalent to the
\emph{$k$-nested transducers}\footnote{Remark: nested transducers should yield a
  machine-independent definition of polyregular functions as the closure of
  regular functions under a CbS-like operation that relies on \emph{origin
    semantics}~\cite[\S5]{MuschollPuppis}.} of~\cite{PebbleMinimization}, which
is deemed \enquote{trivial} and left to the reader
in~\cite[Remark~6]{PebbleMinimization}. But in our case, one direction
(\Cref{thm:cfpebble-hard}) involves an additional subtlety compared to
in~\cite{PebbleMinimization}; to take care of it, we use the fact that the
languages recognized by pebble automata are regular (this is also part
of~\cite[Theorem~4.7]{Pebble}) together with regular conditionals
(\Cref{prop:cfp-conditionals}).

\begin{restatable}{proposition}{cfpebbleeasy}
  \label{prop:cfpebble-easy}
  If $f$ is computed by a $k$-CFPT, and the $g_i$ are computed by $l$-CFPTs,
  then $\CbS(f,(g_i)_{i \in I})$ is computed by a $(k+l)$-CFPT.
\end{restatable}

\begin{restatable}{theorem}{cfpebblehard}
  \label{thm:cfpebble-hard}
  If $f : \Gamma^* \to \Sigma^*$ is computed by a $k$-CFPT, for $k \geq 2$, then
  there exist a finite alphabet $I$, a regular function $h : \Gamma^* \to I^*$
  and a family $(g_i)_{i\in I}$ computed by $(k-1)$-CFPTs such that $f =
  \CbS(h,(g_i)_{i\in I})$.
\end{restatable}

\begin{restatable}{corollary}{corcfpebble}
  \label{cor:cfpebble}
  For all $k \in \naturalN$, the functions computed by $(k+1)$-CFPTs are exactly
  the comparison-free polyregular functions of rank at most $k$.
\end{restatable}



\section{Composition of basic functions}
\label{sec:composition}

Another possible definition of cfp functions consists in swapping out
$\mathtt{squaring}_\Gamma$ for some other function in
\Cref{prop:polyreg-reg-squaring}:
\begin{restatable}{theorem}{compositionsquaring}
  \label{thm:composition-squaring}
  The class of comparison-free polyregular functions is the smallest class
  closed under usual function composition and containing both all regular
  functions and the functions $\cfsquaring_\Gamma$ (cf.\
  \Cref{ex:squaring}) for all finite alphabets $\Gamma$.
\end{restatable}
The hard part is to show that cfp functions are closed under composition. We
exploit the following combinatorial phenomenon, often applied to the study of
copyless SSTs: a copyless register assignment, i.e.\ an element of
$\transmonocl_{R,\Delta}$ (cf.~\Cref{sec:transition-monoids}), can be specified
by
\begin{itemize}
  \item a \enquote{shape} described by an element of the \emph{finite} monoid
        $\transmonocl_{R,\varnothing}$ (\Cref{prop:transmonocl-finite}),
\item plus finitely many \enquote{labels} in $\Sigma^*$ (where $\Sigma$ is the
  output alphabet) describing the constant factors that will be concatenated
  with the old register contents to give the new ones.
\end{itemize}
\begin{proposition}
  \label{prop:shape-plus-labels}
  There is a bijection
  \[ \transmonocl_{R,\Delta} ~~\cong~~ \left\{\left(\alpha,\vec\ell\right) \;\middle|\; \alpha \in \transmonocl_{R,\varnothing},\; \vec\ell \in \prod_{r \in R} (\Delta^*)^{|\alpha(r)|+1} \right\} \]
  through which
  $\tterase_\Delta : \transmonocl_{R,\Delta} \to \transmonocl_{R,\varnothing}$
  can be seen as simply removing the \enquote{labels} $\vec\ell$.
\end{proposition}
\begin{proof}[Proof idea]
  Let $\beta \in \transmonocl_{R,\Delta}$. For each $r \in R$, one can write
  $\beta(r) = w_0 r'_1 w_1 \dots r'_n w_n$ with $w_0,\ldots,w_n \in \Delta^*$
  and $r'_1,\dots,r'_n \in R$ such that
  $r'_1 \dots r'_n = \tterase_\Delta(\beta)(r) \in R^*$.
\end{proof}

This provides a clear way to represent a copyless register assignment inside the
working memory of an SST: store the shape in the state and the labels in
registers. Another important fact for us is that given two assignments
$\beta,\beta' \in \transmonocl_{R,\Delta}$ the labels of $\beta\bullet\beta'$
can be obtained as a \emph{copyless} recombination of the labels of $\beta$ and
$\beta'$.

(There is a subtlety worth mentioning here: while the set of stateful
transitions $\transmonocl_{R,\Delta} \wr Q$ also admits a \enquote{shape +
  labels} representation, its monoid multiplication does \emph{not} have this
copylessness property. This prevents a naive proof of the closure under
composition of copyless SSTs from working. Nevertheless, the composition of two
regular functions \emph{is} always regular, and we rely on this fact to prove
\Cref{thm:composition-squaring}.)

The rest of the proof of \Cref{thm:composition-squaring} is relegated to the
technical appendix.

\section{Rank vs asymptotic growth}
\label{sec:minimization}

Our next result is the comparison-free counterpart to recent work on polyregular
functions by Lhote~\cite{PebbleMinimization}, whose proof techniques (in
particular the use of Ramsey's theorem) we reuse. Compare
item~(\ref{thm-item:rank-growth}) below to the main theorem
of~\cite{PebbleMinimization} and item~(\ref{thm-item:rank-pow})
-- which provides yet another definition of cfp functions --
to~\cite[Appendix~A]{PebbleMinimization}.

\begin{restatable}{theorem}{characterizationrank}
  \label{thm:characterization-rank}
  Let $f : \Gamma^* \to \Sigma^*$ and $k \in \naturalN$. The following are
  equivalent:
  \begin{enumerate}[(i)]
  \item\label{thm-item:rank-def} $f$ is comparison-free polyregular with rank at
    most $k$;
  \item\label{thm-item:rank-growth} $f$ is comparison-free polyregular and
    $|f(w)| = O(|w|^{k+1})$;
  \item\label{thm-item:rank-pow} there exists a regular function $g :
    (\{0,\dots,k\} \times \Gamma)^* \to \Sigma^*$ such that $f = g \circ
    \cfpow^{(k+1)}_\Gamma$,\\ with the following inductive definition:
    $\cfpow^{(0)}_\Gamma :~ w \in \Gamma^* \mapsto \varepsilon \in
    (\varnothing\times\Gamma)^*$ and
    \[ \cfpow^{(n+1)}_\Gamma ~:~~ w \mapsto
      (n,w_1) \cdot \cfpow^{(n)}_\Gamma(w) \cdot\ldots\cdot
      (n,w_{|w|}) \cdot  \cfpow^{(n)}_\Gamma(w)
    \]
  \end{enumerate}
  To make $(\ref{thm-item:rank-growth}) \implies (\ref{thm-item:rank-def})$ more
  precise, if $f$ is cfp with $\rk(f) \geq 1$, then it admits a sequence of
  inputs $w_0,w_1,\ldots\in\Gamma^*$ such that $|w_n|\to+\infty$ and $|f(w_n)| =
  \Omega(|w_n|^{\rk(f)+1})$.
\end{restatable}

Note that $\cfpow^{(2)}_\Gamma$ and $\cfsquaring_\Gamma$ are the
same up to a bijection $\{0,1\}\times\Gamma \cong \Gamma\cup\underline\Gamma$.
\begin{remark}
  \label{rem:marble-minimization}
  The growth of an HDT0L transduction is also related, in a very similar way to
  item~(\ref{thm-item:rank-growth}) above, to the number of layers required in
  any SST that computes it~\cite[\S5]{Marble}.
\end{remark}

\subparagraph*{Some proof elements} Let us present a few definitions and lemmas
to give an idea of the ingredients that go into the proof. Those technical
details take up the rest of this section.

Lhote's paper~\cite{PebbleMinimization} makes a heavy use of
\emph{factorizations} of strings that depend on a morphism to a finite monoid.
This is also the case for our proof, but we have found that a slightly different
definition of the kind of factorization that we want works better for us.

\begin{definition}[similar but not equivalent to
  {\cite[Definition~19]{PebbleMinimization}}]
  \label{def:r-split}
  An \emph{$r$-split} of a string $s \in \Gamma^*$ according to a morphism
  $\varphi : \Gamma^* \to M$ is a tuple $(u,v_1,\dots,v_r,w) \in
  (\Gamma^*)^{r+2}$ such that:
  \begin{itemize}
  \item $s = u v_1 \dots v_r w$ with $v_i$ non-empty for all $i \in \{1,\dots,r\}$;
  \item $\varphi(u) = \varphi(u v_1) = \dots = \varphi(u v_1 \dots v_r)$;
  \item $\varphi(w) = \varphi(v_r w) = \dots = \varphi(v_1 \dots v_r w)$.
  \end{itemize}
\end{definition}
\begin{proposition}[immediate from the definition]
  $(u,v_1,\dots,v_r,w)$ is an $r$-split if and only if, for all
  $i \in \{1,\dots,r\}$, $(u v_1 \dots v_{i-1}, v_i, v_{i+1} \dots v_r w)$ is a
  1-split.
\end{proposition}

The difference with the \emph{$(1,r)$-factorizations}
of~\cite[Definition~19]{PebbleMinimization} is that we have replaced the
equality and idempotency requirements on $\varphi(v_1),\dots,\varphi(v_n)$ by
the \enquote{boundary conditions} involving $\varphi(u)$ and $\varphi(w)$
(actually, $(1,r+2)$-factorizations induce $r$-splits). This change allows us to
establish a subclaim used in the proof of \Cref{lem:dichotomy} in an elementary
way.

The point of $r$-splits is that given a split of an input string according to
the morphism that sends it to the corresponding transition in a SST, we have
some control over what happens to the output of the SST if we pump a middle
factor in the split. Furthermore, it suffices to consider a quotient of the
transition monoid which is finite when the SST is copyless (this is similar to
\Cref{prop:transmonocl-finite}). More precisely, we have the  key
lemma below, which is used pervasively throughout our proof of
\Cref{thm:characterization-rank}:
\begin{restatable}{lemma}{producingtriples}
  \label{lem:producing-triples}
  Let $f : \Gamma^* \to \Sigma^*$ be a regular function. There exist a morphism
  to a \emph{finite} monoid $\nu_{f} : \Gamma^* \to \prodmono(f)$ and, for each $c \in
  \Sigma$, a set of \emph{producing triples} $P(f,c) \subseteq \prodmono(f)^3$ such
  that, for any \emph{1-split} according to $\nu_f$ composed of $u,v,w \in
  \Gamma^*$ -- i.e.\ $\nu_f(uv) = \nu_f(u)$ and $\nu_f(vw) = \nu_f(w)$ -- we
  have:
  \begin{itemize}
  \item if $(\nu_f(u),\nu_f(v),\nu_f(w))\in P(f,c)$, then $|f(uvw)|_c >
    |f(uw)|_c$;
  \item otherwise (when the triple is not producing), $|f(uvw)|_c = |f(uw)|_c$.
  \end{itemize}
  Furthermore, in the producing case, we get as a consequence that $\forall n
  \in \naturalN,\; |f(u v^n w)|_c \geq n$.
\end{restatable}

\begin{definition}
  We fix once and for all a choice of $\prodmono(f)$, $\nu_{f}$ and $P(f,c)$ for
  each $c \in \Sigma$ and regular $f : \Gamma^* \to \Sigma^*$. We say that a
  1-split $(u,v,w)$ is \emph{producing with respect to} $(f,c)$ when
  $(\nu_{f}(u),\nu_{f}(v),\nu_{f}(w))\in P(f,c)$. For $\Pi \subseteq \Sigma$, we
  also set $P(f,\Pi) = \bigcup_{c \in \Pi} P(f,c)$.
\end{definition}

Something like \Cref{lem:producing-triples} (but not exactly) appears in the
proof of {\cite[Lemma~18]{PebbleMinimization}}. We first apply it to prove the
following lemma, which is morally a counterpart to the \enquote{$k=1$ case} of
the central Dichotomy Lemma from~\cite{PebbleMinimization}, with $r$-splits instead of $(k,r)$-factorizations.
\begin{restatable}{lemma}{dichotomy}
  \label{lem:dichotomy}
  Let $f : \Gamma^* \to \Sigma^*$ be regular and $\varphi : \Gamma^* \to M$ be a
  morphism with $M$ finite. Suppose that $\pi\circ\varphi = \nu_{f}$ for some
  other morphism $\pi : M \to \prodmono(f)$. Let $r \geq 1$ and
  $\Pi \subseteq \Sigma$.

  We define $L(f,\Pi,\varphi,r)$ to be the set of strings that admit an $r$-split
  $s = u v_1 \dots v_r w$ according to $\varphi$ such that all the triples $(u
  v_1 \dots v_{i-1}, v_i, v_{i+1} \dots v_r w)$ are producing with respect to
  $(f,\Pi)$ -- let us call this a producing $r$-split with respect to
  $(f,\Pi,\varphi)$.

  Then $L(f,\Pi,\varphi,r)$ is a regular language, and $\sup\{|f(s)|_\Pi \mid s
  \in \Gamma^* \setminus L(f,\Pi,\varphi,r) \} < \infty$.
\end{restatable}

Our proof of the above lemma uses the proposition below, analogous to
{\cite[Claim~20]{PebbleMinimization}}. Its statement is a bit stronger than
necessary for this purpose, but it will be reused in the proof of
\Cref{thm:polyreg-not-cf}; as for its proof, this is where a standard Ramsey
argument occurs.

\begin{restatable}{proposition}{framsey}
  \label{prop:ramsey}
  Let $\Gamma$ be an alphabet, $M$ be a finite monoid, $\varphi :
  \Gamma^* \to M$ be a morphism and $r\geq1$. There exists $N \in \naturalN$ such that any
  string $s = uvw \in \Gamma^*$ such that $|v| \geq N$ admits an $r$-split $s =
  u' v'_1 \dots v'_r w'$ according to $\varphi$ in which $u$ is a prefix of
  $u'$ and $w$ is a suffix of $w'$.
\end{restatable}

To leverage \Cref{lem:dichotomy}, we combine it with an elementary property of
composition by substitution that does not depend on the previous technical
development. (Compare the assumptions of the lemma below with the conclusion of
\Cref{lem:dichotomy}.)

\begin{restatable}{lemma}{inlining}
  \label{lem:inlining}
  Let $g : \Gamma^* \to I^*$ be a regular function and, for each $i \in I$, let
  $h_i : \Gamma^* \to \Sigma^*$ be comparison-free polyregular of rank at most
  $k$. Suppose that $\displaystyle\sup_{s \in \Gamma^*} |g(s)|_J < \infty$ where
  \[ J = \begin{cases}
      \{i \in I \mid \rk(h_i) = k\} & \text{when}\ k \geq 1\\
      \{i \in I \mid |h_i(\Gamma^*)| = \infty \} & \text{when}\ k = 0
    \end{cases}\]
  (Morally, regular functions with finite range play the role of
  \enquote{comparison-free polyregular functions of rank $-1$}.) Then
  $\rk(\CbS(g,(h_i)_{i\in I}))\leq k$.
\end{restatable}

The above lemma can be compared to~\cite[Claim~22]{PebbleMinimization}, but it
also seems to be related to the way the \enquote{nested transducer} $R_{k+1}$ is
defined in the proof of the Dichotomy Lemma in~\cite{PebbleMinimization}:
indeed, $R_{k+1}$ can call either a $k$-nested subroutine or a $(k-1)$-nested
one.

The remainder of the proof of \Cref{thm:characterization-rank} consists mainly
of a rather technical induction on the rank, which we present in the appendix.

\section{Separation results}
\label{sec:separation}

Let us now demonstrate that the class of cfp functions
is incomparable with the class of HDT0L transductions and is a strict subclass
of polyregular functions.
\begin{restatable}{theorem}{cfnothdtOl}
\label{thm:cf-not-hdt0l}
There exist comparison-free polyregular functions which are not HDT0L:
\begin{enumerate}[(i)]
\item the function $a^n \in \{a\}^* \mapsto (a^n b)^{n+1}
  \in \{a,b\}^*$ for $a \neq b$;
\item the function $w \in \Sigma^* \mapsto
  w^{|w|}$ for $|\Sigma| \geq 2$ (a simplification of \Cref{ex:squaring});
\item \emph{(from~\cite[\S6]{Marble})} the cfp functions that map $a^n \# w
  \in \Sigma^*$ to $(w\#)^n$ for $a,\# \in \Sigma$, $a \neq \#$.
\end{enumerate}
\end{restatable}
\begin{remark}
  \label{rem:unary-square-hdt0l}
  The first example in \cite[\S5]{Marble} shows that
  $a^n \mapsto a^{n \times n}$ is HDT0L (via the equivalent model of marble
  transducers), hence the necessity of $|\Sigma| \geq 2$ above. More generally,
  Douéneau-Tabot has shown very recently that \emph{every polyregular function
    with unary output alphabet is HDT0L}~\cite{Gaetan}. So polyregular functions
  with unary output coincide with \emph{polynomial growth $\naturalN$-rational
    series} (cf.~\Cref{rem:rational-series}), and the latter admit several
  algebraic characterizations in the literature (see~\cite{Reutenauer}
  and~\cite[Chapter~9, Exercise~1.2]{BerstelReutenauer}).
\end{remark}
\begin{restatable}{theorem}{polyregnotcf}
  \label{thm:polyreg-not-cf}
  Some HDT0L transductions are polyregular but not comparison-free:
  \begin{enumerate}[(i)]
  \item $f : a^n \in \{a\}^* \mapsto ba^{n-1}b\ldots{}baabab$ (with
    $f(\varepsilon) = \varepsilon$ and $f(a) = b$);
  \item $\fmap(a^n\mapsto a^{n\times{}n}) : a^{n_1}\#\dots\#a^{n_k} \mapsto
    a^{n_1\times n_1}\#\dots\#a^{n_k\times n_k}$ (cf.\ \Cref{def:fmap}).
  \end{enumerate}
\end{restatable}
\begin{remark}
  \label{rem:not-cfp-unary}
  The function
  $a^{n_1}\#\dots\#a^{n_k} \mapsto a^{n_1 \times n_1 + \dots + n_k \times n_k}$
  obtained by erasing the $\#$s in the output of
  $\fmap(a^n\mapsto a^{n\times{}n})$ is also \emph{not} comparison-free. This
  result implies the second item of \Cref{thm:polyreg-not-cf} by composition
  with the erasing morphism; we do not prove it here, but it appears in
  Douéneau-Tabot's aforementioned paper~\cite{Gaetan}. Therefore, according
  to~\cite{Gaetan}, \emph{not} every polyregular function with unary output is
  comparison-free.
\end{remark}

To see why the first of the two functions in \Cref{thm:polyreg-not-cf} is HDT0L,
observe that it is \Cref{ex:prefixes} for $\Gamma = \{a\}$ (taking
$b = \underline{a}$). As for the second one, combine \Cref{prop:fmap-hdt0l} and
the first observation in \Cref{rem:unary-square-hdt0l}.


The non-membership parts of Theorems~\ref{thm:cf-not-hdt0l}
and~\ref{thm:polyreg-not-cf} require more work. For the former, we
use pumping arguments on HDT0L systems. Item (ii) of \Cref{thm:polyreg-not-cf}
is handled by first appealing to \Cref{thm:characterization-rank} to reduce to
showing that $\fmap(a^n \mapsto a^{n\times n}) \neq \CbS(g, (h_i)_{i \in I})$
when $g$ and all the $h_i$ are \emph{regular} functions; a combination of
pumping and of a combinatorial argument then shows that inputs with $\card{I}$
occurrences of $\#$ suffice to discriminate the two sides of the inequality.
This result also has the following consequence:
\begin{corollary}
  \label{cor:fmap-cfpolyreg}
  Comparison-free polyregular functions are not closed under $\fmap$.
\end{corollary}
\begin{remark}
  \label{rem:cfp-list}
  Contrast with \Cref{prop:fmap-hdt0l}. The discussion that follows that
  proposition lends some significance to the above corollary: the latter rules
  out the obvious conjectures for a characterization of cfp functions in the
  style of regular/polynomial list functions.
\end{remark}

As for item (i) of \Cref{thm:polyreg-not-cf}, it concerns a function
whose domain consists of words over a unary alphabet, i.e., up to isomorphism,
a \emph{sequence}. This motivates the study of such sequences, which is the subject of the next section.

\section{Comparison-free polyregular sequences}
\label{sec:sequences}

From now on, we identify $\naturalN$ with the set of words $\{a\}^*$ and freely
speak, for instance, of cfp sequences $\naturalN \to \Gamma^*$ instead of cfp
functions $\{a\}^* \to \Gamma^*$. It turns out that cfp sequences admit a
characterization as finite combinations of what we call \emph{poly-pumping
  sequences}.

\begin{definition}
  A \emph{poly-pumping sequence} is a function of the form $\interp{e} : \naturalN \to \Sigma^*$ where
  \begin{itemize}
  \item $e$ is a \emph{polynomial word expression} generated by $e ~
    \mathrel{::=} ~ w \bnfalt e \cdot e' \bnfalt e^*$ where $w \in \Sigma^*$;
  \item $\interp{w}(n) = w$, $\interp{e \cdot e'}(n) =
    \interp{e}(n)\interp{e'}(n)$ and $\interp{e^*}(n) =
    (\interp{e}(n))^n$.
  \end{itemize}
  The \emph{star-height} of a polynomial word expression is defined in the usual
  way.
\end{definition}

\begin{restatable}{theorem}{cfpseq}
  \label{thm:cfp-seq}
  Let $s : \naturalN \to \Sigma^*$ and $k\in\naturalN$. The sequence $s$ is
  comparison-free polyregular with $\rk(s)\leq k$ if and only if there exists $p
  > 0$ such that, for any $m < p$, there is a polynomial word expression $e$ of
  star-height at most $k+1$ such that $\forall n\in\naturalN,\; s((n+1)p+m)
  = \interp{e}(n)$.
\end{restatable}
In short, the cfp sequences are exactly the \emph{ultimately periodic
  combinations} of poly-pumping sequences.
Our proof strategy is an induction on $k$.

The base case $k=0$ says that regular sequences are ultimately periodic
combinations of \emph{pumping sequences} $n \mapsto
u_0 (v_1)^n \dots (v_l)^n u_l$.
An essentially equivalent result is stated with a proof sketch using 2DFTs
in~\cite[p.~90]{Eugenia}; we propose an alternative proof using
copyless SSTs.
(Non-deterministic two-way transducers (2NFTs) taking unary inputs have also
been studied~\cite{Guillon}; furthermore, the notion of \enquote{$k$-iterative
  language} that appears in a pumping lemma for general 2NFTs~\cite{Pumping2NFT}
is related to the shape of the above pumping sequences.)

To make the inductive step go through, it is enough to synchronize the
periods of the different poly-pumping sequences involved and to observe that
$\CbS(\interp{e},(\interp{e'_i})_{i \in I})$ is realized by an expression
obtained by substituting the $e'_i$ for $i$ in $e$.

Coming back to \Cref{thm:polyreg-not-cf}, we show that $a^n \mapsto
ba^{n-1}b\ldots bab$ is not comparison-free polyregular by proving that its
subsequences are \emph{not} poly-pumping: for every poly-pumping sequence
$s : \naturalN \to \{a,b\}^*$, there is a uniform bound on the number of distinct
contiguous subwords of the shape $baa\ldots ab$ occuring in each $s(n)$ for
$n\in\naturalN$.
Another consequence of \Cref{thm:cfp-seq} that we establish by induction over
expressions contrasts with \Cref{cor:fmap-cfpolyreg}:
\begin{restatable}{corollary}{cfpmapseq}
\label{cor:cfp-map-seq}
If $f : \Gamma^* \to \Sigma^*$ and $s : \naturalN \to (\Gamma \cup \{\#\})^*$ are cfp, so is $\fmap(f) \circ s$.
\end{restatable}

\section{Further topics}
\label{sec:perspectives}



\subparagraph{Functional programming}

We mentioned in the introduction a forthcoming characterization of cfp functions
using Church-encoded strings in a $\lambda$-calculus with linear types, in the
vein of our previous results~\cite{aperiodic,freeadditives}. Meanwhile,
\Cref{cor:fmap-cfpolyreg} could be understood as negative result in the search
for another kind of functional programming characterization
(cf.~\Cref{rem:cfp-list}).

It is also worth noting that the copying discipline of layered SSTs is very
similar to what happens in the \emph{parsimonious
  $\lambda$-calculus}~\cite{Mazza}: a datum of type $\oc\tau$ cannot be
duplicated into two copies of the same type $\oc\tau$, but it may yield an
arbitrary number of copies of type $\tau$ without the modality `$\oc$'. Since
the function classes defined following the methodology
of~\cite{aperiodic,freeadditives} are automatically closed under composition,
\Cref{thm:polyreg-layered} leads us to conjecture that polyregular functions can
be characterized in a variant of the parsimonious $\lambda$-calculus.

\subparagraph{First-order interpretations}

As we already said, regular and polyregular functions both admit logical
characterizations using Monadic Second-Order
Logic~\cite{EngelfrietHoogeboom,polyregularMSO}. The basic conceit behind these
definitions is that a string $w$ may be regarded as a finite model
$\mathfrak{M}(w)$ over a signature containing the order relation $\le$ on
positions and predicates encoding their labeling.

The classes obtained by replacing MSO with first-order logic (FO) are to
(poly)regular functions what star-free languages are to regular languages,
see~\cite{AperiodicSST,polyregular}. We expect that in the same way, replacing
regular functions (i.e.\ MSO transductions) by FO transductions in
\Cref{def:cfpolyreg-cbs} and \Cref{thm:composition-squaring} results in the same
class in both cases, which would then be the natural FO counterpart of cfp
functions. Furthermore, we believe it can be defined logically. Given a finite
model $\mathfrak{U} = (U, R, \ldots)$, we write $\mathfrak{U}^k$ for the
$k^{\text{th}}$ power $(U^k, R_1, \ldots, R_k, \ldots)$ where $R_i(x_1, \ldots,
x_{m})$ of arity $m$ is defined as $R(\pi_i(x_1), \ldots, \pi_i(x_{m}))$ for $1
\le i \le k$.
\begin{conjecture}
\label{conj:fo-interpretations}
A function $f : \Gamma^* \to \Sigma^*$ is \enquote{FO comparison-free
  polyregular} if and only if there exists $k \in \naturalN$ and a
one-dimensional FO interpretation $\varphi$ such that for every
$w \in \Gamma^*$ with $|w| \ge 2$,
there is an isomorphism of structures $\mathfrak{M}(f(w)) \simeq \varphi\left(\mathfrak{M}(w)^k\right)$.
\end{conjecture}

On an intuitive level, this seems to capture the inability to compare the
positions of two heads of comparison-free pebble transducers. However, as
mentioned to us by M.~Bojańczyk, the naive transposition of this conjecture to
MSO fails because the direct product, generalized to Henkin structures,
does not preserve \emph{standard} second-order models.

\subparagraph{Integer sequences}

Recall from Remarks~\ref{rem:unary-square-hdt0l} and~\ref{rem:not-cfp-unary}
that for unary outputs, polyregular and layered HDT0L transductions coincide,
but comparison-free polyregular functions form a strictly smaller class (those
results come from~\cite{Gaetan}). If we also restrict to unary inputs -- in
other words, if we consider sequences $\naturalN\to\naturalN$ -- then we are
fairly confident at this stage that the three classes collapse to a single one,
and that this can be shown by routine methods:
\begin{claim}
  The classes of polyregular, comparison-free polyregular and layered HDT0L
  functions coincide on sequences of natural numbers.
\end{claim}
Note that we already have a description of cfp integer sequences by specializing
\Cref{thm:cfp-seq}.

\subparagraph*{Membership and equivalence}

We presented comparison-free polyregular functions as a strict subclass of
polyregular functions. This leads to a natural \emph{membership problem}, for
which partial results were recently obtained by Douéneau-Tabot~\cite{Gaetan}:

\begin{problem}
Is there an algorithm taking as input a (code for a) pebble transducer
which decides whether the corresponding function $\Sigma^* \to \Gamma^*$
is comparison-free or not?
\end{problem}

There are many similar problems of interest on the frontier between
comparison-free and general polyregular functions. We hope that investigating
such issues
may also lead to machine/syntax-free characterizations of the containment
between the two classes.

Finally, a major open problem on polyregular functions is the \emph{equivalence
  problem}:
\begin{problem}
  Is there an algorithm taking as input two pebble transducers which decides
  whether they compute the same function?
\end{problem}

Interestingly, a positive answer is known for HDT0L transductions. There is an
short proof using Hilbert's basis theorem~\cite{Honkala}, which is now
understood to be an example of a general approach using polynomial grammars (see
e.g.~\cite{PolynomialAutomata,BojanczykPolynomialGrammars}). One could hope that
a restriction to comparison-free pebble transducers also puts the equivalence
problem within reach of known tools. Unfortunately, the extended polynomial
grammars that would serve as the natural target for a reduction from 2-CFPT
equivalence already have an undecidable zeroness problem (this was shown
recently by Schmude~\cite{SchmudePolynomialGrammars}). This does \emph{not}
extend, however, to an undecidability proof for the CFPT equivalence problem, so the latter is still open.

\bibliography{bi}

\begin{thebibliography}{10}

\bibitem{SST}
Rajeev Alur and Pavol {\v{C}}ern{\'{y}}.
\newblock Expressiveness of streaming string transducers.
\newblock In Kamal Lodaya and Meena Mahajan, editors, {\em {IARCS} Annual
  Conference on Foundations of Software Technology and Theoretical Computer
  Science, {FSTTCS} 2010, December 15-18, 2010, Chennai, India}, volume~8 of
  {\em LIPIcs}, pages 1--12. Schloss Dagstuhl - Leibniz-Zentrum für
  Informatik, 2010.
\newblock \href {https://doi.org/10.4230/LIPIcs.FSTTCS.2010.1}
  {\path{doi:10.4230/LIPIcs.FSTTCS.2010.1}}.

\bibitem{PolynomialAutomata}
Michael Benedikt, Timothy Duff, Aditya Sharad, and James Worrell.
\newblock Polynomial automata: {Zeroness} and applications.
\newblock In {\em 2017 32nd {Annual} {ACM}/{IEEE} {Symposium} on {Logic} in
  {Computer} {Science} ({LICS})}, pages 1--12, Reykjavik, Iceland, June 2017.
  IEEE.
\newblock \href {https://doi.org/10.1109/LICS.2017.8005101}
  {\path{doi:10.1109/LICS.2017.8005101}}.

\bibitem{BerstelReutenauer}
Jean Berstel and Christophe Reutenauer.
\newblock {\em {Noncommutative Rational Series with Applications}}, volume 137
  of {\em Encyclopedia of Mathematics and its Applications}.
\newblock {Cambridge University Press}, October 2010.

\bibitem{polyregular}
Mikołaj Bojańczyk.
\newblock Polyregular functions, 2018.
\newblock \href {http://arxiv.org/abs/1810.08760} {\path{arXiv:1810.08760}}.

\bibitem{BojanczykPolynomialGrammars}
Mikołaj Bojańczyk.
\newblock The {Hilbert} method for transducer equivalence.
\newblock {\em {ACM} {SIGLOG} News}, 6(1):5--17, 2019.
\newblock \href {https://doi.org/10.1145/3313909.3313911}
  {\path{doi:10.1145/3313909.3313911}}.

\bibitem{PolyregularGrowth}
Mikołaj Bojańczyk.
\newblock On the growth rate of polyregular functions, 2022.
\newblock \href {http://arxiv.org/abs/2212.11631} {\path{arXiv:2212.11631}}.

\bibitem{Daviaud}
Mikołaj Bojańczyk, Laure Daviaud, and Shankara~Narayanan Krishna.
\newblock Regular and {First}-{Order} {List} {Functions}.
\newblock In {\em Proceedings of the 33rd {Annual} {ACM}/{IEEE} {Symposium} on
  {Logic} in {Computer} {Science} - {LICS} '18}, pages 125--134, Oxford, United
  Kingdom, 2018. ACM Press.
\newblock \href {https://doi.org/10.1145/3209108.3209163}
  {\path{doi:10.1145/3209108.3209163}}.

\bibitem{FOTree}
Mikołaj Bojańczyk and Amina Doumane.
\newblock First-order tree-to-tree functions.
\newblock In Holger Hermanns, Lijun Zhang, Naoki Kobayashi, and Dale Miller,
  editors, {\em {LICS} '20: 35th Annual {ACM/IEEE} Symposium on Logic in
  Computer Science, Saarbrücken, Germany (online conference), July 8-11,
  2020}, pages 252--265. {ACM}, 2020.
\newblock \href {https://doi.org/10.1145/3373718.3394785}
  {\path{doi:10.1145/3373718.3394785}}.

\bibitem{polyregularMSO}
Mikołaj Bojańczyk, Sandra Kiefer, and Nathan Lhote.
\newblock {String-to-String Interpretations With Polynomial-Size Output}.
\newblock In Christel Baier, Ioannis Chatzigiannakis, Paola Flocchini, and
  Stefano Leonardi, editors, {\em 46th International Colloquium on Automata,
  Languages, and Programming (ICALP 2019)}, volume 132 of {\em Leibniz
  International Proceedings in Informatics (LIPIcs)}, pages 106:1--106:14.
  Schloss Dagstuhl--Leibniz-Zentrum fuer Informatik, 2019.
\newblock \href {https://doi.org/10.4230/LIPIcs.ICALP.2019.106}
  {\path{doi:10.4230/LIPIcs.ICALP.2019.106}}.

\bibitem{PolyRecSeq}
Michaël Cadilhac, Filip Mazowiecki, Charles Paperman, Michał Pilipczuk, and
  Géraud Sénizergues.
\newblock On polynomial recursive sequences.
\newblock {\em Theory of Computing Systems}, June 2021.
\newblock \href {https://doi.org/10.1007/s00224-021-10046-9}
  {\path{doi:10.1007/s00224-021-10046-9}}.

\bibitem{Eugenia}
Christian Choffrut.
\newblock Sequences of words defined by two-way transducers.
\newblock {\em Theoretical Computer Science}, 658:85--96, 2017.
\newblock \href {https://doi.org/10.1016/j.tcs.2016.05.004}
  {\path{doi:10.1016/j.tcs.2016.05.004}}.

\bibitem{AperiodicSST}
Luc Dartois, Ismaël Jecker, and Pierre-Alain Reynier.
\newblock Aperiodic {String} {Transducers}.
\newblock {\em International Journal of Foundations of Computer Science},
  29(05):801--824, August 2018.
\newblock \href {https://doi.org/10.1142/S0129054118420054}
  {\path{doi:10.1142/S0129054118420054}}.

\bibitem{Gaetan}
Gaëtan Douéneau-Tabot.
\newblock {Pebble Transducers with Unary Output}.
\newblock In Filippo Bonchi and Simon~J. Puglisi, editors, {\em 46th
  International Symposium on Mathematical Foundations of Computer Science (MFCS
  2021)}, volume 202 of {\em Leibniz International Proceedings in Informatics
  (LIPIcs)}, pages 40:1--40:17, Dagstuhl, Germany, 2021. Schloss Dagstuhl --
  Leibniz-Zentrum f{\"u}r Informatik.
\newblock \href {https://doi.org/10.4230/LIPIcs.MFCS.2021.40}
  {\path{doi:10.4230/LIPIcs.MFCS.2021.40}}.

\bibitem{LastPebble}
Gaëtan Douéneau-Tabot.
\newblock Pebble minimization: the last theorems, 2022.
\newblock \href {http://arxiv.org/abs/2210.02426} {\path{arXiv:2210.02426}}.

\bibitem{Marble}
Gaëtan Douéneau-Tabot, Emmanuel Filiot, and Paul Gastin.
\newblock {Register Transducers Are Marble Transducers}.
\newblock In Javier Esparza and Daniel Kr{\'a}ľ, editors, {\em 45th
  International Symposium on Mathematical Foundations of Computer Science (MFCS
  2020)}, volume 170 of {\em Leibniz International Proceedings in Informatics
  (LIPIcs)}, pages 29:1--29:14, Dagstuhl, Germany, 2020. Schloss
  Dagstuhl--Leibniz-Zentrum f{\"u}r Informatik.
\newblock \href {https://doi.org/10.4230/LIPIcs.MFCS.2020.29}
  {\path{doi:10.4230/LIPIcs.MFCS.2020.29}}.

\bibitem{PebbleComposition}
Joost Engelfriet.
\newblock Two-way pebble transducers for partial functions and their
  composition.
\newblock {\em Acta Informatica}, 52(7-8):559--571, 2015.
\newblock \href {https://doi.org/10.1007/s00236-015-0224-3}
  {\path{doi:10.1007/s00236-015-0224-3}}.

\bibitem{EngelfrietHoogeboom}
Joost Engelfriet and Hendrik~Jan Hoogeboom.
\newblock {MSO} definable string transductions and two-way finite-state
  transducers.
\newblock {\em ACM Transactions on Computational Logic}, 2(2):216--254, April
  2001.
\newblock \href {https://doi.org/10.1145/371316.371512}
  {\path{doi:10.1145/371316.371512}}.

\bibitem{InvisiblePebbles}
Joost Engelfriet, Hendrik~Jan Hoogeboom, and Bart Samwel.
\newblock {XML} navigation and transformation by tree-walking automata and
  transducers with visible and invisible pebbles.
\newblock {\em Theoretical Computer Science}, 850:40--97, January 2021.
\newblock \href {https://doi.org/10.1016/j.tcs.2020.10.030}
  {\path{doi:10.1016/j.tcs.2020.10.030}}.

\bibitem{PebbleString}
Joost Engelfriet and Sebastian Maneth.
\newblock Two-way finite state transducers with nested pebbles.
\newblock In Krzysztof Diks and Wojciech Rytter, editors, {\em Mathematical
  Foundations of Computer Science 2002, 27th International Symposium, {MFCS}
  2002, Warsaw, Poland, August 26-30, 2002, Proceedings}, volume 2420 of {\em
  Lecture Notes in Computer Science}, pages 234--244. Springer, 2002.
\newblock \href {https://doi.org/10.1007/3-540-45687-2_19}
  {\path{doi:10.1007/3-540-45687-2_19}}.

\bibitem{ERS}
Joost Engelfriet, Grzegorz Rozenberg, and Giora Slutzki.
\newblock Tree transducers, {L} systems, and two-way machines.
\newblock {\em Journal of Computer and System Sciences}, 20(2):150--202, 1980.
\newblock \href {https://doi.org/10.1016/0022-0000(80)90058-6}
  {\path{doi:10.1016/0022-0000(80)90058-6}}.

\bibitem{Macro}
Joost Engelfriet and Heiko Vogler.
\newblock Macro tree transducers.
\newblock {\em Journal of Computer and System Sciences}, 31(1):71--146, 1985.
\newblock \href {https://doi.org/10.1016/0022-0000(85)90066-2}
  {\path{doi:10.1016/0022-0000(85)90066-2}}.

\bibitem{FerteMarinSenizergues}
Julien Ferté, Nathalie Marin, and Géraud Sénizergues.
\newblock Word-{Mappings} of {Level} 2.
\newblock {\em Theory of Computing Systems}, 54(1):111--148, January 2014.
\newblock \href {https://doi.org/10.1007/s00224-013-9489-5}
  {\path{doi:10.1007/s00224-013-9489-5}}.

\bibitem{siglog}
Emmanuel Filiot and Pierre-Alain Reynier.
\newblock Transducers, {Logic} and {Algebra} for {Functions} of {Finite}
  {Words}.
\newblock {\em ACM SIGLOG News}, 3(3):4--19, August 2016.
\newblock \href {https://doi.org/10.1145/2984450.2984453}
  {\path{doi:10.1145/2984450.2984453}}.

\bibitem{FiliotReynier}
Emmanuel Filiot and Pierre{-}Alain Reynier.
\newblock Copyful streaming string transducers.
\newblock {\em Fundamenta Informaticae}, 178(1-2):59--76, January 2021.
\newblock \href {https://doi.org/10.3233/FI-2021-1998}
  {\path{doi:10.3233/FI-2021-1998}}.

\bibitem{Guillon}
Bruno Guillon.
\newblock Input- or output-unary sweeping transducers are weaker than their
  2-way counterparts.
\newblock {\em {RAIRO} -- {Theoretical} {Informatics} and {Applications}},
  50(4):275--294, 2016.
\newblock \href {https://doi.org/10.1051/ita/2016028}
  {\path{doi:10.1051/ita/2016028}}.

\bibitem{Honkala}
Juha Honkala.
\newblock A short solution for the {HDT0L} sequence equivalence problem.
\newblock {\em Theoretical Computer Science}, 244(1-2):267--270, 2000.
\newblock \href {https://doi.org/10.1016/S0304-3975(00)00158-4}
  {\path{doi:10.1016/S0304-3975(00)00158-4}}.

\bibitem{PolyregularGrowthBis}
Sandra Kiefer, Lê Thành~Dũng Nguy\~{\^e}n, and Cécilia Pradic.
\newblock Revisiting the growth of polyregular functions: output languages,
  weighted automata and unary inputs, 2023.
\newblock \href {http://arxiv.org/abs/2301.09234} {\path{arXiv:2301.09234}}.

\bibitem{PebbleMinimization}
Nathan Lhote.
\newblock Pebble minimization of polyregular functions.
\newblock In Holger Hermanns, Lijun Zhang, Naoki Kobayashi, and Dale Miller,
  editors, {\em {LICS} '20: 35th Annual {ACM/IEEE} Symposium on Logic in
  Computer Science, Saarbrücken, Germany, July 8-11, 2020}, pages 703--712.
  {ACM}, 2020.
\newblock \href {https://doi.org/10.1145/3373718.3394804}
  {\path{doi:10.1145/3373718.3394804}}.

\bibitem{Lindenmayer}
Aristid Lindenmayer.
\newblock Mathematical models for cellular interactions in development {II}.
  {Simple} and branching filaments with two-sided inputs.
\newblock {\em Journal of Theoretical Biology}, 18(3):300--315, March 1968.
\newblock \href {https://doi.org/10.1016/0022-5193(68)90080-5}
  {\path{doi:10.1016/0022-5193(68)90080-5}}.

\bibitem{Mazza}
Damiano Mazza.
\newblock Simple {Parsimonious} {Types} and {Logarithmic} {Space}.
\newblock In {\em 24th {EACSL} {Annual} {Conference} on {Computer} {Science}
  {Logic} ({CSL} 2015)}, pages 24--40, 2015.
\newblock \href {https://doi.org/10.4230/LIPIcs.CSL.2015.24}
  {\path{doi:10.4230/LIPIcs.CSL.2015.24}}.

\bibitem{Pebble}
Tova Milo, Dan Suciu, and Victor Vianu.
\newblock Typechecking for {XML} transformers.
\newblock {\em Journal of Computer and System Sciences}, 66(1):66--97, 2003.
\newblock Journal version of a PODS 2000 paper.
\newblock \href {https://doi.org/10.1016/S0022-0000(02)00030-2}
  {\path{doi:10.1016/S0022-0000(02)00030-2}}.

\bibitem{MuschollPuppis}
Anca Muscholl and Gabriele Puppis.
\newblock The {Many} {Facets} of {String} {Transducers}.
\newblock In Rolf Niedermeier and Christophe Paul, editors, {\em 36th
  {International} {Symposium} on {Theoretical} {Aspects} of {Computer}
  {Science} ({STACS} 2019)}, volume 126 of {\em Leibniz {International}
  {Proceedings} in {Informatics} ({LIPIcs})}, pages 2:1--2:21. Schloss
  Dagstuhl–Leibniz-Zentrum fuer Informatik, 2019.
\newblock \href {https://doi.org/10.4230/LIPIcs.STACS.2019.2}
  {\path{doi:10.4230/LIPIcs.STACS.2019.2}}.

\bibitem{titoPhD}
Lê Thành~Dũng Nguy\~{\^e}n.
\newblock {\em Implicit automata in linear logic and categorical transducer
  theory}.
\newblock PhD thesis, Université Paris XIII (Sorbonne Paris Nord), December
  2021.
\newblock URL: \url{https://nguyentito.eu/thesis.pdf}.

\bibitem{freeadditives}
Lê Thành D\~ung Nguy{\~{ê}}n, Camille Noûs, and Cécilia Pradic.
\newblock Implicit automata in typed {\(\lambda\)}-calculi {II:} streaming
  transducers vs categorical semantics, 2020.
\newblock \href {http://arxiv.org/abs/2008.01050} {\path{arXiv:2008.01050}}.

\bibitem{aperiodic}
Lê Thành D\~ung Nguy{\~{ê}}n and Cécilia Pradic.
\newblock Implicit automata in typed $\lambda$-calculi {I}: aperiodicity in a
  non-commutative logic.
\newblock In Artur Czumaj, Anuj Dawar, and Emanuela Merelli, editors, {\em 47th
  International Colloquium on Automata, Languages, and Programming, {ICALP}
  2020, July 8-11, 2020, Saarbrücken, Germany (Virtual Conference)}, volume
  168 of {\em LIPIcs}, pages 135:1--135:20. Schloss Dagstuhl - Leibniz-Zentrum
  für Informatik, 2020.
\newblock \href {https://doi.org/10.4230/LIPIcs.ICALP.2020.135}
  {\path{doi:10.4230/LIPIcs.ICALP.2020.135}}.

\bibitem{Reutenauer}
Christophe Reutenauer.
\newblock Sur les séries associées à certains systèmes de {Lindenmayer}.
\newblock {\em Theoretical Computer Science}, 9:363--375, 1979.
\newblock \href {https://doi.org/10.1016/0304-3975(79)90036-7}
  {\path{doi:10.1016/0304-3975(79)90036-7}}.

\bibitem{Sakarovitch}
Jacques Sakarovitch.
\newblock {\em Elements of Automata Theory}.
\newblock Cambridge University Press, 2009.
\newblock Translated by Reuben Thomas.
\newblock \href {https://doi.org/10.1017/CBO9781139195218}
  {\path{doi:10.1017/CBO9781139195218}}.

\bibitem{SchmudePolynomialGrammars}
Janusz Schmude.
\newblock On polynomial grammars extended with substitution, 2021.
\newblock \href {http://arxiv.org/abs/2102.08705} {\path{arXiv:2102.08705}}.

\bibitem{Pumping2NFT}
Tim Smith.
\newblock A pumping lemma for two-way finite transducers.
\newblock In Erzs{\'{e}}bet Csuhaj{-}Varj{\'{u}}, Martin Dietzfelbinger, and
  Zolt{\'{a}}n {\'{E}}sik, editors, {\em Mathematical Foundations of Computer
  Science 2014 - 39th International Symposium, {MFCS} 2014, Budapest, Hungary,
  August 25-29, 2014. Proceedings, Part {I}}, volume 8634 of {\em Lecture Notes
  in Computer Science}, pages 523--534. Springer, 2014.
\newblock \href {https://doi.org/10.1007/978-3-662-44522-8_44}
  {\path{doi:10.1007/978-3-662-44522-8_44}}.

\bibitem{Senizergues}
Géraud Sénizergues.
\newblock Sequences of level 1, 2, 3, ..., \emph{k} , ..
\newblock In Volker Diekert, Mikhail~V. Volkov, and Andrei Voronkov, editors,
  {\em Computer Science - Theory and Applications, Second International
  Symposium on Computer Science in Russia, {CSR} 2007, Ekaterinburg, Russia,
  September 3-7, 2007, Proceedings}, volume 4649 of {\em Lecture Notes in
  Computer Science}, pages 24--32. Springer, 2007.
\newblock \href {https://doi.org/10.1007/978-3-540-74510-5_6}
  {\path{doi:10.1007/978-3-540-74510-5_6}}.

\end{thebibliography}

\newpage

\appendix

\tableofcontents

\section{Details for \Cref{rem:hdt0l-single-state}}
\label{sec:details-hdt0l-single-state}

We recall the \enquote{natural} translation of HDT0L systems into single-state
SSTs, which is relevant to some proofs in \Cref{sec:complements}. Let
$(\Gamma,\Sigma,\Delta,d,(h_c)_{c\in\Gamma},h')$ be a HDT0L system.
It is equivalent to the SST specified by the following data:
\begin{itemize}
\item a singleton set of states: $Q = \{q\}$;
\item the working alphabet as the set of registers: $R = \Delta$ (minor
  technicality: if $\Delta \cap \Sigma \neq \varnothing$, one should take $R$ to
  be a copy of $\Delta$ that is disjoint from $\Sigma$);
\item $h_c \in \Hom(\Delta^*,\Delta^*) \cong (\Delta \to \Delta^*) \subseteq
  (\Delta \to (\Sigma \cup \Delta)^*)$ as the register assignment associated to
  an input letter $c\in \Gamma$ -- in other words, the transition function is
  $\delta : (q,c) \mapsto (q,(h_c)_{\restriction \Delta})$;
\item $(h'(r))_{r \in \Delta} \in (\Sigma^*)^R$ as the initial register values;
\item $F : q \mapsto d$ as the final output function ($d \in \Delta^* \subseteq
  (\Sigma \cup \Delta)^*$).
\end{itemize}
The cases of the transition and output functions involve a codomain extension
from $\Delta^*$ to $(\Sigma \cup \Delta)^*$. This reflects the intuition that a
HDT0L system is the same thing as a single-state SST that \enquote{cannot access
  the output alphabet} (except in the initial register contents).

To prove the equivalence, the key observation is that $h_c$ is turned into
$\delta(-,c)$ by a \emph{morphism} from $\Hom(\Delta^*,\Delta^*)$ to
$\transmono_{\Delta,\varnothing} \wr \{q\} \subset \transmono_{\Delta,\Sigma}
\wr \{q\}$, using the notations from \Cref{sec:transition-monoids}.
We leave the details to the reader.

\section{Proofs for \S\ref{sec:complements}
  (HDT0L systems, SSTs \& polyregular functions)}

\subsection{Proof of \Cref{thm:layered-hdt0l-equiv}}

\layeredhdtOlequiv*

\begin{proof}[Proof of $(\Rightarrow)$]
  The translation from SSTs to HDT0L systems given by
  \cite[Lemma~3.5]{FiliotReynier} turns out to work. It is also formulated in
  terms of \enquote{simple SSTs} (isomorphic to HDT0L systems, cf.\
  \Cref{rem:hdt0l-single-state}) in \cite[\S5.1]{Marble}, where the authors
  remark that \enquote{this construction does not preserve copylessness nor
    $k$-layeredness}: indeed, what we show is that it increments the number of
  layers by one! For the sake of clarity, we give an alternative presentation
  that decomposes it into two steps.

  Let $\Gamma$ be the input alphabet and $\Sigma$ be the output alphabet. Let
  $\cT$ be a SST with a \emph{$k$-layered} set of register variables $R = R_0
  \sqcup \dots \sqcup R_k$. First, we build a \emph{$(k+1)$-layered} SST $\cT'$
  that computes the same function, with the set of registers
  \[ R' = \underline\Sigma \cup R = R'_0 \sqcup \dots \sqcup R'_{k+1} \qquad
    R'_0 = \underline\Sigma \qquad \forall i \in \{1,\ldots,k+1\},\, R'_i =
    R_{i-1} \]
  assuming $\underline\Sigma \cap R = \varnothing$, and whose register
  assignments are \emph{without fresh letters}: the range of every $\alpha' : R'
  \to (\Sigma \cup R')^*$ is included in $R'^*$, which allows us to write
  $\alpha' : R' \to R'^*$. This already brings us closer to the definition of
  HDT0L systems, since $(R \to R^*) \cong \Hom(R^*,R^*)$. Similarly, we will
  ensure that the range of the output function of $\cT'$ is included in $R'$.

  Let $\mathtt{underline}_\Sigma \in \Hom((\Sigma \cup R)^*, (\underline{\Sigma}
  \cup R)^*)$ be defined in the expected way, and note that its codomain is
  equal to $R'^*$. We specify $\cT'$ as follows (and leave it to the reader to
  check that this works):
  \begin{itemize}
  \item the state space $Q$, initial state and state transitions are the same as
    those of $\cT$;
  \item the initial value of $r' \in R'$ is the same as for $\cT$ if $r' \in
    R$, or the single letter $c$ if $r' = \underline{c} \in \underline{\Sigma}$;
  \item every assignment $\alpha : R \to (\Sigma \cup R)$ that appears in some
    transition of $\cT$ becomes, in $\cT'$,
    \[ \alpha' : R'^* \to R'^* \qquad \alpha' : \underline{c} \in
      \underline{\Sigma} \mapsto \underline{c} \qquad \alpha' : r \in R \mapsto
      \mathtt{underline}_\Sigma(\alpha(r)) \]
  \item its output function is $F' = \mathtt{underline}_\Sigma \circ F$
    where $F : Q \to (\Sigma \cup R)^*$ is the output function of $\cT$.
  \end{itemize}
  Thus, the idea is to store a copy of $c \in \Sigma$ in the register
  $\underline{c}$. Since this register may feed in a copyful way all other
  registers (in a SST, there are no restrictions on the use of output alphabet
  letters), it must sit at the lowest layer, hence $R'_0 = \underline\Sigma$ and
  the resulting offset of one layer.

  Next, we turn $\cT'$ into an equivalent HDT0L system with $(k+1)$-layered
  working alphabet
  \[ \Delta = R' \times Q = \Delta_0 \sqcup \ldots \sqcup \Delta_{k+1} \qquad
    \forall i \in \{0,\ldots,k+1\},\,\Delta_i = R'_i \times Q \]

  For $q \in Q$, let $\mathtt{pair}_q \in \Hom(R'^*, \Delta^*)$ be such that
  $\mathtt{pair}_q(r') = (r',q)$ for $r' \in R'$.

  Let $Q = \{q^{(1)}, \ldots, q^{(n)}\}$ be the states of $\cT'$ (which are also
  those of $\cT$), with $q^{(1)}$ being its initial state\footnote{Except for
    that, this enumeration of $Q$ is arbitrary. We write $q^{(i)}$ instead of
    $q_i$ to avoid confusion with the run of an automaton.}. Using the fact that
  $\cT'$ is without fresh letters, let $F' : Q \to R'^*$ be its final output
  function. The initial word of our HDT0L system is then
  \[ d = \mathtt{pair}_{q^{(1)}}\left(F'\left(q^{(1)}\right)\right) \cdot \ldots \cdot
    \mathtt{pair}_{q^{(n)}}\left(F'\left(q^{(n)}\right)\right) \in \Delta^* \]
  From the initial register values $(u_{I,r'})_{r' \in R'} \in (\Sigma^*)^{R'}$
  of $\cT'$, we define the final morphism:
  \[ h' \in \Hom(\Delta^*, \Sigma^*) \qquad \forall r' \in R',\quad
    \left[h'\left(r',q^{(1)}\right) = u_{I,r'}\quad \text{and}\quad \forall q \neq q^{(1)},
      h'(r',q) = \varepsilon \right]\]
  Finally, let $\deltast' : Q \to Q$ and $\deltareg' : Q \to (R'
  \to R'^*)$ be the components of the transition function of $\cT'$.
  The morphisms $h_c \in \Hom(\Delta^*, \Delta^*)$ for $c \in \Gamma$ send
  $(r',q) \in \Delta$ to
  \[ h_c(r',q) = \mathtt{pair}_{q^{(i_1)}}(\deltareg'(q^{(i_1)},c)(r')) \cdot
    \ldots \cdot \mathtt{pair}_{q^{(i_m)}}(\deltareg'(q^{(i_m)},c)(r')) \]
  where $i_1 < \ldots < i_m$ and $\{q^{(i_1)},\ldots,q^{(i_m)}\} = \{q^{(?)} \in
  Q \mid \deltast'(q^{(?)},c) = q \}$.

  Checking that this HDT0L system computes the right function is a matter of
  mechanical verification, that has already been carried out
  in~\cite{FiliotReynier}. To wrap up the proof, we must justify that it
  is $(k+1)$-layered.
  To do so, let us fix a letter $c \in \Gamma$ and two layer indices
  $i,j \in \{0,\ldots,k+1\}$, and count the number $N_{r',q}$ of occurrences of
  $(r',q) \in \Delta_i$ among all the $h_c(\tilde{r}',\tilde{q})$ for
  $(\tilde{r}',\tilde{q}) \in \Delta_j$. The letter $(r',q)$ can only appear in
  $h_c(\tilde{r}',\tilde{q})$ when $\tilde{q} = \delta(q,c)$, and in that case,
  its occurrences (if any) are in the substring
  $\mathtt{pair}_q(\deltareg'(q,c)(\tilde{r}'))$. So $N_{r',q}$ counts the
  occurrences of $r \in R'_i$ among the $\deltareg'(q,c)(\tilde{r}')$ for
  $\tilde{r}' \in R'_j$. Since $\cT'$ is a $(k+1)$-layered SST, we are done.
\end{proof}

\begin{proof}[Proof of $(\Leftarrow)$]
  The translation from HDT0L systems to single-state SSTs mentioned in
  \Cref{rem:hdt0l-single-state} (see \Cref{sec:details-hdt0l-single-state}) is
  not enough: starting from a $(k+1)$-layered HDT0L system, it gives us a
  $(k+1)$-layered SST. But we can bring this down to $k$ layers by adding
  states.

  Let $(\Gamma,\Sigma,\Delta,d,(h_c)_{c\in\Gamma},h')$ be a HDT0L system (with
  $d \in \Delta^*$, $h_c \in \Hom(\Delta^*,\Delta^*)$ for $c \in \Gamma$, and $h
  \in \Hom(\Delta^*,\Sigma^*)$). Suppose that it is $(k+1)$-layered with $\Delta
  = \Delta_0 \sqcup \dots \sqcup \Delta_{k+1}$. This entails that $h_c(\Delta_0)
  \subseteq \Delta_0^*$, and furthermore that $(h_c)_{\restriction\Delta_0} :
  \Delta_0 \to \Delta_0^*$ satisfies a copylessness condition, that may
  succinctly be written as $(h_c)_{\restriction\Delta_0} \in
  \transmonocl_{\Delta_0,\varnothing}$ (cf.\ \Cref{def:transmonocl}).


  We define a $k$-layered SST with:
  \begin{itemize}
  \item $\transmonocl_{\Delta_0,\varnothing}$ as the set of states (finite by
    \Cref{prop:transmonocl-finite}), with the monoid identity as its initial
    state;
  \item the set of registers $R = \Delta \setminus \Delta_0 = \Delta_1 \sqcup
    \dots \sqcup \Delta_{k+1}$, whose $i$-th layer is the $(i+1)$-th layer of
    the original HDT0L system ($0 \leq i \leq k$);
  \item the initial register contents $(h'(r))_{r \in R}$ -- recall that $h'$ is
    the final morphism;
  \item the transition function $(\alpha,c) \mapsto (\alpha \bullet
    (h_c)_{\restriction\Delta_0},\, (h'_{\restriction \Delta_0^*} \circ \alpha)^\odot
    \circ (h_c)_{\restriction R})$ where $(-)^\odot$ extends functions $\Delta_0
    \to \Sigma^*$ into morphisms in $\Hom((\Delta \cup \Sigma)^*, (R \cup
    \Sigma)^*)$ that map each letter in $R \cup \Sigma$ to itself (since $\Delta
    = \Delta_0 \sqcup R$, the domain of these morphisms is $(\Delta_0 \sqcup R
    \sqcup \Sigma)^*$);
  \item the final output function $\alpha \mapsto (h'_{\restriction \Delta_0^*} \circ
    \alpha)^\odot(d)$.
  \end{itemize}
  The layering condition for this SST is inherited is a direct consequence of
  the layering of the original HDT0L system, and one can check the functions
  computed by the two are the same.
\end{proof}

\subsection{Proof of \Cref{cor:regular-single-state-1-layered}}

\regularsinglestateonelayered*

Any regular function is definable by some copyless SST, i.e.\ 0-layered SST. By
\Cref{thm:layered-hdt0l-equiv}, it can be turned into a 1-layered HDT0L system.
The latter can be translated to a single-state SST by the construction of
\Cref{sec:details-hdt0l-single-state}. As can readily be seen from the
definitions, this construction preserves the 1-layered property.

\subsection{Proof of~\Cref{prop:single-state-copyless-weak}}

\singlestatecopylessweak*

Consider any single-state copyless SST computing some $g : \{a,b,c\}^* \to
\{a,b\}^*$ with a set of registers $R$. We wish to show $g$ does not coincide
with the function computed by the sequential transducer of \Cref{fig:mealy}. Let
$\omega \in (\{a,b\} \cup R)^*$ be the image of the single state by the output
function, and for $x \in \{a,b,c\}$, let $\alpha_x : R \to (\{a,b\} \cup R)^*$
be the copyless assignment performed by the SST when it reads $x$ (that is,
using the notations of \Cref{def:sst}, $\omega = F(q)$ and $\alpha_x =
\deltareg(q,x)$ with $Q = \{q\}$). Let $\vec{u}$ be the initial register
contents. Then
\[ \forall x \in \{a,b,c\}^*,\, \forall n \in \naturalN,\; g(x \cdot c^n) = 
  \omega^\dagger \circ (\alpha_c^\dagger)^n \circ \alpha_x^\dagger(\vec{u})
\]
Any register assignment $\beta : R \to (\{a,b,c\} \cup R)^*$ admits a unique
extension into a monoid morphism $\beta^\square \in \Hom((\{a,b,c\} \cup
R)^*,(\{a,b,c\} \cup R)^*)$ that maps every letter in $\{a,b,c\}$ to itself. Let
$\omega_n = \left(\alpha_c^\square\right)^n(\omega)$ (so that $\omega_0 = \omega$). One can
check that, for
all $n \in \naturalN$:
\begin{itemize}
\item $\omega_n^\dagger = \omega^\dagger \circ (\alpha_c^\dagger)^n$;
\item since $\alpha_c$ is copyless, $|\omega_n|_r \leq |\omega|_r$ for all $r
  \in R$, writing $|w|_x$ for the number of occurrences of $x$ in $w \in
  \Sigma^*$ for $x \in \Sigma$.
\end{itemize}
Let $(v_{x,r})_{r \in R} = \alpha_x^\dagger(\vec{u})$ for $x \in \{a,b,c\}$;
that is, $v_{x,r}$ the value stored in the register $r \in R$ after the SST has
read the single letter $x$. We can rewrite the above equation as
\[ \forall x \in \{a,b,c\}^*,\, \forall n \in \naturalN,\; g(x \cdot c^n) = 
  \omega_n^\dagger((v_{x,r})_{r \in R})
\]
and derive a numerical (in)equality
\[ \forall x \in \{a,b,c\}^*,\, \forall n \in \naturalN,\; |g(x \cdot c^n)|_a =
  |\omega_n|_a + \sum_{r \in R} |\omega_n|_r |v_{x,r}|_a
  \underset{n \to +\infty}{=} |\omega_n|_a + O(1)
\]
using the fact that $|\omega_n|_r$, as a non-negative quantity lower than the
constant $|\omega|_r$, is $O(1)$.

From this, it follows that as $n$ increases, the difference between $|g(a\cdot
c^n)|_a$ and $|g(b\cdot c^n)|_a$ stays bounded. This property distinguishes $g$
from the $f : \{a,b,c\}^* \to \{a,b\}^*$ computed by the transducer given in
\Cref{fig:mealy}, since
\[ \forall n\in\naturalN,\quad \left|f(a \cdot c^n)\right|_a =
  \left|a^{n+1}\right|_a = n+1 \quad\text{and}\quad \left|f(b \cdot
    c^n)\right|_a = \left|b^{n+1}\right|_a = 0 \]

\subsection{Proof of \Cref{prop:polyreg-reg-squaring}}

\polyregregsquaring*

It is stated in the introduction to~\cite{polyregular} that all regular
functions are polyregular. One way to see this is discussed in
\Cref{sec:cfpebble}: the characterization by pebble transducers given
in~\cite{polyregular} generalizes the classical definition of regular functions
using two-way finite state transducers. This takes care of one direction of the
equivalence; for the converse, observe that:
\begin{itemize}
\item sequential functions are regular, as already mentioned;
\item since the SST of \Cref{ex:iterated-reverse} is copyless, the iterated
  reverse function is regular.
\end{itemize}

\subsection{Proof of \Cref{thm:polyreg-layered}}

\polyreglayered*

\begin{proof}[Proof of (\ref{thm-item:polyreg-layered-polyreg}) $\Rightarrow$
  (\ref{thm-item:polyreg-layered-1-layered})]
  Thanks to \Cref{prop:polyreg-reg-squaring}, we know that any polyregular
  functions can be written as a composition of a sequence of functions, each of
  which is either regular or equal to $\mathtt{squaring}_\Gamma$ for some finite
  alphabet $\Gamma$. It suffices to show that each function in the sequence can
  in turn be expressed as a composition of single-state 1-layered SSTs.
  
  We decompose $\mathtt{squaring}_\Gamma$ as
  \[ \qquad \mathtt{1234} \mapsto
    \mathtt{\underline{4}321\underline{3}21\underline{2}1\underline{1}} \mapsto
    \mathtt{\underline{1}2341\underline{2}3412\underline{3}4123\underline{4}}
  \]
  The first step is performed by the SST of \Cref{ex:prefixes}, which has a
  single state and, as mentioned in \Cref{sec:polynomial-growth}, is 1-layered. The
  second step can be implemented using a SST with a single state $q$ (that we
  omit below for readability), two registers $X$ (at layer 0) and $Y$ (at layer
  1) with empty initial values, an output function $F(q) = Y$, and
  \[ \forall c \in \Gamma, \quad
    \delta(c) = (X \mapsto X,\, Y \mapsto cY) \quad\text{and}\quad
    \delta(\underline{c}) = (X \mapsto cX,\, Y \mapsto \underline{c}XY) \]
  As for regular functions, \Cref{cor:regular-single-state-1-layered} takes care
  of them.
\end{proof}

\begin{proof}[Proof of (\ref{thm-item:polyreg-layered-1-layered}) $\Rightarrow$
  (\ref{thm-item:polyreg-layered-k-layered})]
  Immediate by definition.
\end{proof}

\begin{proof}[Proof of (\ref{thm-item:polyreg-layered-k-layered}) $\Rightarrow$
  (\ref{thm-item:polyreg-layered-polyreg})]
  All functions computed by $k$-layered SSTs are polyregular; this applies in
  particular to single-state 1-layered SSTs. Therefore, their composition is
  also polyregular (according to \Cref{def:polyreg}, polyregular functions are
  closed under composition).
\end{proof}

\subsection{Proof of \Cref{prop:fmap-hdt0l}}

\fmaphdtOl*

Let $(\Gamma,\Sigma,\Delta,d,(h_c)_{c\in\Gamma},h')$ be a HDT0L system computing
$f : \Gamma^* \to \Sigma^*$. We define below a HDT0L system that computes
$\fmap(f) : (\Gamma\cup\{\#\})^* \to (\Sigma\cup\{\#\})^*$.
\begin{itemize}
\item The intermediate alphabet is $\widehat\Delta = \Delta \cup \Sigma \cup
  \{\#,X\}$, assuming w.l.o.g.\ that $\#\notin\Delta\cup\Sigma$, where
  $X\notin\Delta\cup\Sigma\cup\{\#\}$ is an arbitrarily chosen fresh letter.
\item The starting word is $Xd \in \widehat\Delta^*$.
\item For $c \in \Gamma$, we extend $h_c$ into $\widehat{h}_c \in
  \Hom(\Delta^*,\Delta^*)$ by setting $\widehat{h}_c(x) = x$ for
  $x\in\Sigma\cup\{\#,X\}$. Since the input alphabet is now $\Gamma\cup\{\#\}$,
  we also define the morphism $\widehat{h}_\#$ as the extension of $h'$ (using
  $\Sigma \subset \widehat\Delta$) such that $\widehat{h}_\#(X) = Xd\#$ and
  $\widehat{h}_\#(x) = x$ for $x\in\Sigma\cup\{\#\}$.
\item The final morphism $\widehat{h}'$ extends $h'$ with $\widehat{h}'(X) =
  \varepsilon$ and $\widehat{h}'(x) = x$ for $x\in\Sigma\cup\{\#\}$
\end{itemize}
This shows that HDT0L transductions are closed under $\fmap$.

We now prove that for any $k \in \naturalN_{\geq 1}$, this closure property
holds for $k$-layered HDT0L transductions (so, in particular, for regular
functions by taking $k=1$). Suppose that $f$ is computed by a $k$-layered HDT0L
system with intermediate alphabet $\Delta$ and initial word $d \in \Delta^*$.
One can build a $k$-layered HDT0L system which computes the same function $f$
and such that \emph{the initial word contains at most one occurrence of each
  letter}; the idea is to replace $\Delta$ and $d = d_1 \dots d_n$ by
$\Delta\times\{1,\dots,n\}$ and $(d_1,1),\dots,(d_n,n)$ where $n=|d|$, and to
adapt the morphisms accordingly. Applying the above construction then results in
a $k$-layered HDT0L system that computes $\fmap(f)$; note that if we did not have
this property for the inital word, we would get a $(k+1)$-layering instead.

\section{Proofs for \S\ref{sec:cbs} (composition by substitution)}

\subsection{Proof of \Cref{thm:polyreg-cbs}}

\polyregcbs*

We start by briefly recalling the definition of \emph{polynomial list functions}
from~\cite[Section~4]{polyregular}. The explanation is geared towards a reader
familiar with the simply typed $\lambda$-calculus, which this system extends.
The $\lambda$-terms defining polynomial list functions are generated by the
grammar of simply typed $\lambda$-terms enriched with constants, whose meaning
can be specified by extending the $\beta$-rule. For instance, given a
\emph{finite} set $S$ and $a \in S$, every element of $S$ can be used as a
constant, another allowed constant is $\mathtt{is}_a^S$ and we have
\[ \mathtt{is}^S_a\;b =_\beta \true\quad\text{if $a=b$}\qquad
  \mathtt{is}_a\;b =_\beta \false\quad\text{if $b \in S \setminus \{a\}$}\]
The grammar of simple types and the typing rules are also extended accordingly.
For instance, any finite set $S$ induces a type also written $S$, such
that every element $a \in S$ corresponds to a term $a : S$ of this type.
There are also operations expressing the cartesian product ($\times$) and
disjoint union ($+$) of two types; and, for any type $\tau$, there is a type
$\tau^*$ of lists whose elements are in $\tau$. So the constant
$\mathtt{is}^S_a$ receives the type
\[ \mathtt{is}_a^S : S \to
  \{\true\}+\{\false\}\qquad\text{for any finite set $S$} \]
See~\cite[Section~4]{polyregular} for the other primitive operations that are
added to the simply typed $\lambda$-calculus; we make use of $\mathtt{is}$,
$\mathtt{case}$, $\mathtt{map}$ and $\mathtt{concat}$ here. Bojańczyk's result
is that if $\Gamma$ and $\Sigma$ are finite sets, then the polynomial list
functions of type $\Gamma^* \to \Sigma^*$ correspond exactly the polyregular
functions.

\begin{lemma}
  Let $I = \{i_1, \ldots, i_{|I|}\}$. Then the function $\mathtt{match}^{I,\tau}
  : I \to \tau \to \ldots \to \tau \to \tau$ which returns its $(k+1)$-th
  argument when its 1st argument is $i_k$ is a polynomial list function.
\end{lemma}
\begin{proof}[Proof idea]
  By induction on $|I|$, using $\mathtt{is}^I_i$ ($i \in I$) and
  $\mathtt{case}^{\{\true\},\{\false\},\tau}$.
\end{proof}

\begin{proof}[Proof of closure by CbS]
  Let $f : \Gamma^* \to I^*$, and for $i \in I$, $g_i : \Gamma^* \to \Sigma^*$
  be polyregular functions. Assuming that $f$ and $g_i$ ($i \in I$) are defined
  by polynomial list functions of the same name, the $\lambda$-term
  \[ \lambda w.\; \mathtt{concat}^\Sigma\;(\mathtt{map}^{I,\Sigma^*}\;(\lambda
    i.\; \mathtt{match}^{I,\Sigma^*}\; i\; (g_{i_1}\; w)\;
    \ldots\;(g_{i_{|I|}}\;w))\;(f\;w)) \]
  computes $\CbS(f,(g_i)_{i\in I})$.
\end{proof}

\subsection{Proof of \Cref{prop:CbS-ind}}

\CbSind*

This is equivalent to claiming that the smallest class $\cC$ of functions such that
\begin{itemize}
\item every regular function is in $\cC$,
\item and $\CbS(f, (g_i)_{i \in I}) \in \cC$ for any regular $f : \Gamma^* \to
  I^*$ and any $(g_i : \Gamma^* \to \Sigma^*)_{i \in I} \in \cC^I$,
\end{itemize}
contains all comparison-free polyregular functions. It suffices to show that
$\cC$ is closed under composition by substitution, which can be done by
induction using the equation
\[\CbS(\CbS(f, (g_i)_i), (h_j)_j) = \CbS(f, (\CbS(g_i, (h_j)_j))_i)\]
The same equation explains the inequality on the rank that we claim in the
proposition.

\subsection{Proof of \Cref{prop:cfp-conditionals}}

\cfpconditionals*

\subparagraph*{Closure under regular conditionals}
We first observe that the particular case where $f$ and $g$
are both regular ($\rk(f)=\rk(g)=0$) already appears in the
literature~\cite{SST}. We shall use it in further appendices,
so let us state it as a stand-alone lemma. 
\begin{lemma}[{\cite[Proposition~2]{SST}}]
  \label{lem:regular-conditionals}
  Let $f, g : \Gamma^* \to \Sigma^*$ be regular functions and $L \subseteq
  \Gamma^*$ be a regular language. The function that coincides with $f$ on $L$
  and with $g$ on $\Gamma^* \setminus L$ is regular.
\end{lemma}

With this in hand, let us turn to the general case where $f, g : \Gamma^* \to \Sigma^*$
are cfp. It means that we have regular functions $f'$ and $g'$,
as well as families of functions $(f''_i)_{i \in I}$ and $(g''_j)_{j \in J}$
such that
\[ f ~~=~~\CbS(f', (f''_i)_{i \in I}) \qquad\text{and}\qquad g ~~=~~\CbS(g', (g''_j)_{j \in J})\]
(if $f$ (or $g$) is of rank $0$, we can introduce a spurious $\CbS$ by taking $f'$
(resp. $g'$) to be a constant function outputting a single letter over the singleton
alphabet). Assume without loss of generality that $I \cap J = \varnothing$.
Using \Cref{lem:regular-conditionals} applied to $f',g' : \Gamma^* \to (I \cup J)^*$,
there is a function $h' : \Gamma^* \to (I \cup J)^*$ coinciding with $f'$ over $L$
and $g'$ over $L \setminus \Gamma^*$. Setting $(h''_k)_{k \in I \cup J}$ to be the
family of functions such that $h''_i = f''_i$ for $i \in I$ and $h''_j = g''_j$ for
$j \in J$, we obtain a cfp function $h = \CbS(h',(h''_k)_{k \in I \cup J})$
corresponding to the desired conditional.

\subparagraph*{Closure under concatenation}
Similarly, first observe that the result holds for regular functions.
Using SSTs, this can be shown using a product construction.

Then, taking regular functions $f'$ and $g'$, as well as families of
functions $(f''_i)_{i \in I}$ and $(g''_j)_{j \in J}$ so that
$f =\CbS(f', (f''_i)_{i \in I})$ and $g =\CbS(g', (g''_j)_{j \in J})$
with $I$ and $J$ disjoint as above, one check that the pointwise concatenation
$f \cdot g$ is equal to $\CbS(f' \cdot g',(h''_k)_{k \in I \cup J})$, which is
cfp since $f' \cdot g'$ is regular.

\section{Proofs for \S\ref{sec:cfpebble}
  (comparison-free pebble transducers)}

\subsection{Proof of \Cref{prop:cfpebble-easy}}

\cfpebbleeasy*

First note that any $k$-CFPT can be transformed into
an equivalent $k$-CFPT whose transition functions
$ \delta :  Q \times (\Gamma \cup \{\leftmarker,\rightmarker\})^p \to Q \times
    (\naturalN^p \to \Stack{k}) \times \Sigma^*$
are such that, for every input $(q,\vec b)$, we have either $\pi_3(\delta(q, \vec b)) = \varepsilon$ (in which case we call $\delta(q,\vec b)$ a \emph{silent transition}) or
$\pi_3(\delta(q,\vec b)) \in \Sigma$ and $\pi_2(\delta(q,\vec b))$ is the identity.
So, without loss of generality, suppose that we have a $k$-CFPT $\cT_f$ implementing $f$ is of this shape,
with state space $Q_f$ and transition function $\delta_f$
Similarly, we may assume without loss of generality that the current height of the stack
is tracked by the state of CFPTs if we allow multiple final states; assume that we have
such height-tracking $l$-CFPT and that we have $l$-CFPTs $\cT_i$ implementing $g_i$ with state spaces $Q_i$
and transition functions $\delta_i$.

We combine these CFPTs into a single $k+l$ CFPT $\cT'$ with state space
\[Q' ~~~=~~~ Q_f ~~ \sqcup ~~ Q_f \times \bigsqcup_{i \in I} Q_i\]
The initial and final states are those of $\cT_f$.
The high-level idea is that $\cT'$ behaves as $\cT_f$ until it produces an output $i \in I$; in such
a case it ``performs a call'' to $\cT_i$ that might spawn additional heads to perform its computations.
At the end of the execution of $\cT_i$, we return the control to $\cT_f$.
Formally speaking, the transition function $\delta'$ of $\cT'$ behaves as follows:
\begin{itemize}
\item $\delta'(q,\vec b) = \delta_f(q,\vec b)$ if $q \in Q_f$ and $\delta_f(q,\vec b)$ is \emph{silent}.
\item otherwise we take, we have $\pi_3(\delta_f(q, \vec b)) = i$ for some $i \in I$.
Calling $r_i$ the initial state of $\cT_i$, we set
$\pi_1(\delta'(q,\vec b)) = (q, r_i)$ and $\pi_2(\delta'(q,\vec b))$ corresponds to push a new pebble onto the stack.
We make $\delta'(q, \vec b)$ silent in such a case.
\item $\delta'((q,r), \vec b \vec b')$ then corresponds to $\delta_i(r, \vec b')$ if we are not in the situation
where the stack height is $1$ and the stack update function is pop.
\item otherwise we take $\pi_1(\delta'((q,r), \vec b b')) = \pi_1(\delta_f(q,\vec b))$,
$\pi_2(\delta'(q,n+1, \vec b b'))$ to be a pop action and $\pi_3(\delta'((q,r), \vec b b')) = \pi_3(\delta_i((q,r),\vec b b'))$.
\end{itemize}

\subsection{Proof of \Cref{thm:cfpebble-hard}}

\cfpebblehard*

Assume we have $f : \Gamma^* \to \Sigma^*$ computed by a $k$-CFPT $\cT$ with state space $Q$ and
transition function $\delta$ that we assume to be disjoint from $\Sigma$.
For each $q \in Q$, we describe a $k-1$ CFPT $\cT_q$ with the same state space,
initial state $q$ and transition function $\delta_q$ such that,
for every $b' \in \naturalN$, $\vec b \in \Stack{l}$ for $l \le k-1$ and $q' \in Q$,
$\delta(q',b'\vec b)$ and $\delta_q(q',\vec b)$ coincide on the first and last component;
on the second component, we require they also coincide up to the difference in stack size.
If we fix $r \in Q$, by~\cite[Theorem~4.7]{Pebble}, the language consisting of those $w \in \Gamma^*$
such that $\cT_q$ halts on $r$ is regular. Since regular languages are closed under intersection,
for any map $\gamma \in Q^Q$, the language $L_\gamma \subseteq \Gamma^*$ of those words $w$
such that $\cT_q$ halts on $\gamma(q)$ is regular.

Now fix $\gamma \in Q^Q$ and let us describe a $1$-CFPT transducer $\cT_\gamma$
intended to implement the restriction of a function $h : \Gamma^* \to (\Sigma \cup Q)^*$
to $L_\gamma$. $\cT_\gamma$ has the same state space and initial state as $\cT$,
but has a transition function $\delta_\gamma$ defined by
\[
\delta_\gamma(q,b) = \left\{
\begin{array}{ll}
\delta(q,b) & \text{if $\pi_2(\delta(q,b))$ is not a push} \\
(\gamma(r), (p \mapsto p), r) & \text{otherwise, for $r = \pi_1(\delta(q,b))$} \\
\end{array}\right.
\]
Since $\Gamma^* = \bigcup_{\gamma \in Q^Q} L_\gamma$,
by applying repeatedly~\Cref{lem:regular-conditionals}, this
determines the regular function $h : \Gamma^* \to (\Sigma \cup Q)^*$.
We can then check that $f = \CbS(h, (g_i)_{i \in \Sigma \cup Q})$ where
$g_a$ is the constant function outputting the one-letter word $a$ for $a \in \Sigma$
(which can certainly be implemented by a $1$-CFPT)
and $g_q$ is the function $\Gamma^* \to \Sigma^*$ implemented by the $(k-1)$-CFPT
$\cT_q$.


\subsection{Proof of \Cref{cor:cfpebble}}

\corcfpebble*

The proof goes by induction over $k \in \naturalN$.
By \Cref{thm:reg-2dft}, the result holds for $k = 0$ since 2DFTs
characterize regular functions; let us detail each direction of the inductive case $k>0$:
\begin{itemize}
\item for the left-to-right inclusion, assume we are given a $(k+1)$-CFPT computing $f$
and apply \Cref{thm:cfpebble-hard} to obtain $h$ and $g_i$s such that $f = \CbS(h,(g_i)_{i \in I})$
with $h$ regular and the $g_i$s computable by $k$-CFPTs. The induction hypothesis implies
that the $g_i$s have rank $< k$, and thus $f$ has rank $\le k$.
\item conversely, if $f$ has rank $k$, it can be written as $\CbS(h,(g_i)_{i \in I})$ with
$h$ regular and the $g_i$s with rank $< k$; the induction hypothesis implies that the
$g_i$s can be computed by $k$-CFPTs. By \Cref{thm:reg-2dft}, $h$ is computable by a
$1$-CFPT, so by \Cref{prop:cfpebble-easy}, $f$ is computed by a $(k+1)$-CFPT
\end{itemize}

\section{Closure under composition}
\label{sec:proof-composition}

This section is dedicated to establishing the following theorem, which
constitutes most of the work that goes into proving
\Cref{thm:composition-squaring}.
\begin{theorem}
  \label{thm:composition}
  Comparison-free polyregular functions are closed under composition.
\end{theorem}

First, by induction on the rank of the left-hand side of the composition, we can
reduce to the case where that side is a mere regular function, using the
straightforward identity
\[ \CbS(f, (g_i)_{i \in I}) \circ h =
   \CbS(f \circ h, (g_i \circ h)_{i \in I}) \]
We then treat this case by another induction, this time on the rank of the
right-hand side. The base case is handled by invoking the closure under
composition of regular functions.
Therefore, what remains is the following inductive case.
\begin{lemma}
  \label{lem:composition-inductive-case}
  Let $f : \Gamma^* \to I^*$ be a regular function and let $(g_i)_{i \in I}$ be
  a family of comparison-free polyregular functions $\Gamma^* \to \Sigma^*$.
  Suppose that for all regular $h : \Sigma^* \to \Delta^*$ and all $i \in I$,
  the composite $h \circ g_i$ is comparison-free polyregular.

  Then, for all regular $h : \Sigma^* \to \Delta^*$, $h \circ \CbS(f,(g_i)_{i \in
    I})$ is comparison-free polyregular.
\end{lemma}

Our proof of the above lemma and related subclaims rely on the properties of
transition monoids introduced in \Cref{sec:transition-monoids} and on the
combinatorics of register transitions discussed in the paragraphs following
\Cref{prop:shape-plus-labels}.

\begin{lemma}
  \label{lem:shape-plus-labels}
  Let $\delta$ be the transition function of some copyless SST $\Sigma^* \to
  \Delta^*$ whose sets of states and registers are $Q$ and $R$ respectively, so
  that $\delta(-,c) \in \transmonocl_{R,\Delta} \wr Q$ for $c \in \Sigma$.
  Let
  \[ \qquad \psi_\delta \in \Hom(\Sigma^*,\, \transmonocl_{R,\Delta} \wr Q)
    \quad\text{such that}\quad \forall c \in \Sigma,\; \psi_\delta(c) =
    \delta(-,c) \]
  and $\varphi_\delta = \tterase_\Delta \circ \psi_\delta$ as in
  \Cref{rem:substitution-transition-monoid}, $q \in Q$, $r \in R$, $\alpha
  \in \transmonocl_{R,\varnothing}$ and $j \in \{0,\dots,|\alpha(r)|\}$. Then
  the following function $\Sigma^* \to \Delta^*$, defined thanks to
  \Cref{prop:shape-plus-labels}, is regular:
  \[ s \quad\mapsto\quad
    \begin{cases}
      w_j \quad\text{where}\quad \pi_2(\psi_\delta(s)(q))(r) = w_0 r'_1 w_1 \dots
      r'_n w'_n  &\text{if}\ \pi_2(\varphi_\delta(s)(q)) = \alpha\\
      \varepsilon &\text{otherwise}
    \end{cases}
  \]
  (recall that $\pi_2 : Q \times M \to M$ is the second projection and $M \wr Q
  = Q \to Q \times M$).
\end{lemma}
\begin{proof}
  We consider during this proof that the names $q$, $r$, $\alpha$ and $j$
  introduced in the above statement are \emph{not in scope}, so that we can use
  those variable names for generic elements of $Q$, $R$,
  $\transmonocl_{R,\varnothing}$ and $\naturalN$ instead. Those data will be
  given other names when we need them.

  We build a copyless SST whose set of states is $Q \times
  \transmonocl_{R,\varnothing}$. This is made possible by the finiteness of
  $\transmonocl_{R,\varnothing}$ (\Cref{prop:transmonocl-finite}). As for the
  set of registers, we would like it to \emph{vary depending on the current
    state} for the sake of conceptual clarity, i.e.\ to have a family of finite
  sets indexed by $Q \times \transmonocl_{R,\varnothing}$; when the SST moves
  from state $(q,\alpha)$ to $(q',\alpha')$, it would perform a register
  assignment from $R_{q,\alpha}$ to $R_{q',\alpha'}$ (described by a map
  $R_{q',\alpha'} \to (\Delta \cup R_{q,\alpha})^*$). Such devices have been
  called \emph{state-dependent memory copyless SSTs} in~\cite{freeadditives},
  and they are clearly equivalent in expressive power to usual copyless SSTs.

  The idea is that we want the configuration (current state plus register
  contents) of our new SST, after reading $s = s_1 \dots s_n$, to faithfully
  represent
  \[ \psi_\delta(s)(q_0) =(\delta(-,s_1) \bullet \dots \bullet
    \delta(-,s_n))(q_0) \in Q \times \transmonocl_{R,\Delta} \]
  where $\delta$ and $\psi_\delta$ are given in the lemma statement, and $q_0$
  is the given state that was called $q$ in that statement. Following
  \Cref{prop:shape-plus-labels}, since we already have the \enquote{shape}
  stored in the second component $\transmonocl_{\Delta,\varnothing}$ of the set
  $Q \times \transmonocl_{\Delta,\varnothing}$ of new states, it makes sense to
  use the register to store the \enquote{labels}, hence $R_{q,\alpha} =
  R_{\alpha}$ with
  \[ R_{\alpha} = \left\{\left(r,j\right) \;\middle|\;
      r \in R,\; j \in \{0,\ldots,|\alpha(r)|\} \right\}
      \quad\text{so that}\quad (\Delta^*)^{R_{\alpha}} \cong
    \prod_{r \in R} (\Delta^*)^{|\alpha(r)| + 1} \]
  The configurations of our SST are thus in bijection with $Q \times
  \transmonocl_{R,\Delta}$ via \Cref{prop:shape-plus-labels}, and we would like
  the transition performed when reading $c \in \Sigma$ to correspond through
  this bijection to (using the notations of \Cref{def:sst})
  \[ (q,\beta) \in Q \times \transmonocl_{R,\Delta} \quad\mapsto\quad
    (\deltast(q),\, \beta \bullet \deltareg(q))\]
  For a fixed $\beta' \in \transmonocl_{R,\Delta}$, let us consider the right
  multiplication $\beta \mapsto \beta \bullet \beta'$ in
  $\transmonocl_{R,\Delta}$. Since $\tterase_\Delta :
  \transmonocl_{R,\Delta} \to \transmonocl_{R,\varnothing}$ is a morphism, the
  \enquote{shape} of $\beta \bullet \beta'$ can be obtained from the
  \enquote{shape} of $\beta$ by multiplying by $\alpha' =
  \tterase_\Delta(\beta')$. The important point is to show that we can
  obtain the new labels from the old ones by a \emph{copyless assignment} --
  formally speaking, that for any $\alpha \in \transmonocl_{R,\Delta}$ there
  exists a copyless
  \[ \gamma_{\alpha,\beta'} : R_{\alpha \bullet \alpha'} \to (\Delta \cup
    R_{\alpha})^* \]
  such that for any $\beta \in \transmonocl_{R,\Delta}$ such that
  $\tterase_\Delta(\beta) = \alpha$, which therefore corresponds to
  \[ \left(\alpha,\vec\ell\right)\quad\text{for some}\quad
    \vec\ell \in (\Delta^*)^{R_{\alpha}} \cong \prod_{r \in R}
    (\Delta^*)^{|\alpha(r)|+1} \]
  the shape-label pair that corresponds to $\beta \bullet \beta'$ is $(\alpha
  \bullet \alpha',\, \gamma_{\alpha,\beta'}^\dagger(\vec\ell))$ (cf.\
  \Cref{def:register-assignment}).

  Our next task is to analyze the composite assignment $\beta \bullet \beta'$ in
  order to derive a $\gamma_{\alpha,\beta'}$ that works. Let $r'' \in R$. First,
  if $\alpha'(r'') = r'_1 \dots r'_n \in R^*$, then
  \[ \beta'(r'') = w'_0 r'_1 w'_1 \dots r'_n w'_n \quad\text{for some}\quad
    w'_0,\ldots,w'_n \in \Delta^* \]
  and by applying the unique morphism $\beta^\odot \in \Hom((\Delta \cup
  R)^*,(\Delta \cup R)^*)$ that extends $\beta$ and sends letters of $\Delta$ to
  themselves, we have
  \[ (\beta\bullet\beta')(r'') = \beta^\odot(\beta'(r)) = w'_0
    \cdot\beta(r'_1)\cdot w'_1\cdot \ldots \cdot\beta(r'_n)\cdot w'_n \]
  Let us decompose further, for $i \in \{1,\ldots,n\}$:
  \[ \beta(r'_i) = w_{i,0} r_{i,1} w_{i,1} \dots w_{i,n_i} r_{n_i}
    \quad\text{for some}\quad w_{i,0},\ldots,w_{i,n_i} \in \Delta^*\]
  By plugging this into the previous equation, we have $(\beta\bullet\beta')(r'')
  = w_0 r_1 w_1 \dots r_m w_m$ where
  \[ \{r_1,\dots,r_m\} = \bigcup_{i=1}^n \{r_{i,1},\dots,r_{i,n_i}\} \]
  Furthermore, \emph{each $w_k$ for $k \in \{0,\dots,m\}$ is a concatenation of
    some $w'_i$ and some $w_{i,j}$}, and from the formal expression of $w_k$
  depending on these $w'_i$ and $w_{i,j}$ -- which only depends on the shape
  $\alpha$ and $\alpha'$ -- we can derive a definition of
  $\gamma_{\alpha,\beta'}(r'',k)$. For instance,
  \[ w_{42} = w_{3,2} w'_3 w_{4,0} \qquad\rightsquigarrow\qquad \gamma(r'',42) =
    (r'_3, 2) \cdot w'_3 \cdot (r'_4, 0) \in (\Delta \cup
    R_{\alpha\bullet\alpha'})^* \]
  Observe that this does not refer to the $w_{i,j}$; therefore,
  $\gamma_{\alpha,\beta'}$ does not depend on $\beta$, as required. One can
  check that defined this way, $\gamma_{\alpha,\beta'}$ is indeed a copyless
  assignment and that the desired property of $\gamma_{\alpha,\beta'}^\dagger$
  holds.

  What we have just seen is the heart of the proof. We leave it to the reader to
  finish the construction of the copyless SST.
\end{proof}


With this done, we can move on to proving \Cref{lem:composition-inductive-case},
which suffices to finish the proof of \Cref{thm:composition}.

\begin{proof}[Proof of \Cref{lem:composition-inductive-case}]
  Let $w \in \Gamma^*$ be an input string. In the composition, we feed to a
  copyless SST $\cT_h$ that computes $h$ the word $\CbS(f,(g_i)_{i \in I})(w) =
  g_{i_1}(w) \ldots g_{i_k}(w)$ where $f(w) = i_1 \ldots i_k$. A first idea is
  therefore to tweak $\cT_h$ into a new copyless SST that takes $I^*$ as input
  and which executes, when it reads $i \in I$, the transition of $\cT_h$ induced
  by $g_i(w)$. If we call $h'_w$ the regular function computed by this new SST,
  we would then have $h'_w(f(w)) = h \circ \CbS(f,(g_i)_{i \in I})(w)$. The
  issue is of course that $h'_w$ depends on the input $w$.

  More precisely, the data that $h'_w$ depends on is the family of transitions
  \[ \qquad(\psi_\delta \circ g_i(w))_{i \in I} \in
    (\transmonocl_{R,\Delta} \wr Q)^I\qquad\text{(see
      \Cref{lem:shape-plus-labels} for $\psi_\delta$)} \]
  where $Q$, $R$ and $\delta$ are respectively the set of states, the set of
  registers and the transition function of $\cT_h$. We will be able to
  disentangle this dependency by working with
  \[ (\varphi_\delta \circ g_i(w))_{i \in I} = (\tterase_\Delta \circ
    \psi_\delta \circ g_i(w))_{i \in I} \in (\transmonocl_{R,\varnothing} \wr
    Q)^I \]
  Concretely:
  \begin{claim}
   For each $\vec\mu \in (\transmonocl_{R,\varnothing} \wr Q)^I$, there exist:
  \begin{itemize}
  \item a finite alphabet $\Lambda_{\vec\mu}$ equipped with a function
    $\iota_{\vec\mu} : \Lambda_{\vec\mu} \to I$;
  \item a regular function $h''_{\vec\mu} : I^* \to (\Delta \cup
    \Lambda_{\vec\mu})^*$;
  \item and regular functions $l_\lambda : \Sigma^* \to \Delta^*$ for $\lambda
    \in \Lambda$;
  \end{itemize}
  such that for $i_1 \dots i_n \in I^*$ and $w \in \Gamma^*$, \emph{if
    $(\varphi_\delta \circ g_i(w))_{i \in I} = \vec\mu$, then}
  \[ h(g_{i_1}(w) \cdot\ldots\cdot g_{i_n}(w)) ~~=~~ \text{replace each}\
    \lambda \in \Lambda_{\vec\mu}\ \text{in}\ h''_{\vec\mu}(i_1 \dots i_n)\ \text{by}\
    l_\lambda \circ g_{\iota(\lambda)}(w)
  \]
  \end{claim}
  \begin{claimproof}
  \Cref{prop:shape-plus-labels} says that every $\beta = \psi_\delta(g_i(w)) \in
  \transmonocl_{R,\Delta}$ can be decomposed into a shape $\alpha =
  \tterase_\Delta(\beta) \in \transmonocl_{R,\varnothing}$ and a finite
  family $\vec\ell$ of strings in $\Delta^*$. Each $\beta(r)$ for $r \in R$ can
  then be reconstituted as an interleaving of letters in $\alpha(r)$ with labels
  in $\vec\ell$, a process that can be decomposed into two steps:
  \begin{itemize}
  \item first, interleave the letters of $\alpha(r)$ with placeholder letters,
    taken from an alphabet disjoint from both $\Delta$ and $R$;
  \item then \emph{substitute} the labels for those letters.
  \end{itemize}
  Roughly speaking, this will allow us to manipulate an assignment with
  placeholders without knowing the labels, and then add the labels afterwards.
  
  Let $\vec\mu \in (\transmonocl_{R,\varnothing} \wr Q)^I$. We define a copyless
  SST $\cT_{\vec\mu}$ with the same sets of states and registers as $\cT_h$,
  namely $Q$ and $R$. Its initial register values and final output function are
  also the same. It computes a function $I^* \to (\Delta \cup
  \Lambda_{\vec\mu})^*$, and its transition function is
  \[ \delta_{\vec\mu} : (q,i) \mapsto \left(\pi_1 \circ \mu_i(q),\; \left(r
        \mapsto \mathtt{interleave}\left(\lambda^{q,i,r}_0 \dots
          \lambda^{q,i,r}_{|\pi_2(\mu_i(q))(r)|},\,
          \pi_2(\mu_i(q))(r)\right)\right)\right) \]
  where $\mathtt{interleave}(u_0 \dots u_n, v_1 \dots v_n) = u_0 v_1 u_1 \dots
  v_n u_n$ for letters $u_0, \dots, u_n, v_1, \dots, v_n$ over some alphabet
  (recall also that $\mu_i : Q \to Q \times \transmonocl_{R,\varnothing}$ for
  $i\in I$). Thus, we take
  \[ \Lambda_{\vec\mu} = \left\{\lambda^{q,i,r}_j \;\middle|\; q \in Q,\, i \in
      I,\, r \in R,\, j \in \{0, \dots, |\pi_2(\mu_i(q))(r)|\} \right\} \qquad
    \iota(\lambda^{q,i,r}_j) = i \]
  and $h''_{\vec\mu}$ to be the function computed by $\cT_{\vec\mu}$. (Note that
  although $\delta_{\vec\mu}$ does not involve letters from $\Delta$, the final
  output function and the initial register contents do.) Finally, given $\lambda
  = \lambda^{q,i,r}_j \in \Lambda_{\vec\mu}$, we define $l_\lambda$ to be the
  regular function provided by \Cref{lem:shape-plus-labels} for the transition
  function $\delta$ of $\cT_h$, the state $q_0$ (which is the initial state of
  both $\cT_h$ and $\cT_{\vec\mu}$), the register $r$, the assignment shape
  $\alpha = \pi_2(\mu_i(q))$ and the position $j \in \{0,\ldots,|\alpha(r)|\}$.

  Let $w \in \Gamma^*$ be such that $(\varphi_\delta \circ g_i(w))_{i \in I} =
  \vec\mu$. Consider $\chi_w \in \Hom((\Delta \cup \Lambda_{\vec\mu})^*,
  \Delta^*)$ which maps each letter of $\Delta$ to itself and each $\lambda \in
  \Lambda_{\vec\mu}$ to $l_\lambda \circ g_{\iota(\lambda)}(w)$. It lifts to a
  morphism $\widehat{\chi_w} \in \Hom(\transmonocl_{R,\Delta \cup
    \Lambda_{\vec\mu}}, \transmonocl_{R,\Delta})$, and we have
  $\widehat{\chi_w}(\delta_{\vec\mu}(-,i)) = \psi_\delta \circ g_i(w)$. This
  leads to the following invariant: the configuration of $\cT_h$ after reading
  $g_{i_1}(w) \cdot \ldots \cdot g_{i_n}(w)$ is, in a suitable sense, the
  \enquote{image by $\chi_w$} of the configuration of $\cT_{\vec\mu}$ after
  reading $i_1 \ldots i_n$. (In other words, the \enquote{image of the SST
    $\cT_{\vec\mu}$ by $\chi_w$} is the copyless SST computing $h'_w$ that we
  sketched at the very beginning of this proof of
  \Cref{lem:composition-inductive-case}.) This directly implies the property
  relating $h$, $h''_{\vec\mu}$ and $(l_\lambda)_{\lambda \in \Lambda_{\vec\mu}}$ that
  we wanted.
  \end{claimproof}
  
  Let us finish proving \Cref{lem:composition-inductive-case} using the fact we
  just proved.
  First of all, since the letters of $\Lambda_{\vec\mu}$ only serve as
  placeholders to be eventually substituted, they can be renamed at our
  convenience. That means that we can take the $\Lambda_{\vec\mu}$ to be
  \emph{disjoint} for $\vec\mu \in (\transmonocl_{R,\varnothing} \wr Q)^I$, and
  define $\Lambda$ to be their disjoint union. We also take $\iota : \Lambda \to
  I$ to be the unique common extension of the $\iota_{\vec\mu}$. In the same
  spirit, we glue together the functions $h''_{\vec\mu} \circ f$ into
  \[ H : w \in \Gamma^* \mapsto h''_{(\varphi_\delta \circ g_i(w) \mid i \in
      I)}(f(w)) \in (\Delta \cup \Lambda)^* \]
  From the above equation on $h''_{\vec\mu}$, one can then deduce for \emph{all}
  $w \in \Gamma^*$ \emph{without condition} that
  \[ h(\CbS(f,(g_i)_{i \in I})(w)) = \CbS(H, (l_\lambda \circ
    g_{\iota(\lambda)})_{\lambda \in \Lambda})(w) \]
  (strictly speaking, one should have a family indexed by $\Delta \cup \Lambda$
  on the right-hand side -- to comply with that, just extend the family with
  constant functions equal to $x$ for each $x \in \Delta$).

  Using the above equation, we can rephrase our goal: we want to prove that the
  function $\CbS(H, (l_\lambda \circ g_{\iota(\lambda)})_{\lambda \in \Lambda})$
  is comparison-free polyregular. This class of functions is -- by definition --
  closed under composition by substitution, so we can reduce this to the
  following subgoals:
  \begin{itemize}
  \item $H$ is comparison-free polyregular: in fact, it is regular, because
    regular functions are closed under composition and regular conditionals
    (\Cref{lem:regular-conditionals}). This argument relies on the finiteness of
    the indexing set $(\transmonocl_{R,\varnothing} \wr Q)^I$ -- a consequence
    of \Cref{prop:transmonocl-finite} -- and on the regularity of the language
    $\{w \in \Gamma^* \mid (\varphi_\delta \circ g_i(w))_{i \in I} = \vec\mu \}$
    for any $\vec\mu$. The reasons for the latter are as follows:
    \begin{itemize}
    \item $\varphi_\delta$ is a morphism whose codomain
      $\transmonocl_{R,\varnothing} \wr Q$ is finite, so
      $\varphi_\delta^{-1}(\{\mu_i\})$ is regular for $i \in I$;
    \item the functions $g_i$ for $i \in I$ are assumed to be comparison-free
      polyregular, so they preserve regular languages by inverse image, as all
      polyregular functions do~\cite{polyregular};
    \item regular languages are closed under finite intersections, and $I$ is
      finite.
    \end{itemize}
  \item $l_\lambda \circ g_{\iota(\lambda)}$ is comparison-free polyregular for
    all $\lambda \in \Lambda$: because our main existence claim states that
    $l_\lambda$ is regular for all $\lambda \in \Lambda$, and one of our
    assumptions is that any $g_i$ (for $i \in I$) postcomposed with any regular
    function gives us a comparison-free polyregular function.
  \end{itemize}
\end{proof}

\section{A lower bound on growth from the rank}
\label{sec:appendix-minimization}

In this section, we prove a statement that directly implies the last claim of
\Cref{thm:characterization-rank}:
\begin{theorem}
  \label{thm:minimization}
  Let $f : \Gamma^* \to \Sigma^*$ be comparison-free polyregular of rank at
  least 1. Then there exists a sequence of inputs $(s_n)_{n \in \naturalN} \in
  (\Gamma^*)^\naturalN$ such that $|s_n| = O(n)$ and $|f(s_n)| \geq
  n^{\rk(f)+1}$.
\end{theorem}
Let us start by proving the lemmas stated in the main text which this theorem
depends on.

\subsection{Proofs for the lemmas in \Cref{sec:minimization}}

Recall that the notion of $r$-split has been defined in \Cref{def:r-split}.

\producingtriples*
\begin{proof}[Proof idea]
  We reuse an idea from~\cite{PebbleMinimization}, but instead of using
  transition monoids of two-way transducers, we rely on monoids of copyless
  register assignments. We shall use the notations introduced in
  \Cref{sec:transition-monoids} for these monoids and the operations they
  support.
  
  Let $R$ and $\Sigma$ be finite alphabets. First, we factor $\tterase_\Sigma :
  \transmonocl_{R,\Sigma} \to \transmonocl_{R,\varnothing}$ into two surjective
  morphisms $\transmonocl_{R,\Sigma} \to \transmonozo_{R,\Sigma} \to
  \transmonocl_{R,\varnothing}$, going through a new monoid which keeps some
  information about the letters of $\Sigma$ but is still \emph{finite}. To do
  so, we define an equivalence relation on register assignments as follows: for
  $\alpha,\beta \in \transmonocl_{R,\Sigma}$, we say that $\alpha \sim \beta$
  when
  \begin{itemize}
  \item $\tterase_\Sigma(\alpha) = \tterase_\Sigma(\beta)$;
  \item for each $r \in R$, the sets of letters \emph{from $\Sigma$} that appear
    in $\alpha(r)$ and $\beta(r)$ are equal.
  \end{itemize}
  One can show that $\sim$ is a congruence, so we may form the quotient monoid
  $\transmonozo_{R,\Sigma} = \transmonocl_{R,\Sigma}/\!\sim$. Thanks to the first
  clause in the definition of $\sim$, the morphism $\tterase_\Sigma$ factors
  through the canonical projection. The quotient is finite since each
  equivalence class has a representative $\alpha$ such that $|\alpha(r)| \leq
  |R|+|\Sigma|$ for all $r \in R$: essentially, $\sim$ only takes into account
  the presence or absence of each letter in $\Sigma$, not their multiplicity
  (hence the notation \enquote{01}).
  
  Next, let $f : \Gamma^* \to \Sigma^*$ be computed by some copyless SST $(Q,
  q_0, R, \delta, \vec{u}_I, F)$. We take $\prodmono(f) =
  \transmonozo_{R,\Sigma} \wr Q$ and define $\nu_f$ as a composition $\Gamma^*
  \to \transmonocl_{R,\Sigma} \wr Q \to \prodmono(f)$ where the first morphism
  -- which we may call $\psi_\delta$, as in \Cref{lem:shape-plus-labels} -- maps
  $c \in \Gamma$ to $\delta(-,c)$ and the second morphism is the canonical
  projection.

  What we need to show now is that, given a \emph{1-split} $(u,v,w) \in
  (\Gamma^*)^3$ with respect to $\nu_f$, the comparison between $|f(uvw)|$ and
  $|f(uw)|$ depends only on $\nu_f(x)$ for $x\in\{u,v,w\}$.
  
  Let $q'$ and $\vec{u}'_I$ be the state and register values of the SST after
  reading $u$; to be more formal, $\psi_\delta(u)(q_0) = (q', \alpha)$ and
  $\alpha^\dagger(\vec{u}_I) = \vec{u}'_I$. Note that $q'$ is also the first
  component of the pair $\nu_f(u)(q_0)$; since $\nu_f(uv) = \nu_f(u)$ (by
  definition of 1-split), the SST reaches the state $q'$ after reading $uv$ as
  well: $\psi_\delta(q',v) = (q',\beta)$ for some $\beta \in
  \transmonocl_{R,\Sigma}$.

  Let $\psi_\delta(w)(q') = (q'', \gamma)$. Then
  \[ \qquad f(uvw) = F(q'')^\dagger \circ
    (\beta\bullet\gamma)^\dagger(\vec{u}'_I) \qquad f(uw) = F(q'')^\dagger \circ
    \gamma^\dagger(\vec{u}'_I)\]
  Since $\nu_f(vw) = \nu_f(w)$, we have $\tterase_\Sigma(\beta\bullet\gamma) =
  \tterase_\Sigma(\gamma)$. Therefore, $\omega =
  (\beta\bullet\gamma)^\odot(F(q''))$ and $\omega' = \gamma^\odot(F(q''))$
  have the same letters from $R$ with the same multiplicities (and appearing in
  the same order, although this does not matter for us here): $|\omega|_r =
  |\omega'|_r$ for all $r \in R$. This is why the two sums over $R$ cancel out
  in the following computation (writing $\vec{u}'_I = (u'_r)_{r \in R}$):
  \begin{align*}
    \forall c \in \Sigma,\quad
    |f(uvw)|_c - |f(uw)|_c &= |\omega^\dagger(\vec{u}'_I)|_c -
                             |(\omega')^\dagger(\vec{u}'_I)|_c \\
                           &= |\omega|_c + \sum_{r \in R} |\omega|_r \cdot |u'_r|_c -
                             |\omega'|_c - \sum_{r \in R}|\omega'|_r \cdot |u'_r|_c\\
                           &= |\omega|_c - |\omega'|_c
  \end{align*}
  From now on, let $c \in \Sigma$. From the definition of $\beta^\odot$, we have
  \[ |\omega|_c = |\beta^\odot(\omega')|_c = |\omega'|_c + \sum_{r\in R}
    |\omega'|_r \cdot |\beta(r)|_c \]
  So we get the dichotomy of the lemma statement:
  \begin{itemize}
  \item if there exists some $r \in R$ such that $|\omega'|_r > 0$ and
    $|\beta(r)|_c > 0$, then $|f(uvw)| > |f(uw)|$;
  \item otherwise, $|f(uvw)| = |f(uw)|$.
  \end{itemize}
  For each $r \in R$, the condition $|\omega'|_r > 0$ can be checked from $q'$
  and $\nu_f(w)$; in turn, $q'$ depends only on $\nu_f(u)$. As for $|\beta(r)|_c
  > 0$, since it is a condition on the presence or not of a certain letter from
  $\Sigma$ in $\beta(r)$, without considering its precise multiplicity, it
  depends only on $\nu_f(v)$: this is the information that $\nu_f$ was designed
  to encode. This gives us the definition of the set of producing triples
  $P(f,c)$.

  There remains a final claim to prove in the lemma statement, concerning
  $|f(uv^n w)|_c$ when $(u,v,w) \in P(f,c)$ and $n \in \naturalN$. Using
  $\mu(uv) = \mu(u)$, one can show that for all $m \in \naturalN$, the triple
  $(uv^m,v,w)$ is also a producing 1-split. So we have
  \[ |f(uv^n w)|_c > |f(uv^{n-1}w)|_c > \dots > |f(uw)|_c \] and since all
  elements of this sequence are natural numbers, $|f(uv^n w)|_c \geq n$.
\end{proof}

\noindent
Next, we prove \Cref{prop:ramsey} before \Cref{lem:dichotomy}, following the
order of logical dependency.

\framsey*
\begin{proof}
  By the finite Ramsey theorem for pairs, there exists $R \in \naturalN$ such
  that every complete undirected graph with at least $R$ vertices whose edges
  are colored using $|M|$ colors contains a monochromatic clique with $r+3$
  vertices. We take $N=R-1$.

  Let $s = uvw \in \Gamma^*$ with $|v| \geq N$. Let us write $s[i\dots j]$ for
  the substring of $s$ between two positions $i,j \in \{0,\dots,|s|\}$. Those
  indices are considered as positions \emph{in-between} letters, so, for
  instance, $s = s[0\dots |s|]$, while $s[(i-1)\dots i]$ is the $i$-th letter of
  $s$; note also that $s[i\dots j] \cdot s[j \dots k] = s[i\dots k]$. In
  particular, we have $v = s[|u| \dots |uv|]$.

  Consider the following coloring of the complete graph over $V =
  \{|u|,\dots,|uv|\}$: the edge $(i,j) \in V^2$ with $i < j$ is given the color
  $\varphi(s[i\dots j])$. Since $|V| \geq N+1 = R$, there exists a monochromatic
  clique $\{i_0,\dots,i_{r+2}\} \subseteq V$ with $i_0 < \dots < i_{r+2}$.

  We now define $u' = s[0\dots i_1]$ and $w' = s[i_{r+1} \dots |s|]$, which
  ensures that $u$ is a prefix of $u'$ and $w$ is a suffix of $w'$ since $i_1$
  and $i_{r+1}$ are positions in $v$. For $m \in \{1,\dots,r\}$, we also take
  $v'_m = s[i_m \dots i_{m+1}] \neq \varepsilon$ (because $|v'_m| = i_{m+1} -
  i_m \geq 1$). Then $s = u' v'_1 \dots v'_r w'$, and
  \begin{align*}
    \varphi(u' v'_1 \dots v'_m) &= \varphi(s[0\dots i_{m+1}]) =
                               \varphi(s[0\dots{}i_0]) \varphi(s[i_0 \dots i_{m+1}])\\
                             &= \varphi(s[0\dots{}i_0]) \varphi(s[i_0 \dots i_1])
                               \quad \text{by monochromaticity}\\
                             &= \varphi(s[0\dots{}i_1]) = \varphi(u')
  \end{align*}
  and similarly, $\varphi(v'_m \dots v'_1 w') = \varphi(w')$. Thus, by
  definition, we have an $r$-split of $s$.
\end{proof}

\dichotomy*
\begin{proof}
  $L(f,\Pi,\varphi,r)$ can be recognized by a non-deterministic automaton that
  guesses an adequate $r$-split and computes
  $\varphi(u),\varphi(v_1),\dots,\varphi(v_r),\varphi(w)$. The hard part is
  showing that $|f(-)|_\Pi$ is bounded on the complement of this language.

  By the previous proposition, there exists some $N \in
  \naturalN$ such that any string $s \in \Gamma^*$ of length at least $N$ admits
  an $r$-split according to $\varphi$. Thanks to the existence of $\pi$, it is
  also an $r$-split according to $\nu_f$. So if this long string is in
  $\Gamma^* \setminus L(f,\Pi,\varphi,r)$, then it is of the form $s = u v_1 \dots
  v_r w$ where, for some $i \in \{1,\dots,r\}$, $(u v_1 \dots v_{i-1}, v_i,
  v_{i+1} \dots v_r w)$ is \emph{not} producing. Therefore, $|f(u v_1 \dots
  v_{i-1} v_{i+1} \dots v_r w)|_\Pi = |f(u v_1 \dots v_r w)|_\Pi$. The important
  part is that the argument in the left-hand side is strictly shorter (the
  definition of $r$-split contains $v_i \neq \varepsilon$). Furthermore, we
  claim that $s' = u v_1 \dots v_{i-1} v_{i+1} \dots v_r w \in \Gamma^*
  \setminus L(f,\Pi,\varphi,r)$. Once this is established, a strong induction on
  the length suffices to show that $|f(-)|_\Pi$ restricted to $\Gamma^*
  \setminus L(f,\Pi,\varphi,r)$ reaches its maximum at some string of length
  smaller than $N$, and thus to conclude the proof.

  It remains to show that $s' \notin L(f,\Pi,\varphi,r)$. If this were false, then
  by definition we would have a producing $r$-split $s' = u' v'_1 \dots v'_r
  w'$. Assuming this, we will lift this split to a producing $r$-split of $s$ in
  order to contradict $s \notin L(f,\Pi,\varphi,r)$. We give notations to the
  components of our non-producing triple: $\hat{u} = u v_1 \dots v_{i-1}$,
  $\hat{v} = v_i$, $\hat{w} = v_{i+1} \dots v_r w$.

  Suppose that for some $j \in \{1,\dots,r\}$ and $x \in \Gamma^*$, we have
  $\hat{u} = u' v'_1 \dots v'_{j-1} x$ and $|x| \leq |v'_j|$. Then there must
  exist a unique $y \in \Gamma^*$ such that $v'_j = xy$ and $\hat{w} = y
  v'_{j+1} \dots v'_r w'$. What we want to show now is that $(u', v'_1, \dots,
  v'_{j-1}, x\hat{v}y, v'_{j+1}, \dots, v'_r, w')$ is a producing $r$-split of
  $s$.
  \begin{itemize}
  \item First, the concatenation of this sequence of length $r+2$ is indeed
    equal to $\hat{u}\hat{v}\hat{w} = s$.
  \item Next, we have $\varphi(u' v'_1 \dots v_{j-1} (x\hat{v}y)) =
    \varphi(\hat{u}\hat{v}y) = \varphi(\hat{u}\hat{v})\varphi(y) =
    \varphi(\hat{u})\varphi(y)$ since $(\hat{u},\hat{v},\hat{w})$ is a 1-split
    of $s$, and $\varphi(\hat{u})\varphi(y) = \varphi(\hat{u}y) = \varphi(u'
    v'_1 \dots v'_{j})$. For $k \geq j$, by multiplying by $\varphi(v'_{j+1}
    \dots v'_k)$ on the right, we get $\varphi(u' v'_1 \dots v'_{j-1}
    (x\hat{v}y) v'_{j+1} \dots v'_k) = \varphi(u' v'_1 \dots v'_k)$. Similarly,
    for $k \leq j$, we have $\varphi(v'_k \dots v'_{j-1} (x\hat{v}y) v'_{j+1}
    \dots v'_r w') = \varphi(v'_k \dots v'_r w')$.
  \item Combining the above with the fact that $(u',v'_1,\dots,v'_r,w')$ is an
    $r$-split of $s'$ gives us directly from the definitions that our new
    $(r+2)$-tuple with $x\hat{v}y$ is an $r$-split of $s$.
  \item Finally, we must check that it is producing.
    \begin{itemize}
    \item Let $k \leq j-1$. We must show that $(u' v'_1 \dots v'_{k-1}, v_k,
      v'_{k+1} \dots v'_{j-1} (x\hat{v}y) v_{j+1} \dots v'_r w')$ is producing
      with respect to $(f,\Pi)$. We have seen previously that the componentwise
      image by $\varphi$ of this triple is the same as the one for $(u' v'_1
      \dots v'_{k-1}, v'_k, v'_{k+1} \dots v'_r w')$. The latter is producing
      (since it comes from an $r$-split chosen to be producing), and therefore
      so is the former, because thanks to $\nu_f = \pi \circ \varphi$, the
      image by $\varphi$ suffices to determine whether a triple is producing.
    \item The case $k \geq j+1$ is symmetrical.
    \item The remaining case is $(u' v'_1 \dots v'_{j-1}, x\hat{v}y, v'_{j+1}
      \dots v'_r w')$. It would be convenient if $x\hat{v}y$ and $v'_j$ had the
      same image by $\varphi$, but this is not guaranteed. Instead, we come back
      to the dichotomy concerning what happens when we remove the substring
      $x\hat{v}y$ in $s$. This can be done in two steps: first remove $\hat{v}$
      in $s$, which gives us $s'$, then remove $xy = v'_j$ from $s'$, resulting
      in $s'' = u' v'_1 \dots v'_{j-1} v'_{j+1} \dots v'_r w'$. Using the fact
      that the $r$-split of $s'$ is producing while $(\hat{u},\hat{v},\hat{w})$
      is not, we have $|f(s)|_\Pi = |f(s')|_\Pi > |f(s'')|_\Pi$. This means that
      the 1-split containing $x\hat{v}y$ must be producing.
    \end{itemize}
  \end{itemize}
  If the $(j,x)$ chosen previously does not exist, then either $u'$ is a prefix
  of $\hat{u}$ or $w'$ is a suffix of $\hat{w}$. In those cases, there is an
  analogous lifting procedure, and its proof of correctness is simpler; we leave
  this to the reader.
\end{proof}
\begin{remark}
  \label{rem:dichotomy}
  It is not clear whether the above reasoning can be made to work if we require
  idempotency in the definition of $r$-split.
  An analogous argument is made in the first paragraph of the proof of the original
  Dichotomy Lemma in~\cite{PebbleMinimization}, but we were unable to check that
  $s' \notin L(f,\Pi,\varphi,r)$ when forcing the central elements of producing
  triples to be idempotent. Thankfully it does not seem to be required to carry out
  further arguments leading to a proof of \Cref{thm:minimization}.
\end{remark}

\inlining*
\begin{proof}
  We write $f = \CbS(g,(h_i)_{i \in I})$.
  
  \proofsubparagraph{First, let us consider the case $k=0$.}
  For convenience, we assume
  w.l.o.g.\ that $I \cap \Sigma = \varnothing$. Let $N = \sup\{|g(s)|_J \mid s
  \in \Gamma^*\}$ -- in the degenerate case $J = \varnothing$, this leads to $N
  = 0$ -- and $\iota_n(s)$ be the $n$-th letter of $J$ in $g(s)$ if it exists,
  or else $\varepsilon$. Then we use the equation $g(s) =
  \rho_0(s)\iota_1(s)\rho_1(s) \dots \iota_N(s)\rho_N(s)$ to define uniquely
  $\rho_0,\dots,\rho_N : \Gamma^* \to (I \setminus J)^*$. One can build for each
  $n \in \{0,\dots,N\}$ a sequential transducer whose composition with $g$
  yields $\rho_n$; therefore, since $g$ is regular, so is $\rho_n$. We define
  $\psi_n(s)$ next as $h_{\iota_n(s)}(s)$ when $\iota_n(s) \in J$, and
  $\varepsilon$ otherwise. For any $n \in \{1,\dots,N\}$, since the languages
  $g^{-1}(((I \setminus J)^* J)^{n-1} (I\setminus J)^* i I^*)$ are regular for all $i \in
  J$, this defines $\psi_n$ as a combination of $\{s \mapsto \varepsilon\} \cup
  \{h_i \mid i \in J\}$ by regular conditionals, so $\psi_n$ is regular.
  Finally, we set $f'(s) = \rho_0(s)\psi_1(s)\rho_1(s) \dots \psi_N(s) \rho_N(s)
  \in (\Sigma \cup I \setminus J)^*$; the function $f'$ thus defined is regular
  by closure under concatenation (use a product construction on copyless SSTs).
  Observe that $f'(s)$ is obtained by substituting each occurrence of a letter
  $i\in J$ in $g(s)$ by $h_i(s)$ (thus, it is equal to $g(s)$ when $J =
  \varnothing$, and to $f(s)$ when $J = I$).

  What remains to do is to substitute the letters of $I \setminus J$ to get $f$.
  To do so, let us define $L_{\vec{w}} = \{ s \in \Gamma^* \mid \forall i \in I
  \setminus J,\; h_i(s) = w_i \}$ for $\vec{w}=(w_i)_{i\in I\setminus J} \in
  \prod_{i \in I\setminus J} h_i(\Gamma^*)$. The function $f$ coincides on
  $L_{\vec{w}}$ with $f'$ postcomposed with the morphism that replaces each $i
  \in I \setminus J$ by $w_i$; this is regular by closure under composition.
  Furthermore, the factors of $\prod_{i \in I\setminus J} h_i(\Gamma^*)$ are
  finite by definition of $J$, and $I \setminus J$ itself is a subset of the
  finite alphabet $I$. So there are finitely many $L_{\vec{w}}$, and they
  partition $\Sigma^*$; they are also all regular, as finite intersections of
  preimages of singletons by regular functions. Therefore, $f$ is obtained by
  combining regular functions by a regular conditional, so it is regular, i.e.\
  $\rk(f) = 0$ as we wanted.
  
  \proofsubparagraph{This being done, let us move on to the case $k \geq 1$.}
  For $i \in J$, let $h_i = \CbS(g'_i, (h'_{i,x})_{x \in X_i})$ where all the
  $g'_i$ are regular and the $h'_{i,x}$ are of rank at most $k-1$, choosing the
  $X_i$ to be pairwise disjoint as well as disjoint from $I$. Let $f'(s)$ be
  obtained from $g(s)$ by substituting each occurrence of a letter $i \in J$ by
  $g'_i(s)$. For the same reasons as those exposed in the first paragraph of the
  case $k=0$, this defines a regular function $f'$. By taking its composition by
  substitution with the disjoint union of the families $(h_i)_{i \in I \setminus
    J}$ and $(h'_{i,x})_{x \in X_i}$ for $i \in J$, we recover $f$. Since the
  functions involved in this union family are all of rank at most $k-1$ (by
  definition of $J$), this means that $\rk(f) \leq k$.
\end{proof}

\subsection{Wrapping up the proof of \Cref{thm:minimization}}

\begin{lemma}
  \label{lem:cf-triples}
  Let $f : \Gamma^* \to \Sigma^*$ be comparison-free polyregular. There exists a
  morphism to a \emph{finite} monoid $\nu'_{f} : \Gamma^* \to \prodmonocf(f)$ such that,
  for any \emph{1-split} according to $\nu'_f$ composed of $u,v,w \in \Gamma^*$
  and any $c \in \Sigma^*$, the sequence $(|f(uv^nw)|_c)_{n\in\naturalN}$ is
  non-decreasing.
\end{lemma}
\begin{proof}
  By straightforward induction on $\rk(f)$, using \Cref{lem:producing-triples}:
  for $f = \CbS(g,(h_i)_{i \in I})$ where $g$ is regular and the $h_i$ are
  comparison-free, we take $\prodmonocf(f) = \prodmono(g) \times
  \displaystyle\prod_{i \in I} \prodmonocf(h_i)$.
\end{proof}

\begin{lemma}
  \label{lem:dichotomy-redux}
  Let $f : \Gamma^* \to \Sigma^*$ be a comparison-free polyregular function. Let
  $\varphi : \Gamma^* \to M$ be a morphism to a finite monoid and let $r \geq 1$. Then there exists a regular language
  $\widehat{L}(f,\varphi,r) \subseteq \Gamma^*$ such that:
  \begin{itemize}
  \item the function which maps $\widehat{L}(f,\varphi,r)$ to $\varepsilon$ and
    coincides with $f$ on $\Gamma^* \setminus \widehat{L}(f,\varphi,r)$
    \begin{itemize}
    \item is regular and takes finitely many values if $\rk(f) = 0$ i.e.\ $f$ is
      regular;
    \item is comparison-free polyregular with rank strictly lower than $\rk(f)$
      otherwise;
    \end{itemize}
  \item for any $s \in \widehat{L}(f,\varphi,r)$, there exist $k = \rk(f)+1$
    $r$-splits according to $\varphi$ -- let us write them as $s = u^{(m)}
    v^{(m)}_1 \dots v^{(m)}_r w^{(m)}$ for $m \in \{1,\dots,k\}$ -- such that,
    for any factorization $s = \alpha_0 \beta_1 \alpha_1 \dots \beta_{k}
    \alpha_{k}$ where, for some \emph{permutation} $\sigma$ of $\{1,\dots,k\}$,
    each $\beta_m$ coincides with some $v^{(\sigma(m))}_l$ (in the sense that
    their positions as substrings of $s$ are equal), we have
    \[ \forall n \in \naturalN,\;
      |f(\alpha_0 \beta^n_1 \alpha_1 \dots \beta^n_{k} \alpha_k)| \geq n^k \]
    (note that in general, such factorizations $s = \alpha_0 \beta_1 \alpha_1
    \dots \beta_{k} \alpha_{k}$ might not exist, for instance when $r = 1$ and
    all the substrings $v^{(m)}_1$ overlap)
  \end{itemize}
\end{lemma}

\begin{proof}
  We proceed by induction on $\rk(f)$.

  \proofsubparagraph{Base case: $\rk(f)=0$.} In this case, $f$ is regular. Let
  $\psi : \Gamma^* \to M \times \prodmono(f)$ be the monoid morphism obtained by pairing
  $\varphi$ (given in the lemma statement) with $\nu_f$ (given in
  \Cref{lem:dichotomy}). Then, using \Cref{lem:dichotomy}, one can see
  that taking $\widehat{L}(f,\varphi,r) = L(f,\Sigma,\psi,r)$ works.

  \proofsubparagraph{Inductive case: $\rk(f) \geq 1$.} In this case, $f =
  \CbS(g, (h_i)_{i \in I})$ for some regular $g : \Gamma^* \to I^*$ and some
  comparison-free polyregular $h_i : \Gamma^* \to \Sigma^*$ with $\rk(h_i) \leq
  \rk(f) - 1$ for all $i \in I$. Let $\varphi$ and $r$ be as given in the lemma
  statement. Let $J$ be defined as in \Cref{lem:inlining}:
  \[ J = \begin{cases}
      \{i \in I \mid \rk(h_i) = \rk(f)-1 \} & \text{when}\ \rk(f) \geq 2\\
      \{i \in I \mid |h_i(\Gamma^*)| = \infty \} & \text{when}\ \rk(f) = 1
    \end{cases}\]
  For $i \in J$, let $\psi_i : \Gamma^* \to M \times \prodmono(g) \times \prodmonocf(h_i)$ be
  obtained by combining $\varphi$ with the morphisms given by
  Lemmas~\ref{lem:producing-triples} and~\ref{lem:cf-triples}. We shall consider
  the regular languages $\widehat{L}(h_i,\psi_i,r)$ provided by the inductive
  hypothesis.

  Let us take a copy $\underline{J} = \{\underline{i} \mid i \in J\}$ of $J$
  such that $\underline{J} \cap I = \varnothing$. We define the regular function
  $g' : \Gamma^* \to (I \cup \underline{J})^*$ as follows: for any input $s \in
  \Gamma^*$, to build the output $g'(s)$, we start from $g(s)$ and then, for
  each $i \in J$ \emph{such that $s \notin \widehat{L}(h_i,\psi_i,r)$}, we
  replace all the occurrences of $i$ by $\underline{i}$. For $i \in J$, we also
  define $h_{\underline{i}}$ to be the function which maps
  $\widehat{L}(h_i,\psi_i,r)$ to $\varepsilon$ and coincides with $h_i$ on
  $\Gamma^* \setminus \widehat{L}(h_i,\psi_i,r)$. By construction, $f = \CbS(g',
  (h_i)_{i\in{}I\cup\underline{J}})$.

  Note that $g'$ is regular: it is indeed generated from $g$ using regular
  conditionals and postcomposition by letter-to-letter morphisms. We can
  therefore build a morphism
  \[ \chi : \Gamma^* \to M \times \prodmono(g') \times \prod_{i \in J} \prodmonocf(h_i) \]
  in the expected way, and define the language provided by the lemma statement
  as
  \[ \widehat{L}(f,\varphi,r) = L(g',J,\chi,r) \]
  According to \Cref{lem:dichotomy}, it is indeed a regular language.
  Concerning the first item of the lemma statement, the function that it
  considers can be expressed as
  \[ \CbS(g'',(h_i)_{i\in{}I\cup\underline{J}}) \qquad\text{where}\qquad
    g'' : s \mapsto
    \begin{cases}
      \varepsilon & \text{when}\ s \in L(g',J,\chi,r)\\
      g'(s) & \text{otherwise}
    \end{cases}\]
  We want to show that $\rk(\CbS(g'',(h_i)_{i\in{}I\cup\underline{J}})) \leq
  \rk(f) - 1$. The shape of this statement fits with the conclusion of
  \Cref{lem:inlining}, so we just have to check the corresponding assumptions.
  \begin{itemize}
  \item $g''$ is regular, by closure of regular functions under regular
    conditionals.
  \item for $i \in I\cup\underline{J}$, the function $h_i$ is comparison-free
    polyregular of rank at most $\rk(f)-1$:
    \begin{itemize}
    \item for $i \in I$, this was required in our choice of expression for $f =
      \CbS(g,(h_i)_{i \in I})$ (and such a choice was possible by definition of
      rank);
    \item for $i = \underline{j} \in \underline{J}$, we get this by applying the
      first item of the inductive hypothesis to $h_j$ (indeed, the function
      introduced by this item is none other than $h_{\underline{j}} = h_i$).
    \end{itemize}
  \item We also get that, \emph{with the same $J$ as before,}
    \[ J = \begin{cases}
        \{i \in I \cup \underline{J} \mid \rk(h_i) = \rk(f)-1 \} & \text{when}\
        \rk(f) \geq 2\ \text{i.e.}\ \forall i \in J,\; \rk(h_i) \geq 1\\
        \{i \in I \cup \underline{J} \mid |h_i(\Gamma^*)| = \infty \} &
        \text{when}\ \rk(f) = 1\ \text{i.e.}\ \forall i \in J,\; \rk(h_i) = 0
      \end{cases}\]
    using again the first item of the inductive hypothesis to handle the case of
    indices in $\underline{J}$.
  \item Finally, by definition of $g''$ and by \Cref{lem:dichotomy}, using the
    convention $\sup\varnothing=0$,
    \[ \sup_{s \in \Gamma^*} |g''(s)|_J = \sup\{|g'(s)|_J \mid s \in \Gamma^*
      \setminus L(g',J,\chi,r) \} < \infty \]
  \end{itemize}

  Let us now check the second item concerning splits and factorizations. Let $s
  \in \widehat{L}(f,\varphi,r)$. By definition, there exists $i \in J$ such that
  $s \in L(g',i,\chi,r)$. In particular, $|g'(s)|_i \geq 1$, which entails that
  $s \in \widehat{L}(h_i,\psi_i,r)$ by definition of $g'$. The inductive
  hypothesis gives us a family of $r$-splits $s = u^{(m)} v^{(m)}_1 \dots
  v^{(m)}_r w^{(m)}$ according to $\psi$ for $m \in \{1,\dots,k-1\}$ -- recall
  that $\rk(h_i)+1 = \rk(f) = k-1$. We complete it by taking $(u^{(k)},
  v^{(k)}_1, \dots, v^{(k)}_r, w^{(k)})$ to be a producing $r$-split of $s$ with
  respect to $(g',i,\chi)$, whose existence is guaranteed by definition of
  $L(g',i,\chi,r)$. Since $\varphi$ factors through both $\psi_i$ and $\chi$ by
  construction, this indeed gives us a family of $k$ $r$-splits according to
  $\varphi$.
  
  Now, let $s = \alpha_0 \beta_1 \alpha_1 \dots \beta_{k} \alpha_{k}$ be a
  factorization and $\sigma$ be a permutation of $\{1,\dots,k\}$ such each
  $\beta_m$ coincides with some $v^{(\sigma(m))}_l$ for some $l$. Note that from
  the original expression of $f$ as a composition by substitution, we have
  \[ \forall s' \in \Gamma^*,\quad |f(s')| \geq |g(s')|_i \cdot |h_i(s')| \]
  Therefore, our desired inequality will follow once we prove the ones below:
  \[ \forall n \in \naturalN,\quad |g(\alpha_0 \beta^n_1 \alpha_1 \dots \beta^n_{k}
    \alpha_{k})|_i \geq n \quad\text{and}\quad |h_i(\alpha_0 \beta^n_1
    \alpha_1 \dots \beta^n_{k} \alpha_{k})| \geq n^{k-1} \]
  To illustrate the idea, we assume $\sigma(k) = k$, so that $\beta_k =
  v^{(k)}_l$ for some $l$, and we invite the reader to convince themself that
  this is merely a matter of notational convenience for the rest of the proof.

  Let us start with $h_i$. Since $\nu'_{h_i}$ factors through $\chi$, the triple
  \[ \qquad (\alpha_0 \beta_1 \dots \beta_{k-1} \alpha_{k-1}, \beta_k, \alpha_k)
    = (u^{(k)} v^{(k)}_1 \dots v^{(k)}_{l-1}, v^{(k)}_l, v^{(k)}_{l+1} \dots
    v^{(k)}_r w^{(k)})\]
  is a 1-split according to $\nu'_{h_i}$. Using the fact that $\nu'_{h_i}$
  factors through $\psi_i$, one can show that $(\alpha_0 \beta^n_1 \dots
  \beta^n_{k-1} \alpha_{k-1}, \beta_k, \alpha_k)$ is still a 1-split according
  to $\nu'_{h_i}$. Therefore, for $n \in \naturalN$,
  \[ |h_i(\alpha_0 \beta^n_1 \dots \beta^n_{k-1} \alpha_{k-1} \beta^n_{k}
    \alpha_{k})| \geq |h_i(\alpha_0 \beta^n_1 \dots \beta^n_{k-1} \alpha_{k-1}
    \beta_k \alpha_{k})| \geq n^{k-1}\]
  where $\beta_k$ is \emph{not} raised to the $n$-th power in the middle; the
  left inequality comes from \Cref{lem:cf-triples}, while the right inequality
  is part of the induction hypothesis applied to $h_i$.

  The case of $g$ requires an additional step. We know that $(\alpha_0
  \beta_1 \dots \beta_{k-1} \alpha_{k-1}, \beta_k, \alpha_k)$ is a producing
  triple with respect to $(g',i,\chi)$; therefore, by
  \Cref{lem:producing-triples},
  \[ \forall n \in \naturalN,\; |g'(\alpha_0 \beta_1 \dots \beta_{k-1}
    \alpha_{k-1} \beta^n_{k} \alpha_{k})|_i \geq n \]
  To replace $g'$ by $g$ in the above inequality, recall that by definition of
  $g'$, since $i \in J$,
  \[ \forall s' \in \Gamma^*,\; \left( |g'(s')|_i \neq 0 \implies |g'(s')|_i =
      |g(s')|_i \right) \]
  One can then conclude by \Cref{prop:monotone-multipumping} below, taking
  $l=k-1$. There is a subtlety
  here: our definitions ensure that $\nu_g$ factors through $\psi_i$, but this
  might not be the case for $\nu_{g'}$ (because $\psi_i$ had to be defined
  before $g'$). So for this final step, we must work with the function $g$,
  whereas to leverage the producing triple, we had to use $g'$.
\end{proof}

The following proposition, which we used at the end of the above proof, will
also be useful to prove \Cref{thm:polyreg-not-cf}.

\begin{proposition}
  \label{prop:monotone-multipumping}
  Let $g : \Gamma^* \to \Sigma^*$ be a regular function and $s =
  \alpha_0\beta_1\alpha_1 \dots \beta_l \alpha_l \in \Gamma^*$ such that every
  triple $(\alpha_0\beta_1 \dots \alpha_m, \beta_{m+1}, \alpha_{m+1}\beta_{m+2}
  \dots \alpha_l)$ is a 1-split according to $\nu_g$. Then for every $c \in
  \Sigma$, the function
  \[ (n_1,\dots,n_l) \mapsto |g(\alpha_0\beta_1^{n_1}\alpha_1 \dots
    \beta_l^{n_l} \alpha_l)|_c \]
  is monotone according to the product partial order on $\naturalN^l$.
\end{proposition}
\begin{proof}[Proof idea]
  In order to apply \Cref{lem:dichotomy}, the key observation is that the triple
  \[ (\alpha_0\beta^{n_1}_1 \dots \alpha_m \beta_{m+1}^{n_{m+1}}, \beta_{m+1},
    \alpha_{m+1}\beta_{m+2}^{n_{m+2}} \dots \alpha_l) \]
  is also a 1-split. This is because we have, by definition of 1-split,
  $\nu_g(\alpha_0\beta^{n_1}_1) = \nu_g(\alpha_0\beta_1)$, then
  $\nu_g(\alpha_0\beta_1\alpha_1\beta_2^{n_2}) =
  \nu_g(\alpha_0\beta_1\alpha_1\beta_2)$, etc., and similarly on the right side.
\end{proof}

After having established \Cref{lem:dichotomy-redux}, we can use it to finally
wrap up this section.

\begin{proof}[Proof of \Cref{thm:minimization}]
  We apply \Cref{lem:dichotomy-redux} to get a language
  $\widehat{L}(f,\varphi,\rk(f)+1)$ where $\varphi$ does not matter (take for
  instance the morphism from $\Gamma^*$ to the trivial monoid). It must be
  non-empty (or else we would have the contradiction $\rk(f) < \rk(f)$), so we
  can choose an arbitrary element $s \in \widehat{L}(f,\varphi,\rk(f)+1)$.
  
  Let $k = \rk(f)+1$. \Cref{lem:dichotomy-redux} gives us $k$ factorizations $s
  = u^{(m)} v^{(m)}_1 \dots v^{(m)}_k w^{(m)}$ satisfying certain properties.
  Note that $k$ plays two roles here that were distinct in the lemma. We claim
  that thanks to this, there exists a factorization $s = \alpha_0 \beta_1
  \alpha_1 \dots \beta_{k} \alpha_k$ as described in \Cref{lem:dichotomy-redux}.
  This entails that setting $s_n = \alpha_0 \beta^n_1 \alpha_1 \dots \beta^n_{k}
  \alpha_k$ proves the theorem.

  Our task is therefore to select one element in each of the $k$ sets
  $\{v^{(m)}_l \mid l \in \{1,\dots,k\}\}$ of substrings of $s$ for $m \in
  \{1,\dots,k\}$, such that the selected substrings are pairwise
  non-overlapping. There is a strategy for this which is similar to the
  classical greedy algorithm for computing a maximum independent set in an
  interval graph. We take $\beta_1$ to be the substring of $s$ among the
  $v^{(m)}_1$ whose right endpoint is leftmost. One can check that $\beta_1$
  cannot overlap with any $v^{(m)}_l$ for $l \geq 2$. Thus, by discarding the
  set to which $\beta_1$ belongs, as well as each $v^{(m)}_1$ in the other sets,
  we reduce the remainder of the task to our original goal with $k$ being
  decremented by 1. At this stage, an induction suffices to conclude the proof.
\end{proof}

\section{Proofs of Theorems~\ref{thm:composition-squaring}
  and~\ref{thm:characterization-rank}}

Now that we have shown that cfp functions are closed under composition and that
their asymptotic growth are tightly linked to their ranks, we have the essential
ingredients to prove Theorems~\ref{thm:composition-squaring}
and~\ref{thm:characterization-rank}.
There are a couple of preliminary lemmas helpful for both that we first prove here.

\begin{lemma}
\label{lem:app-cfp-cfpow}
  For any comparison-free polyregular function $f : \Gamma^* \to \Sigma^*$ and
  $k \geq \rk(f)$, there exists a regular function $f' : (\{0,\dots,\rk(f)\}
  \times \Gamma)^* \to \Sigma^*$ such that $f =
  f'\circ\cfpow^{(k+1)}_\Gamma$.
\end{lemma}
\begin{proof}
  By induction on $\rk(f)$ (with an inductive hypothesis that quantifies over
  $k$).

  \proofsubparagraph{Base case ($f$ regular).} 

  Consider the unique
  $\varphi\in\Hom((\{0,\dots,\rk(f)\}\times\Gamma)^*,\Gamma^*)$ such that for
  every $c\in\Gamma$, $\varphi(k,c)=c$ and $\varphi(m,c)=\varepsilon$ when $m <
  k$. Since regular functions are closed under composition, $f'=f\circ\varphi$
  is regular, and the desired equation follows from the fact that
  $\varphi\circ\cfpow^{(k+1)}_\Gamma = \mathrm{id}_{\Gamma^*}$.

  \proofsubparagraph{Inductive case.} 

  Let $f = \CbS(g,(h_i)_{i \in I})$ with $g : \Gamma^* \to I^*$ regular and $h_i
  : \Gamma^* \to \Sigma^*$ cfp such that $\rk(h_i) \leq \rk(f)-1$ for all $i \in
  I$.
  Using the inductive hypothesis, we know that
  $h_i = h'_i \circ \cfpow^{(k)}_\Gamma$ for some family of regular functions
  $(h'_i)_{i \in I}$. Thus, let us assume we are given 2DFTs $\cT$ and
  $(\cT_i')_{i \in I}$ corresponding to $g$ and the family $(h'_i)_{i \in I}$.
  Without loss of generality, let us assume further that $\cT$ always output
  at most one letter at each transition, never outputs a letter upon reading $\rightmarker$,
  and that the $\cT_i$ always terminate on the marker $\leftmarker$ by a
  transition that does not move the reading head.
  With these assumptions, let us describe informally a 2DFT $\cT''$
  corresponding to the function $f'$ such that
  $\CbS(g,(h'_i \circ \cfpow^{(k)}_\Gamma)_{i \in I}) =
   f' \circ \CbS(g,(h'_i)_{i \in I}) \circ \cfpow^{(k+1)}_\Gamma$.

  Assuming that the state space of $\cT$ is $Q$ and the state space of $\cT'_i$
  is $Q'_i$, with $Q$ and the $Q'_i$s all pairwise disjoint, we take the state space
  of $\cT''$ to be
  \[ Q'' ~~=~~ Q \times \{{\sf L},{\sf R},{\sf S}\} \times \left(\{ \bullet\} \sqcup
               \bigcup_{i \in I} Q_i' \times \{{\sf L}, {\sf R}, {\sf S}\}\right) \]
  with initial state $(q_0,{\sf R}, \bullet)$, if $q_0$ is the initial state of $\cT$
  and final states the triples $(q_f,{\sf M}, \bullet)$ such that $q_f$ is a final
  state of $\cT$. To guide intuitions, the elements ${\sf L}$, ${\sf R}$ and ${\sf S}$
  should be respectively read as \enquote{left}, \enquote{right} and \enquote{stay}.
  With this in mind, the high-level description of computations carried out
  by $\cT''$ over words $\mathtt{pow}^{(k)}_\Gamma$ is as follows.
  \begin{itemize}
  \item When in a state $(q_0,{\sf M}, \bullet)$, $\cT''$ essentially
        acts as $\cT$ on letters of the shape $(k,a)$ or end-markers and
        ignores letters $(l,a)$ for $l < k$; the central component ${\sf M}$ then
        determines whether to seek the next relevant position to the left or to
        the right when reading such an irrelevant letter. This continues up until
        upon reading a letter $(k,a)$ in state $(q, {\sf M},\bullet)$ such that
        $\cT$ would ouput $i$ when reading $a$ in $q$, $\cT''$ moves into the state
        $(r, {\sf M'}, (q'_{0,i}, {\sf R}))$ where $q'_{0,i}$ is the initial
        state of $\cT_i'$ and $(r,{\sf M'})$ is determined by the transition in $\cT$.
  \item When in a state $(q,{\sf M}, (q_i', {\sf M'}))$ for $q_i' \in Q_i'$,
        $\cT$ behaves exactly as $\cT_i$ as long as the current transition does not reach the
        final state, treating letters outside of its input alphabets
        as end markers; this is possible because of the component ${\sf M'}$ of the state,
        that we use to keep track of the last move of the reading head.
        Meanwhile the components $q \in Q$ and ${\sf M}$ are untouched.
        When a final transition is taken, by our assumption we return control to $\cT$ by
        going to state $(q,{\sf M},\bullet)$ and moving in the direction prescribed by ${\sf M}$
        (recall that, by assumption, we are moving away from (or staying in) the position at
        which $\cT_i$ started running).
   \end{itemize}
   We leave formalizing this definition and checking that $\cT''$ has a desirable behaviour
   to the reader.
\end{proof}

\begin{lemma}
\label{lem:app-cfpow-cfsq}
  For every $k \in \naturalN$, $\cfpow^{(k)}_{\Gamma}$ is equal to a
  a composition of sequential functions and squaring functions $\cfsquaring_\Delta$.
\end{lemma}
\begin{proof}
  We proceed by induction over $k$. The cases of $k = 0, 1, 2$ are immediate
  as $\cfpow^{(k)}_\Gamma$ then corresponds, up to isomorphism of output alphabet,
  to a constant function, the identity and $\cfsquaring_{\Gamma}$ respectively, so we
  focus on the inductive step. To achieve the desired result, it suffices to show
  that there exists a sequential function
  \[f : ((\{0, \ldots, k\} \times \Gamma) \cup \underline{(\{0, \ldots, k\} \times \Gamma)})^* \to
       (\{0, \ldots, k+1\} \times \Gamma)^*\]
  such that $\cfpow_{\Gamma}^{(k+1)} =
  f \circ
  \cfsquaring_{\{0, \ldots, k\} \times \Gamma}
  \circ \cfpow_{\Gamma}^{(k+1)}$.
  In fact, the sequential transducer pictured below computes such an $f$:
  \begin{center}
  {
    \begin{tikzpicture}[scale=1.25]
    \node (init) at (1,2) {};
    \node (outa) at (1,-0.5) {};
    \node (outb) at (6,-0.5) {};
    \node[circle,draw,thick,minimum size=20pt] (qa) at (1,1) {$o$};
    \node[circle,draw,thick,minimum size=20pt] (qb) at (6,1) {$i$};
    
    \draw[->] (init) -- (qa);
    \draw[->] (qa) edge [loop left] node
     {\footnotesize $\begin{array}{c@{\;}c@{\;}c}
                      \underline{(k,a)}&|&(k+1,a) \\
                                  (k,a)&|&(k,a)   \\
                                  (m,a)&|&(m,a)
                      \end{array}$} ();
    \draw[->] (qa) edge  [bend left] node[above] {\footnotesize $\underline{(m,a)}\; |\; \varepsilon$} (qb);
    \draw[->] (qb) edge [loop right] node
     {\footnotesize $\begin{array}{c@{\;}c@{\;}c}
                      \underline{(m,a)}&|&\varepsilon \\
                                  (k,a)&|&\varepsilon   \\
                                  (m,a)&|&\varepsilon
                      \end{array}$} ();
    \draw[->] (qb) edge [bend left] node[below] {\footnotesize $\underline{(k,a)}\;|\;(k+1,a)$} (qa);
    \draw[->] (qa) edge node[left] {$\varepsilon$} (outa);
    \draw[->] (qb) edge node[right] {$\varepsilon$}(outb);
  \end{tikzpicture}
  }
\end{center}
where $m$ designates any element of $\{0, \ldots, k-1\}$.
\end{proof}

Now we turn to the proofs of our main theorems.

\compositionsquaring*

\begin{proof}[Proof of \Cref{thm:composition-squaring}]
  The direct implication is obtained by combining the two lemmas above:
  every cfp function can be written as a composition $f \circ \cfpow^{(k)}_\Gamma$
  for some $k \in \naturalN$ and $f$ regular by \Cref{lem:app-cfp-cfpow},
  and \Cref{lem:app-cfpow-cfsq} guarantees that in turn, $\cfpow^{(k)}_\Gamma$
  is a composition of sequential (and a fortiori regular) functions and squarings.
  Conversely, that cfp functions are closed under composition is proven 
  in \Cref{sec:proof-composition}, which is enough to conclude as regular
  functions and $\cfsquaring_\Gamma$ are cfp.
\end{proof}

\characterizationrank*

\begin{proof}[Proof of \Cref{thm:characterization-rank}]
  We prove the circle of implications (\ref{thm-item:rank-def}) $\Rightarrow$
  (\ref{thm-item:rank-pow}) $\Rightarrow$ (\ref{thm-item:rank-growth})
  $\Rightarrow$ (\ref{thm-item:rank-def}). (The claim after this equivalence has
  already been established in \Cref{thm:minimization}.)

  The first implication (\ref{thm-item:rank-def}) $\Rightarrow$ (\ref{thm-item:rank-pow}) corresponds exactly to \Cref{lem:app-cfp-cfpow} we just proved.

  The implication (\ref{thm-item:rank-pow}) $\Rightarrow$ (\ref{thm-item:rank-growth}) is also relatively easy: $\cfpow^{(k)}_\Gamma$
  is cfp (this is a consequence of \Cref{lem:app-cfpow-cfsq} and \Cref{thm:composition-squaring},
  although $\cfpow^{(k)}_\Gamma$ can also be shown to fit
  \Cref{def:cfpolyreg-cbs} in a more elementary way) and
  so is $f \circ \cfpow^{(k+1)}_\Gamma$ by \Cref{thm:composition-squaring} for $f$ regular.
  Furthermore, $|\cfpow^{(k+1)}_\Gamma(w)| = O(|w|^{k+1})$ and, since $f$ is regular,
  $|f(u)| = O(|u|)$, so we have, as expected, $|(f \circ \cfpow^{(k+1)}_\Gamma)(w)| = O(|w|^{k+1})$.

  The final implication (\ref{thm-item:rank-growth}) $\Rightarrow$ (\ref{thm-item:rank-def}) is technically the hardest as it relies on
  \Cref{thm:minimization}.
  Let $f : \Gamma^* \to \Sigma^*$ be cfp and $k\in\naturalN$ such that
  $|f(w)|=O(|w|^{k+1})$. If $f$ is regular, then $\rk(f)=0 \leq k$. Otherwise,
  by \Cref{thm:minimization}, there exists a sequence $(w_n)_{n\in \naturalN}$
  of inputs such that $|w_n|=O(n)$ and $|f(w_n)| \geq n^{\rk(f)+1}$. So
  $n^{\rk(f)+1} = O(n^{k+1})$, hence $\rk(f)\leq k$.
\end{proof}

\section{Comparison-free polyregular sequences}

\subsection{Proof of \Cref{thm:cfp-seq}}

\cfpseq*

As announced, we prove \Cref{thm:cfp-seq} inductively on the rank of the sequence under consideration.
The bulk of the reasoning is concentrated in the base case, stating that regular
sequences are exactly the ultimately periodic combinations of pumping sequences.

\begin{lemma}
\label{lem:reg-seq}
A sequence of words $s : \naturalN \to \Sigma^*$ is regular if and only if there is
$m > 0$ such that for every $k < m$, there are words $u_0, \ldots, v_l, v_1, \ldots, v_l$
such that for every $n \in \naturalN$, we have
\[\forall n \in \naturalN, \; s((n+1)m + k) = u_0(v_1)^n\ldots (v_l)^nu_l\]
\end{lemma}

\begin{proof}
The \enquote{if} direction is straightforward, so we only prove the
\enquote{only if} part of the statement.
To keep notations harmonized, let us work with $f : \{a\}^* \to \Sigma^*$
such that $s(n) = f(a^n)$ for every $n \in \naturalN$ and fix
a copyless SST computing $f : \{a\}^* \to \Sigma^*$ whose set of
states, set of registers and transition function we call $Q$, $R$ and $\delta$
respectively. We use the monoid $\transmonocl_{R,\varnothing} \wr Q$ introduced
in \Cref{sec:transition-monoids}, which contains $\mu =
\tterase_\Sigma(\delta(-,a))$. Since $\transmonocl_{R,\varnothing} \wr Q$
is finite (\Cref{prop:transmonocl-finite}), there is an exponent $m \in
\naturalN \setminus \{0\}$ such that $\mu^{\bullet m} = \mu \bullet \ldots
(\text{$m$ times}) \ldots \bullet \mu$ is idempotent, i.e.\ $\mu^{\bullet m} =
\mu^{\bullet 2m}$. This $m$ is the one put forth in the lemma statement.

Let us fix $k < m$. Let $(q,\alpha) = \mu^{\bullet m}(q_0)$ where $q_0$ is
the initial state of the SST. We have $\mu^{\bullet (m + k)}
\bullet \mu^{\bullet m} = \mu^{\bullet (2m + k)} = \mu^{\bullet k} \bullet
\mu^{\bullet 2m} = \mu^{\bullet k}\bullet\mu^{\bullet m} = \mu^{\bullet (m + k)}$
as usual. Therefore, $\mu^{\bullet m}(q) = (q,\beta)$ with
$\alpha\bullet\beta=\alpha$ and $\beta\bullet\beta=\beta$ (the latter is because
of $\mu^{\bullet 2m} = \mu^{\bullet m}$). Thus, $q$ is the state reached by the
SST after reading $a^{m(n+1)+k}$ for any $n \in \naturalN$. We also have
$(\delta(-,a))^{\bullet m}(q) = (q,\gamma)$ with $\gamma \in
\transmonocl_{R,\Sigma}$ and $\tterase_\Sigma(\gamma) = \beta$.

Given $r \in R$, we distinguish two cases.
\begin{itemize}
\item First, suppose that $\beta(r) = \varepsilon$ or equivalently that
  $\gamma(r) \in \Sigma^*$ (in general, the codomain of $\gamma$ is $(\Sigma
  \cup R)^*$). When the SST is in state $q$ and reads $a^m$, it executes the
  assignment $\gamma$; when $\beta(r) = \varepsilon$, the new value of the
  register $r$ is this $\gamma(r) \in \Sigma^*$ which does not depend on the old
  value of any register. Therefore, for all $n \in \naturalN$, the content of
  the register $r$ after having read $a^{m(n+1)+k}$ (starting from the initial
  configuration) is the constant $\gamma(r)$.
\item We now treat the case where $\beta(r)$ is non-empty. By definition,
  $\beta\bullet\beta = \beta^* \circ \beta$ where $\beta^* \in \Hom(R^*,R^*)$
  extends $\beta : R \to R^*$. Since we know, as a consequence of the
  idempotency of $\mu^{\bullet m}$, that $\beta\bullet\beta = \beta$,
  we have $\beta^*(\beta(r)) = \beta(r) \neq \varepsilon$.

  Let us study in general the situation $\beta^*(\rho) = \beta(r) \neq
  \varepsilon$ for $\rho \in R^*$. A first observation is that the letters in
  $\beta(r)$ cannot be found in any other $\beta(r')$ for $r' \in R \setminus
  \{r\}$ because $\beta$ is copyless, so $\rho \notin (R \setminus \{r\})^*$. We
  therefore have $n \geq 1$ occurrences of $r$ in $\rho$, so $\rho = \rho_0 r
  \dots r \rho_n$ with $\rho_0,\dots,\rho_n \notin (R \setminus \{r\})^*$. By
  coming back to $\beta^*(\rho) = \beta(r)$, into which we plug this expression
  for $\rho$, and using the fact that $\beta(r)$ has non-zero length, we can see
  that $n=1$ and $\beta^*(\rho_0) = \beta^*(\rho_1) = \varepsilon$.

  Let us apply this to $\rho = \beta(r) = \tterase_\Sigma(\gamma)(r)$ and
  lift the result to $\gamma(r)$:
  \[ \gamma(r) = u_r r v_r \quad\text{for some}\quad u_r, v_r \in (\Sigma \cup
    \beta^{-1}(\{\varepsilon\}))^*\]
  In the previous case ($\beta(r') = \varepsilon$ for $r' \in R$), we saw that
  $\gamma(\beta^{-1}(\{\varepsilon\})) \subseteq \Sigma^*$. Therefore
  $\gamma^\odot(u_r), \gamma^\odot(v_r) \in \Sigma^*$, where $\gamma^\odot \in
  \Hom((\Sigma \cup R)^*, (\Sigma \cup R)^*)$ extends $\gamma : R \to (\Sigma
  \cup R)^*$ by being the identity on $\Sigma$. Since $\Sigma^*$ is fixed by
  $\gamma^\odot$, when we iterate, we obtain
  \[ \gamma^{\bullet(n+1)}(r) = (\gamma^\odot)^n \circ \gamma(r) =
    (\gamma^\odot(u_r))^n \cdot u_r r v_r \cdot (\gamma^\odot(v_r))^n  \]
  
\end{itemize}
Now, let $F$ be the final output function of the SST that computes $f$, and
$\vec{w}_{m+k}$ be the register values after it has read a prefix $a^{m+k}$. Then
after reading $a^{m(n+1)+k}$, the new register values are
$(\gamma^{\bullet(n+1)})^\dagger (\vec{w}_{m+k})$. More precisely, the register
$r$ contains:
\begin{itemize}
\item $\gamma(r) \in \Sigma^*$ if $\beta(r) = \varepsilon$;
\item $(\gamma^\odot(u_r))^n \cdot ((u_r r v_r)^\dagger(\vec{w}_{m+k})) \cdot
  (\gamma^\odot(v_r))^n$ otherwise.
\end{itemize}
These values are combined by $F(q)^\dagger$ -- where $q$ is the recurrent state
we have been working with all along, and $F$ is the final output function -- to
produce the output $f(a^{m(n+1)+k})$. This yields the desired shape: an
interleaved concatenation of finitely many factors that are either constant,
$(\gamma^\odot(u_r))^n$ or $(\gamma^\odot(v_r))^n$ for some $r \in R$.
\end{proof}

\begin{proof}[Proof of \Cref{thm:cfp-seq}]
We proceed by induction on the rank of the sequence $s : \naturalN \to \Sigma^*$ under consideration.
If the rank of $s$ is $0$, it is regular and we apply \Cref{lem:reg-seq} and the desired polynomial word
expression is of the shape $u_0 \cdot (v_1)^* \ldots (v_l)^* \cdot u_l$.

If the rank of $s$ is $k+1$, thanks to the induction hypothesis and the base case above,
it can be written as $\CbS(\interp{e}, (\interp{e'_i})_{i \in I})$ where $e$
is an expression over the alphabet $I$ with star-height at most one and
the $e'_i$s expressions over $\Sigma$ with star-height at most $k$.
Without loss of generality, we may assume that that terminal nodes
of polynomial word expressions are words of length at most one.
For such an expression over alphabet $I$, one may define inductively
the following \emph{substitution operation} to obtain an expression
of $\Sigma^*$:
\[
\begin{array}{l@{\;}c@{\;}l!\qquad l@{\;}c@{\;}l}
j[(e'_i)_{i\in I}] &=& e'_j & \varepsilon[(e'_i)_{i\in I}] &=& \varepsilon \\
(f \cdot f')[(e'_i)_{i\in I}] &=&
f[(e'_i)_{i\in I}]
\cdot 
f'[(e'_i)_{i\in I}]
&
f^*[(e'_i)_{i\in I}] &=&
(f[(e'_i)_{i\in I}])^*
\end{array}
\]
One can then check by induction on the structure of $e$ that
$\interp{e[(e'_i)_{i \in I}]} = \CbS(\interp{e}, (\interp{e'_i})_{i \in I})$
and that $e[(e'_i)_{i \in I}]$ has star-height bounded by $k+1$.
\end{proof}

\subsection{Proof of \Cref{cor:cfp-map-seq}}

We finally show that cfp sequences are closed by post-composing
with functions $\fmap(f)$ for $f$ cfp.

\cfpmapseq*

We first prove the result
for poly-pumping sequences.

\begin{lemma}
\label{lem:map-pumpseq}
If $\interp{e} : \naturalN \to (\Gamma \cup \{\#\})^*$ is a poly-pumping sequence
and $f : \Gamma^* \to \Sigma^*$ is comparison-free polyregular, then $\fmap(f) \circ \interp{e}$ is a cfp sequence.
\end{lemma}

For the rest of this subsection, we write $\Succ$ for the successor function $n \mapsto n+1$ over $\naturalN$.
We will use the fact that $s$ is a cfp sequence iff $s \circ \Succ$ also is.
\begin{proof}
We first note that if the separator $\#$ does not occur at any leaf of $e$, then
the result is immediate as we would have $\fmap(f) \circ \interp{e} = f \circ \interp{e}$.
We thus focus on the cases when it does occur, and proceed inductively over $e$.
\begin{itemize}
\item If $e = w \in (\Gamma \cup \{\#\})^*$, then $\fmap(f) \circ \interp{e}$ is
a constant sequence, which is obviously cfp.
\item If $e = (e')^*$, with $\#$ occuring in $e'$, let $h_{\mathrm{l}},h_{\mathrm{r}} : \naturalN \to
  \Gamma^*$ and $h_{\mathrm{c}} : \naturalN \to (\Gamma\sqcup\{\#\})^*$ be the sequences
  such that
  \[\interp{e'} \circ \Succ \circ \Succ = h_{\mathrm{l}} \cdot \# \cdot
  h_{\mathrm{c}} \cdot \# \cdot h_{\mathrm{r}}\]
  with $h_{\mathrm{l}}(n)$ being the largest $\#$-free prefix of
  $\interp{e'}(n+2)$ and $h_{\mathrm{r}}(n)$ the largest $\#$-free suffix of
  $\interp{e'}(n+2)$.
There is a regular function
\begin{align*}
  f' : \quad\qquad (\Gamma \sqcup \{\#\})^* \qquad
  &~~\to~~ (\Gamma \sqcup \{\#\})^*\\
  w_0\#w_1\#\ldots w_{n-1}\#w_n
  &~~\mapsto~~ w_1\# \ldots w_{n-1} \quad
    (w_0,\ldots,w_n \in (\Gamma \sqcup \underline{\Gamma})^*)
\end{align*}
stripping away the first and last component of its input, so that it satisfies
\[f' \circ \fmap(f) \circ \interp{e'} \circ \Succ \circ \Succ ~~ = ~~ \fmap(f) \circ h_{\mathrm{c}}\]
By the inductive hypothesis, we know that $\fmap(f) \circ \interp{e'}$ is
comparison-free polyregular.
We may therefore conclude by composition (cf.\ \Cref{thm:composition-squaring})
that $\fmap(f) \circ h_{\mathrm{c}}$ is cfp. One can check
analogously that $h_{\mathrm{l}}$ and $h_{\mathrm{r}}$ are also cfp.
Then observe that
\[
\begin{array}{llcl}
&(\interp{e} \circ \Succ \circ \Succ)(n) & =& (h_{\mathrm{l}} \cdot \# \cdot h_{\mathrm{c}} \cdot \# \cdot h_{\mathrm{r}})(n)^{n+2} \\
& &=& (h_{\mathrm{l}} \cdot (\# \cdot h_{\mathrm{c}} \cdot \# \cdot h_{\mathrm{r}} \cdot h_{\mathrm{l}})^{n+1} \cdot \# \cdot h_{\mathrm{c}} \cdot \# \cdot h_{\mathrm{r}})(n) \\
\end{array}
\]
which means that we have
\[
\fmap(f) \circ \interp{e} \circ \Succ \circ \Succ =
{\small \left\{ \begin{array}{c}
(f \circ h_{\mathrm{l}}) \\ \cdot \\
(\# \cdot (\fmap(f) \circ h_{\mathrm{c}}) \cdot \# \cdot (f \circ (h_{\mathrm{r}} \cdot h_{\mathrm{l}})))^{*} \\ \cdot \\
\# \cdot (\fmap(f) \circ h_{\mathrm{c}}) \cdot \# \cdot (f \circ (h_{\mathrm{r}} \cdot h_{\mathrm{l}})) \\ \cdot \\
\# \cdot (\fmap(f) \circ h_{\mathrm{c}})  \cdot \# \cdot (f \circ h_{\mathrm{r}})
\end{array} \right.}
\]
Thanks again to the closure under composition, each component of this expression is cfp, so $\fmap(f) \circ \interp{e} \circ \Succ \circ \Succ$ is also cfp.
Hence, so is $\fmap(f) \circ \interp{e}$.
\item The last case where $e = e' \cdot e''$ is handled similarly after a case analysis determining whether $\#$ occurs only in $e'$, $e''$ or in both; we leave it to the reader. \qedhere
\end{itemize}
\end{proof}
\begin{proof}[Proof of \Cref{cor:cfp-map-seq}]
Suppose we are given $f : \Gamma^* \to \Sigma^*$ and
$s : \naturalN \to (\Gamma \cup \{\#\})^*$ cfp.
By \Cref{thm:cfp-seq}, $s$ is an ultimately periodic combination
of poly-pumping sequences, so that there are $m > 0$ and some expressions
$e_0, \ldots, e_{m-1}$ such that $s(m(n+1) + k) = \interp{e_k}(n)$ for every $k < m$.
By \Cref{lem:map-pumpseq}, every $\fmap(f) \circ \interp{e_k}$ is
cfp. The set $L_k = \{ m(n+1) + k \mid n \in \naturalN\}$
is semi-linear, i.e., corresponds to a regular language, and there are regular sequences
$r_k: \naturalN \to \naturalN$ such that $r_k(m(n+1)+k) = n$.
Further, $\naturalN = \{ n \mid n < m\} \cup \bigcup_{k < m} L_k$, so we may use the regular
conditional provided by \Cref{prop:cfp-conditionals} to show that the combination
of the $\fmap(f) \circ \interp{e_k} \circ r_k$ and the first $m$ values of $\fmap(f) \circ s$,
which corresponds exactly to $\fmap(f) \circ s$, is indeed cfp.
\end{proof}

\section{Separation results}

\subsection{Proof of \Cref{thm:cf-not-hdt0l}}

\cfnothdtOl*

\subparagraph{These examples are comparison-free.}

We have seen in \Cref{ex:squaring} that $w \mapsto w^{|w|}$ is a comparison-free
polyregular function. For the other examples:
\begin{itemize}
\item $(a^n \mapsto (a^n b)^{n+1}) = \CbS((a^n \mapsto a^{n+1}),\, (a^n \mapsto
  a^n b)_{i \in \{a\}})$ is obtained as a composition by substitution of
  \emph{sequential functions}, i.e.\ functions computed by sequential
  transducers (cf.\ \Cref{sec:reg}), which are in particular regular;
\item for an alphabet $\Sigma$ with $a,\# \in \Sigma$, there exist sequential
  functions $f : \Sigma^* \to \{a\}^*$ and $g : \Sigma^* \to \Sigma^*$ such that
  $f(a^n\#w) = a^n$ and $g(a^n\#w) = w\#$ for $n \in \naturalN$ and $w \in
  \Sigma^*$, so that $\CbS(f,(g)_{i \in \{a\}})(a^n\#w) = (w\#)^n$.
\end{itemize}

\subparagraph{(i) is not HDT0L.}

Let us fix a HDT0L system $(\{a\},\{a,b\},\Delta,d,(h)_{i\in\{a\}},h')$ and
show that it does not compute $a^n \mapsto (a^n b)^{n+1}$. Let
$\mathtt{letters}(w)$ be the set of letters occurring in the string $w$ at least
once. By the infinite pigeonhole principle, there exists an infinite $X
\subseteq \naturalN$ such that $\mathtt{letters}(h^n(d))$ has
the same value $\Delta'$ for all $n \in X$. Let us do a case analysis:
\begin{itemize}
\item Suppose first that for some $r \in \Delta'$ and some $m \in \naturalN$,
  the letter $b$ appears twice in $h' \circ h^m(r)$; in other words,
  that the latter contains a factor $b a^k b$ for some $k \in \naturalN$.
  Then for all $n \in X$, $h' \circ h^{m+n}(d) \in \Sigma^* b a^k b
  \Sigma^*$. Since $X$ is infinite, this holds for some $n$ such that $m+n >
  k$, so that this word -- i.e.\ the output of the HDT0L system for $a^{m+n}$
  -- is different from $(a^{m+n}b)^{m+n+1} \notin \Sigma^* b a^k b \Sigma^*$.
\item Otherwise, for all $r \in \Delta'$ (that includes the degenerate case
  $\Delta' = \varnothing$) and all $m \in \naturalN$, there is at most one
  occurrence of $b$ in $h' \circ h^m(r)$. Then for all $m \in
  \naturalN$, the length of $h^{\min(X)}(d)$ bounds the number of
  occurrences of $b$ in $h' \circ h^{m + \min(X)}(d)$, and this bound
  is independent of $m$. On the contrary, in the sequence $((a^{n}b)^{n+1})_{n
    \geq m + \min(X)}$, the number of occurrences of $b$ is unbounded.
\end{itemize}

\subparagraph{(ii) is not HDT0L.}

The second counterexample, namely $w \mapsto w^{|w|}$, reduces to the first one:
indeed, $(a^{n}b)^{n+1} = (a^{n}b)^{|a^{n}b|}$ for all $n \in \naturalN$, which
can also be expressed as
\[ (w \mapsto w^{|w|}) \circ (u \in \{a\}^* \mapsto ub) = (a^n \mapsto
  (a^nb)^{n+1}) \]
Suppose for the sake of contradiction that there is a HDT0L system
$(\Sigma,\Sigma,\Delta,d,(h_c)_{c\in\Sigma},h')$ that computes $w \mapsto
w^{|w|}$ with $|\Sigma| \geq 2$; we may assume without loss of generality that
$a,b \in \Sigma$. Then $(\{a\},\{a,b\},\Delta,h_b(d),(h_a)_{c\in\{a\}},h')$
computes $a^n \mapsto (a^{n}b)^{n+1}$.

\subparagraph{(iii) is not HDT0L.}

(This is claimed without proof in \cite[Section~6]{Marble}.)

Let $\Sigma \supseteq \{a,\#\}$ be an alphabet and let
$(\Sigma,\Sigma,\Delta,d,(h_c)_{c\in\Sigma},h')$ be a HDT0L system. We reuse
a similar argument to our treatment of the counterexample \textbf{(i)}.
Let the sets $\Delta' \subseteq
\Delta$ and $X \subseteq \naturalN$ with $X$ infinite be such that
$\mathtt{letters}(h_a^n(d)) = \Delta'$ for all $n \in X$.
\begin{itemize}
\item Suppose first that for some $r \in \Delta'$ and some $m \in \naturalN$,
  the string $h' \circ h_a^m \circ h_\# (r)$; contains a factor $\# \cdot a^k
  \cdot \#$ for some $k \in \naturalN$. Then for all $n \in X$, the given HDT0L
  system maps $a^m\#a^{n}$ to a string in $\Sigma^* \cdot \# \cdot a^k \cdot \#
  \cdot \Sigma^*$. For $n > k$, this language does not contain $(a^{n}\#)^m$;
  such a $n \in X$ exists because $X$ is infinite.
\item Otherwise, for any $m \in \naturalN$, since $\#$ occurs at most once in
  $h' \circ h_a^m \circ h_\# (r)$ for $r \in \Delta'$, the output of the HDT0L
  system has at most $|h^{\min(X)}(d)|$ occurrences of $\#$ on input
  $a^m\#a^{\min(X)}$. Therefore, for large enough $m$, this output is different
  from $(a^{\min(X)}\#)^m$.
\end{itemize}

\subsection{Proof of \Cref{thm:polyreg-not-cf}}

Let us recall the theorem.

\polyregnotcf*

(i) and (ii) are proven separately.

\subsubsection{Proof of \Cref{thm:polyreg-not-cf} item (i)}

As mentioned in the body of the paper, this is proven by showing
that the lengths of blocks of $baa\ldots aab$, or equivalently, maximal
blocks of $a$ in the output of a given regular sequence is determined
by a finite number of polynomial expressions.
Let us formalize this notion.

\begin{definition}
\label{def:poly-uniformity}
Let $\Sigma$ be a finite alphabet and $c \in \Sigma$.
Call $\beta_c : \Sigma^* \to \powerset(\naturalN)$ the
function assigning to a word $w$ the set of lengths of its maximal factors
lying in $\{c\}^*$ (including $\varepsilon$):
 \[\beta_c(w) ~~ = ~~ \{ k \in \naturalN \;\mid\; w \in (\Sigma^* \setminus
   (\Sigma^* \cdot c)) \cdot c^k \cdot (\Sigma^* \setminus (c \cdot \Sigma^*))
   \}\]
 We say that a sequence $s : \naturalN \to \Sigma^*$ is \emph{poly-uniform} if for
 every $c \in \Sigma$ there exists a \emph{finite} set of polynomials $A_{s,c}
 \subseteq \bbQ[X]$ such that, for every $n \in \naturalN$,
 \[\beta_c(s(n)) ~\subseteq ~A_{s,c}(n) = \{ P(n) \mid P \in A_{s,c} \}\]
\end{definition}

\begin{lemma}
\label{lem:comparisonfree-polyuniform}
Every comparison-free polyregular sequence $f : \naturalN \to \Sigma^*$ is
poly-uniform.
\end{lemma}
\begin{proof}
First, observe that any ultimately periodic combination of poly-uniform sequence is poly-uniform.
Indeed, assume that we have such a sequence $s$ and $m > 0$
so that $n \mapsto s(m(n+1) + k)$ is poly-uniform for every $k$, and finite sets $A_{k,c} \subseteq \bbQ[X]$
so that $\beta_c(s(m(n+1)+k)) \subseteq A_{k,c}(n)$. Then we have
\[ A_{s,c} ~~=~~ \bigcup_{l < m} \left\{ P\left({\small\frac{X-l}{m}}\right) \mid P \in A_{k,c} \right\} \cup \beta_c(s(l))\]
witnessing that $s$ is poly-uniform.

Hence, by \Cref{thm:cfp-seq}, it suffices to show that poly-pumping sequences
are all poly-uniform. We proceed by induction over polynomial word expressions $e$,
defining suitable finite sets of polynomials $A_{e,c}$
for $c \in \Sigma$ such that $\beta_c(\interp{e}(n)) \subseteq A_{e,c}(n)$ and $0 \in A_{e,c}$:
\[
\begin{array}{lcl!\qquad lcl}
A_{e \cdot e',c} &~=~& \{ P + Q ~ \mid ~ (P,Q) \in A_{e,c} \times A_{e',c}\}
&
A_{w,c} &~=~& \beta_c(w) ~\cup~ \{0\}
\\
A_{e^*,c} &~=~& A_{e,c} ~\cup~ \{ XP ~ \mid ~ P \in A_{e,c}\}
\end{array}
\]
\end{proof}

We can now conclude the proof of the first item of \Cref{thm:polyreg-not-cf}
by observing that the function $f : a^n \mapsto ba^{n-1}b\ldots bab$ does \emph{not}
correspond to a poly-uniform sequence: $\beta_c(f(a^n)) = \{0, \ldots, n-1\}$
is unbounded, and thus cannot be covered by a finite set of functions, let alone polynomials
in $\bbQ[X]$.

\subsection{Proof of \Cref{thm:polyreg-not-cf} item (ii)}

Suppose for the sake of contradiction that $f = \fmap(a^n\mapsto a^{n\times n})$
is comparison-free. Using \Cref{thm:characterization-rank}, it must then have
rank 1 since $|f(w)| = O(|w|^2)$. Thus, we may write $f=\CbS(g,(h_i)_{i\in I})$
where $g : \{a,\#\}^* \to I^*$ and all the $h_i : \{a,\#\}^* \to \{a,\#\}^*$ are
regular.

For each $J \subseteq I$ and $k \in \{0,\dots,|I|\}$ (though the definition
would make sense for $k \in \naturalN$), let $\rho_{J,k} : \{a^*\} \to (I
\setminus J)^*$ be uniquely defined by the condition
\[ \qquad \forall w \in \{a,\#\}^*,\; \rho_{J,k}(w) =
  \begin{cases}
    s & \text{when}\quad g(w) \in ((I \setminus J)^* J)^k \cdot s
    \cdot (\{\varepsilon\} \cup JI^*) \\
    \varepsilon & \text{when}\quad |g(w)|_J < k
  \end{cases}
\]
(recall from \Cref{sec:appendix-minimization} the notation $|\cdot|_J$). To put
it plainly, $\rho_{J,k}(w)$ is the $k$-th block of letters from $I \setminus J$
that appears in $g(w)$ (the block may be the empty string if there are
consecutive letters from $J$), or the empty string if this $k$-th block does not
exist. The function $\rho_{J,k}$ is regular because it is the composition of a
sequential function with $g$. 

We reuse some tools from \Cref{sec:appendix-minimization}, especially the notion
of producing 1-split from \Cref{lem:dichotomy}. There is a unique sensible way
to combine the morphisms $\nu_{f'} : \{a,\#\}^* \to \prodmono(f')$ given by this
lemma into a morphism
\[ \varphi : \{a,\#\}^* \to \prod_{f' \in \mathcal{F}} \prodmono(f')
\quad\text{for}\quad
\mathcal{F} = \{g\} \cup \{h_i \mid i \in I\} \cup \{\rho_{J,k} \mid J
\subseteq I,\, k \in \{0,\dots,|I|\}\}\]
Note that the codomain above is a finite monoid: this allows us to apply
\Cref{prop:ramsey} to this morphism $\varphi$ and $r=1$, which gives us some
$N\in\naturalN$. Let $s = a^N (\#a^N)^{|I|}$. For each $m \in \{0,\dots,|I|\}$,
we apply the proposition to the factorization $s = u_k v_k w_k $ with $u_k =
(a^N\#)^k$, $v_k = a^N$ and $w_k = (\#a^N)^{|I|-k}$ to get a 1-split $s = u'_k
v'_k w'_k$ according to $\varphi$ where $u_k$ is a prefix of $u'_k$ and $w_k$ is
a suffix of $w'_k$. Let $p_k = |v_k| \neq 0$ and $q_k = N - |v_k|$.

For $f' \in \mathcal{F}$ (the finite set of functions introduced above), we then
define
\[ \qquad \widetilde{f'} : (n_0,\dots,n_{|I|}) \in \naturalN^{|I|+1} \mapsto
  f'\left(a^{n_0p_0+q_0}\#\dots\#a^{n_{|I|}p_{|I|}+q_{|I|}}\right) \]

Thanks \Cref{prop:monotone-multipumping} and to the 1-split conditions that we
made sure to get previously, we see that for each letter $c$ in the codomain of
$f'$ (either $\{a,\#\}$ or $I$), $|\widetilde{f'}|_c : \naturalN^{|I|+1} \to
\naturalN$ is monotone for the product partial order. Since $f=\fmap(a \mapsto
a^{n \times n})=\CbS(g,(h_i)_{i\in I})$,
\[ \forall x \in \naturalN^{|I|+1},\;\sum_{i \in I} \len{\widetilde{g}(x)}_i \cdot
  \len{\widetilde{h}_i(x)}_\# = \len{\widetilde{f}(x)}_\# = |I| \]
\begin{align*}
  \text{where}\qquad \widetilde{f} : (n_0,\dots,n_{|I|}) \in \naturalN^{|I|+1}
  &\mapsto f\left(a^{n_0p_0+q_0}\#\dots\#a^{n_{|I|}p_{|I|}+q_{|I|}}\right)\\
  &= a^{(n_0p_0+q_0)^2}\#\dots\#a^{(n_{|I|}p_{|I|}+q_{|I|})^2}
\end{align*}
Since $|\widetilde{g}|_i$ and $|\widetilde{h}_i|_\#$ are monotone for all
$i\in{}I$, and $\naturalN^{|I|+1}$ admits a minimum $(0,\dots,0)$, the fact that
the above sum is constant means that, for each $i \in I$,
\begin{itemize}
\item either one of $|\widetilde{g}|_i$ and $|\widetilde{h}_i|_\#$ is
  constant equal to 0,
\item or both are non-zero constant.
\end{itemize}
Let $J^\# \subseteq I$ be the set of indices that fit the second case. We claim
that for $i \in J^\#$, the constant value taken by $|\widetilde{h}_i|_\#$ must
be 1. If this were not the case, then for all $n \in \naturalN$, there would be
a substring of the form $\#a\dots a\#$ in $|\widetilde{h}_i(n,\dots,n)|_\#$, and
since $f=\CbS(g,(h_i)_{i\in I})$ and $|\widetilde{g}(n,\dots,n)|_i \neq 0$, it
would also be a substring of $\widetilde{f}(n,\dots,n)$ with length at most
$|\widetilde{h}_i(n,\dots,n)|_\# = O(n)$ (since $h_i$ is regular). This is
impossible: for $k \in \{0,\dots,|I|\}$, the $k$-th substring of this form in
$\widetilde{f}(n,\dots,n)$ has length $(np_k+q_k)^2+2 = \Theta(n^2)$.

Combining this with the above equation for $|\widetilde{f}|_\#$, we see that
$|\widetilde{g}|_{J^\#}$ is the constant function equal to $|I|$. Let us
abbreviate $\rho_k = \rho_{J^\#,k} \in \mathcal{F}$ (recall that we defined it
at the beginning of this proof) for $k \in \{0,\dots,|I|\}$; then
\[ \forall x \in \naturalN^{|I|+1},\; \exists! \iota_1(x),\dots,\iota_{|I|}(x)
  \in J^\# : \widetilde{g}(x) = \widetilde{\rho}_0(x) \iota_1(x)
  \widetilde{\rho}_1(x) \dots \iota_{|I|}(x) \widetilde{\rho}_{|I|}(x) \]
Using this, we define $h'_k(x) = \widetilde{h}_{\iota_k(x)}(x)$ for $x \in
\naturalN^{|I|+1}$ and $k \in \{1,\dots,|I|\}$, plus two edge cases $h'_0 : x
\mapsto \varepsilon$ and $h'_{|I|+1} : x \mapsto \varepsilon$.

Let $k \in \{0,\dots,|I|\}$. Write $\vec{e}_k = (0,\dots,0,1,0,\dots,0) \in
\naturalN^{|I|+1}$ for the $k$-th vector of the canonical basis of
$\bbQ^{|I|+1}$. By looking again at the $k$-th substring of the form $\#a\dots
a\#$ in $\widetilde{f}(x)$, with $x = n\vec{e}_k$ here, we get
\[ \forall n \in \naturalN,\; (np_k+q_k)^2+2 \leq |h'_k(n\vec{e}_k)| +
  \sum_{i\in I} \len{\widetilde\rho_k(n\vec{e}_k)}_i \cdot
  \len{\widetilde{h}_i(n\vec{e}_k)} + |h'_{k+1}(n\vec{e}_k)| \]
Note that all the lengths involved in the right-hand side above are linearly
bounded in $n$ because of the regularity of the functions involved. So there
must exist $i_k \in I$ such that both $|\widetilde\rho_k(n\vec{e}_k)|_{i_k}$ and
$|\widetilde{h}_{i_k}(n\vec{e}_k)|$ are unbounded: otherwise, the whole RHS
would be $O(n)$, contradicting the $\Omega(n^2)$ lower bound induced by the
above inequality.

We thus get a finite sequence of indices $i_0, \dots, i_{|I|} \in I$. By the
pigeonhole principle, there must exist $k, l \in \{0,\dots,|I|\}$ such that $k
\neq l$ and $i_k = i_l$; we call $i$ this common value. Let $m \in \naturalN$ be
such that $\len{\widetilde\rho_k(m\vec{e}_k)}_i \geq 1$. By monotonicity
(since $\rho_k, h_i \in \mathcal{F}$):
\begin{itemize}
\item $\len{\widetilde\rho_k(m\vec{e}_k + n\vec{e}_l)}_i \geq 1$ for all
  $n\in\naturalN$;
\item $\len{\widetilde{h}_i(m\vec{e}_k + n\vec{e}_l)}$ is unbounded when
  $n\to+\infty$.
\end{itemize}
The product of those two quantities is a lower bound for the length of the
$k$-th substring of the form $\#a\dots a\#$ in
$\widetilde{f}(m\vec{e}_k+n\vec{e}_l)$, which contradicts the fact that this
length does not depend on $n$ (it is equal to $(mp_k+q_k)^2 + 2$).

\end{document}